\documentclass[a4paper,11pt]{article}
\usepackage[a4paper, margin=1in]{geometry}
\usepackage[utf8]{inputenc}
\usepackage[T1]{fontenc}
\usepackage{graphicx}
\usepackage{amsmath, amssymb, amstext, mathtools}
\usepackage{amsthm}
\usepackage{enumitem}
\usepackage{tikz}
\usepackage{pgfplots}
\usepackage{hyperref}
\usepackage[titlenumbered, linesnumbered, ruled]{algorithm2e}

\usetikzlibrary{backgrounds,patterns,arrows,automata,positioning,decorations,shapes,calc}

\definecolor[named]{urlblue}{cmyk}{1,0.58,0,0.21}

\hypersetup{
 breaklinks=true,
 colorlinks=true,
 citecolor=purple!70!blue!60!black,
 linkcolor=purple!70!blue!90!black,
 urlcolor=urlblue,
 pdflang={en},
 pdftitle={Hypergraph Isomorphism for Groups with Restricted Composition Factors},
 pdfauthor={Daniel Neuen}
}

\newtheorem{theorem}{Theorem}[section]
\newtheorem{lemma}[theorem]{Lemma}
\newtheorem{corollary}[theorem]{Corollary}
\newtheorem{example}[theorem]{Example}
\newtheorem{observation}[theorem]{Observation}
\theoremstyle{definition}
\newtheorem{definition}[theorem]{Definition}
\theoremstyle{remark}
\newtheorem{remark}[theorem]{Remark}
\newtheorem{claim}[theorem]{Claim}

\newenvironment{claimproof}{\begin{proof}}{\end{proof}}

\tikzstyle{normalvertex}=[circle,fill=white,draw=black]

\definecolor{darkpastelgreen}{rgb}{0.01, 0.75, 0.24}
\definecolor{royalpurple}{rgb}{0.47, 0.32, 0.66}
\definecolor{lime}{rgb}{0.75, 1.0, 0.0}
\definecolor{salmon}{rgb}{1.0, 0.55, 0.41}
\definecolor{orchid}{rgb}{0.85, 0.44, 0.84}
\definecolor{royalblue}{rgb}{0.0, 0.14, 0.4}
\definecolor{aquamarine}{rgb}{0.5, 1.0, 0.83}
\definecolor{wildstrawberry}{rgb}{1.0, 0.26, 0.64}

\definecolor{dandelion}{rgb}{0.94, 0.88, 0.19}
\definecolor{purple}{rgb}{0.63, 0.36, 0.94}
\definecolor{darkspringgreen}{rgb}{0.09, 0.45, 0.27}

\tikzstyle{color1}=[fill=red!80]
\tikzstyle{color3}=[fill=darkpastelgreen]
\tikzstyle{color6}=[fill=cyan!80]
\tikzstyle{color4a}=[fill=salmon]

\tikzstyle{lightcolor1}=[fill=red!30]
\tikzstyle{lightcolor2}=[fill=blue!30]
\tikzstyle{lightcolor3}=[fill=darkpastelgreen!30]
\tikzstyle{lightcolor4}=[fill=orange!30]
\tikzstyle{lightcolor5}=[fill=violet!30]
\tikzstyle{lightcolor6}=[fill=cyan!30]
\tikzstyle{lightcolor7}=[fill=yellow!30]
\tikzstyle{lightcolor8}=[fill=pink!40]
\tikzstyle{lightcolor9}=[fill=royalpurple!40]
\tikzstyle{lightcolor10}=[fill=lime!40]
\tikzstyle{lightcolor11}=[fill=salmon!40]
\tikzstyle{lightcolor12}=[fill=orchid!40]
\tikzstyle{lightcolor13}=[fill=royalblue!40]
\tikzstyle{lightcolor14}=[fill=aquamarine!40]
\tikzstyle{lightcolor15}=[fill=wildstrawberry!40]

\tikzstyle{emptyvertex}=[draw,circle,minimum size=7pt,inner sep=0pt]
\tikzstyle{tinyvertex}=[draw,circle,minimum size=3pt,inner sep=0pt]
\tikzstyle{thickedge}=[draw,gray!60,line width=1.6pt,-]

\allowdisplaybreaks

\DeclareMathOperator{\id}{id}

\DeclareMathOperator{\Iso}{Iso}
\DeclareMathOperator{\Aut}{Aut}
\DeclareMathOperator{\Sym}{Sym}
\DeclareMathOperator{\Alt}{Alt}

\DeclareMathOperator{\Aff}{Aff}

\DeclareMathOperator{\im}{im}

\DeclareMathOperator{\dist}{dist}

\DeclareMathOperator{\mgamma}{\widehat{\Gamma}}
\DeclareMathOperator{\funcnorm}{f_{\sf norm}}

\DeclareMathOperator{\Br}{Br}
\DeclareMathOperator{\Unf}{Unf}

\DeclareMathOperator{\Str}{Str}

\newcommand{\ColRef}[1]{\chi_{\sf CR}^{#1}}
\newcommand{\ColRefIt}[2]{\chi_{(#2)}^{#1}}
\newcommand{\tColRef}[2]{\chi_{#1{\text -}{\sf CR}}^{#2}}

\newcommand{\CO}{\mathcal O}
\newcommand{\CE}{\mathcal E}
\newcommand{\CH}{\mathcal H}
\newcommand{\CT}{\mathcal T}
\newcommand{\CS}{\mathcal S}
\newcommand{\CI}{\mathcal I}
\newcommand{\CJ}{\mathcal J}
\newcommand{\CU}{\mathcal U}
\newcommand{\CR}{\mathcal R}
\newcommand{\CC}{\mathcal C}

\newcommand{\Fx}{\mathfrak x}
\newcommand{\Fy}{\mathfrak y}
\newcommand{\Fz}{\mathfrak z}

\newcommand{\FA}{\mathfrak A}
\newcommand{\FB}{\mathfrak B}
\newcommand{\FC}{\mathfrak C}
\newcommand{\FG}{\mathfrak G}
\newcommand{\FP}{\mathfrak P}
\newcommand{\FX}{\mathfrak X}
\newcommand{\FY}{\mathfrak Y}
\newcommand{\FZ}{\mathfrak Z}

\newcommand{\NN}{\mathbb N}

\newcommand{\angles}[1]{\left\langle#1\right\rangle}

\newcommand{\case}[1]{\par\medskip\noindent\textit{Case #1: }}

\newenvironment{cs}{
  \begin{description}
    \renewcommand{\case}[1]{\item[{\itshape\mdseries Case ##1:}]}
  }{
  \end{description}
}

\newcommand{\orcid}[1]{\href{https://orcid.org/#1}{\includegraphics[height=1.8ex]{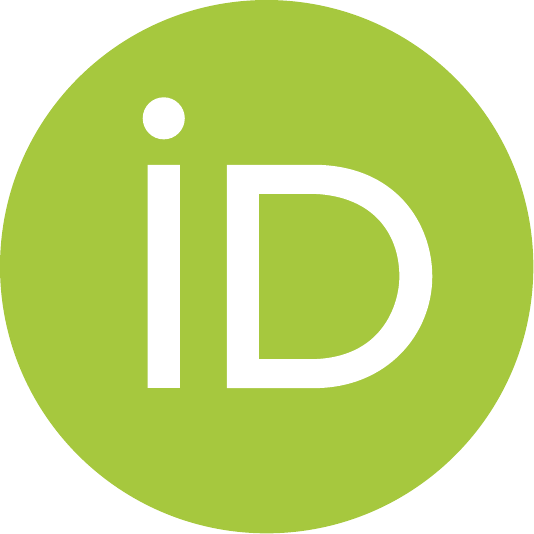}}}
\newcommand{\email}[1]{\href{mailto:#1}{\texttt{#1}}}

\title{Hypergraph Isomorphism for Groups\\with Restricted Composition Factors}
\author{
Daniel Neuen \orcid{0000-0002-4940-0318}\\
Simon Fraser University\\
\email{dneuen@sfu.ca}
}
\date{}

\begin{document}

\maketitle

\begin{abstract}
 We consider the isomorphism problem for hypergraphs taking as input two hypergraphs over the same set of vertices $V$ and a permutation group $\Gamma$ over domain $V$, and asking whether there is a permutation $\gamma \in \Gamma$ that proves the two hypergraphs to be isomorphic.
 We show that for input groups, all of whose composition factors are isomorphic to a subgroup of the symmetric group on $d$ points, this problem can be solved in time $(n+m)^{\CO((\log d)^{c})}$ for some absolute constant $c$ where $n$ denotes the number of vertices and $m$ the number of hyperedges.
 In particular, this gives the currently fastest isomorphism test for hypergraphs in general.
 The previous best algorithm for this problem due to Schweitzer and Wiebking (STOC 2019) runs in time $n^{\CO(d)}m^{\CO(1)}$.
 
 As an application of this result, we obtain, for example, an algorithm testing isomorphism of graphs excluding $K_{3,h}$ ($h \geq 3$) as a minor in time $n^{\CO((\log h)^{c})}$.
 In particular, this gives an isomorphism test for graphs of Euler genus at most $g$ running in time $n^{\CO((\log g)^{c})}$.
\end{abstract}

\section{Introduction}

Determining the computational complexity of the Graph Isomorphism Problem (GI) is a long-standing open question in theoretical computer science (see, e.g., \cite{Karp72}).
The problem is easily seen to be contained in NP, but it is neither known to be in PTIME nor known to be NP-complete.
In a breakthrough result, Babai \cite{Babai16} obtained a quasipolynomial-time algorithm for testing isomorphism of graphs (i.e., an algorithm running in time $n^{\CO((\log n)^{c})}$ where $n$ denotes the number of vertices of the input graphs and $c$ is a fixed constant), achieving the first improvement over the previous best algorithm running in time $n^{\CO(\sqrt{n / \log n})}$ \cite{BabaiKL83} in over three decades.

A natural generalization of GI is the isomorphism problem for hypergraphs.
It is easy to see that both problems are equivalent under polynomial-time reductions since every hypergraph $\CH$ can be translated into its bipartite incidence graph $B_\CH$.
Together with Babai's quasipolynomial-time isomorphism test for graphs, this gives an isomorphism test for hypergraphs running time $(n + m)^{\CO((\log(n + m))^{c})}$ where $n$ denotes the number of vertices and $m$ the number of edges of the input hypergraphs.
However, with $m$ appearing in the exponent, this running time seems far from optimal for large numbers of hyperedges (e.g., $m = 2^{\Omega(n^\varepsilon)}$ for some $\varepsilon > 0$).
Indeed, with the number of hyperedges potentially being much larger than the number of vertices, it is desirable to limit the dependence of $m$ in the running time of any hypergraph isomorphism test.
One of the first results in this direction is an algorithm testing isomorphism of hypergraphs in time $2^{\CO(n)}$ \cite{Luks99} where $n$ denotes the number of vertices.
Assuming the hyperedges are not too large, this result was improved by Babai and Codenotti in \cite{BabaiC08} to a running time of $n^{\widetilde{\CO}(k^{2} \cdot \sqrt{n})}$ where $\widetilde{\CO}(\cdot)$ hides polylogarithmic factors and $k$ denotes the maximal size of a hyperedge.
Again assuming small hyperedges, another improvement follows from the work of Grohe et al.\ \cite{GroheNS18} resulting in an isomorphism test for hypergraphs running in time $n^{\CO(k \cdot (\log n)^{c})}$ for a constant $c$.

All of the above algorithms rely on group-theoretic methods to decide isomorphism.
Towards this end, algorithms typically consider a more general problem which is also the core problem of this work.
We define the Hypergraph Isomorphism Problem to take as input two hypergraphs $\CH_1 = (V,\CE_1)$ and $\CH_2 = (V,\CE_2)$ over the same set of vertices and a permutation group $\Gamma \leq \Sym(V)$, and the task is to decide whether there is some $\gamma \in \Gamma$ that transforms $\CH_1$ into $\CH_2$.
Note that testing isomorphism of hypergraphs is a special case of this problem where $\Gamma = \Sym(V)$ is the full symmetric group.
On the other hand, by putting additional restrictions on the input group $\Gamma$, it is often possible to obtain improved algorithms.
A very common restriction, that goes back to Luks's celebrated polynomial-time isomorphism test for bounded degree graphs \cite{Luks82}, is to consider groups from the class $\mgamma_d$, the class of groups all of whose composition factors are isomorphic to a subgroup of $S_d$ (the symmetric group on $d$ points).
Miller \cite{Miller83b} showed that the Hypergraph Isomorphism Problem for $\mgamma_d$-groups can be solved in time $(m + n)^{\CO(d)}$.
This result was recently improved to $n^{\CO(d)}m^{\CO(1)}$ by Schweitzer and Wiebking \cite{SchweitzerW19}.
On the other hand, if we restrict our attention to graphs, by a recent algorithm of Grohe et al.\ \cite{GroheNS18}, this problem can be solved in time $n^{\CO((\log d)^{c})}$.
This naturally raises the question whether similar results are achievable for hypergraphs.

\paragraph{Results.}

The main result of this paper is a new algorithm for the Hypergraph Isomorphism Problem for $\mgamma_d$-groups that generalizes the algorithm by Grohe et al.\ \cite{GroheNS18} which only works for graphs.

\begin{theorem}
 \label{thm:main}
 The Hypergraph Isomorphism Problem for $\mgamma_d$-groups can be solved in time $(n+m)^{\CO((\log d)^{c})}$ for some absolute constant $c$ where $n$ denotes the number of vertices and $m$ the number of edges of the input hypergraphs.
\end{theorem}

Since every permutation group acting on $n$ vertices is contained in the class $\mgamma_n$, an immediate consequence of this result is the fastest algorithm for testing isomorphism of hypergraphs in general (assuming the number of hyperedges is moderately exponentially bounded in the number of vertices, i.e., $m = 2^{\CO(n^{1 - \varepsilon})}$ for some $\varepsilon > 0$).

\begin{corollary}
 \label{cor:hypergraph-isomorphism-intro}
 The Hypergraph Isomorphism Problem can be solved in time $(n+m)^{\CO((\log n)^{c})}$ for some absolute constant $c$ where $n$ denotes the number of vertices and $m$ the number of edges of the input hypergraphs.
\end{corollary}

In particular, compared to the previous results listed above, we remove the dependence on $k$, the maximal size of a hyperedge, in the running time.
Observe that this improvement is significant if $k$ is large and $m$ is small compared to $\binom{n}{k}$ (which is typical for many applications).

For the algorithm, we follow a similar route than \cite{GroheNS18}.
The first step is to invoke a normalization procedure.
Let $\Gamma \leq \Sym(\Omega)$.
Consider a sequence of $\Gamma$-invariant partitions $\{\Omega\} = \FB_0 \succ \dots \succ \FB_\ell = \{\{\alpha\} \mid \alpha \in \Omega\}$ where $\FB_{i} \prec \FB_{i-1}$ means that $\FB_{i}$ refines $\FB_{i-1}$, $i \in [\ell]$. 
Such a sequence is called \emph{almost $d$-ary} if, for all $i \in [\ell]$ and $B \in \FB_{i-1}$, the induced action of $\Gamma$, after stabilizing $B$ setwise, on the partition classes of $\FB_i$ contained in $B$ is a subgroup of the symmetric group $S_d$ or semi-regular (i.e., only the identity element has fixed points).
We say that $\Gamma$ is \emph{$d$-normalized} if it has an almost $d$-ay sequence of partitions.
The normalization procedure takes an instance $(\Gamma,\CH_1,\CH_2)$ of the Hypergraph Isomorphism Problem, and returns an equivalent instance $(\Gamma^*,\CH_1^*,\CH_2^*)$ such that $\Gamma^*$ is $d$-normalized.
Indeed, the algorithm of Grohe et al.\ performs the same normalization and we can easily adapt the normalization procedures from \cite{GroheNS18} to our setting.

Given a $d$-normalized input group, Grohe et al.\ prove a variant of Babai's Unaffected Stabilizers Theorem \cite{Babai16} and building on that, extend the Local Certificates Routine to their setting to obtain the crucial subroutine of the isomorphism test.
Our algorithm builds on the same variant of the Unaffected Stabilizers Theorem, but for hypergraphs, extending the Local Certificates Routine proves to be significantly more challenging and requires additional ideas.
Intuitively speaking, the reason is that the Local Certificates Routine crucially exploits that positions in disjoint sets can be permuted independently of each.
More precisely, let $W \subseteq \Omega$ be a $\Gamma$-invariant set and suppose $\Gamma$ respects the input structures restricted to the window $W$.
In this situation the standard Local Certificates Routine computes automorphisms of the input structures by considering the group $\Gamma_{(\Omega \setminus W)}$ which fixes all points outside of $W$.
In the case of strings, which are the combinatorial objects under consideration in \cite{Babai16,GroheNS18}, it is easy to verify that each permutation in $\Gamma_{(\Omega \setminus W)}$ is indeed an automorphism.
However, for hypergraphs, this is not true anymore since the hyperedges may enforce non-trivial dependencies between $W$ and $\Omega \setminus W$.

To circumvent this problem we introduce a novel simplification step which is repeatedly executed during the Local Certificates Routine.
The idea for this simplification step is based on the following observation.
Suppose that all hyperedges are identical on the window $W$ (i.e., $E_1 \cap W = E_2 \cap W$ for all hyperedges $E_1,E_2 \in \CE$).
In this case each permutation in $\Gamma_{(\Omega \setminus W)}$ is indeed an automorphism since there are no additional dependencies between $W$ and $\Omega \setminus W$ coming from the hyperedges.
Similar to the normalization procedure, the main idea of the simplification step is to create a new equivalent instance $(\Gamma^*,\CH_1^*,\CH_2^*)$ together with a window $W^*$ (that corresponds to $W$) that satisfies the above property that all hyperedges are identical on the window $W^*$.
Intuitively speaking, this is achieved as follows.
We call two hyperedges $E_1,E_2 \subseteq \Omega$ \emph{$W$-equivalent} if $E_1 \cap W = E_2 \cap W$.
The algorithm creates, for each equivalence class, a copy of the vertex set and associates all hyperedges from the equivalence class with this copy.
This creates a new equivalent instance where each copy satisfies the requirement that all hyperedges are $W$-equivalent (actually, to realize this step, we shall consider a more general problem).

Unfortunately, while this modification of instances settles the original problem, it creates several new issues that need to be addressed.
Most importantly, after this modification, the permutation group of the updated instance may not be $d$-normalized anymore.
Since the Unaffected Stabilizers Theorem only holds for $d$-normalized groups, we need to invoke a normalization procedure again.
Unfortunately, the normalization procedure by Grohe et al.\ \cite{GroheNS18} is too expensive (i.e., the newly created instances are too large) to be employed during each step of the recursive main algorithm.
Luckily, the simplification step does not completely destroy $d$-normalization, meaning that the updated permutation group created by the simplification step is still ``close'' to being $d$-normalized (assuming the input group is $d$-normalized).
This allows us to design a specialized renormalization procedure which again creates an equivalent $d$-normalized instance that is only slightly larger.
In order to be able to still analyze the progress made by the recursion, we also need to introduce the notion of a \emph{virtual size} of an instance.
Intuitively speaking, the virtual size takes into account the cost of all potential renormalization steps which means that, when the algorithm runs the renormalization procedure and creates a new instance, while its actual size might be larger than size of the original instance, the virtual size does not increase.

The renormalization of instances also creates another problem.
In the \emph{aggregation of local certificates}, the outputs of the Local Certificates Routine are compared with each other requiring the output to be isomorphism-invariant.
However, the renormalization procedure does not preserve isomorphisms between different instances when applied independently to those instances.
The solution to this problem is to run the Local Certificates Routine in parallel on all pairs of test sets compared later on.
This way, we can ensure that all instances are modified in the same way.

\paragraph{Applications.}

Over the last decades, Miller's algorithm \cite{Miller83b} for the Hypergraph Isomorphism Problem for $\mgamma_d$-groups as well as variants of this algorithm have been used as a subroutine in a number of isomorphism algorithms for restricted classes of graphs (see, e.g., \cite{BabaiCSTW13, Miller83b, Miller83a, Neuen16, Ponomarenko89, Ponomarenko91}).
This naturally raises the question which of these algorithms can be improved by using Theorem \ref{thm:main} as a subroutine instead.
In this work, we discuss two simple consequences of our improved isomorphism test for hypergraphs.
As a first application, we consider \emph{$t$-CR-bounded graphs} originally introduced by Ponomarenko (under a different name) in \cite{Ponomarenko89}.
Intuitively speaking, a vertex-colored graph is $t$-CR-bounded if the vertex-coloring of the graph can be turned into a discrete coloring (i.e., each vertex has its own color) by repeatedly applying the standard Color Refinement algorithm (see, e.g., \cite{CaiFI92,ImmermanL90}) and by splitting color classes of size at most $t$.
This class is interesting mainly because the isomorphism problem for several important classes of graphs naturally reduces to the isomorphism problem for $t$-CR-bounded graphs.
This includes for example graphs of bounded maximum degree, because every connected graph of maximum degree $t+1$ with one vertex being individualized is $t$-CR-bounded (see, e.g., \cite{Ponomarenko89}).
We show that the isomorphism problem for $t$-CR-bounded graphs can be solved in time $n^{\CO((\log t)^{c})}$ by providing a reduction to the Hypergraph Isomorphism Problem for $\mgamma_t$-groups.

As a second application of our results, we consider graphs of bounded Euler genus and, more generally, graphs that exclude $K_{3,h}$ as a minor.
We prove that the isomorphism problem for graph classes that exclude $K_{3,h}$ as a minor is polynomial-time reducible to testing isomorphism of $t$-CR-bounded graphs where $t \coloneqq h-1$.
This implies the following theorem.

\begin{theorem}
 The Graph Isomorphism Problem for graphs that exclude $K_{3,h}$ as a minor can be solved in time $n^{\CO((\log h)^{c})}$ for some absolute constant $c$ where $n$ denotes the number of vertices of the input graphs.
\end{theorem}

The previous best algorithm for this setting is due to Ponomarenko \cite{Ponomarenko91} who provided a polynomial time isomorphism test for graph classes that exclude an arbitrary minor.
While Ponomarenko does not provide a precise analysis on the running time of his algorithm, the exponent depends at least linearly on $h$ when excluding $K_h$ as a minor.
Hence, in the case of excluding $K_{3,h}$ as a minor, the above theorem significantly improves on all previous algorithms.

Finally, exploiting the fact that $K_{3,h}$ has Euler genus linear in $h$ \cite{Ringel65,Harary91}, we obtain the following corollary.

\begin{corollary}
 The Graph Isomorphism Problem for graphs of Euler genus at most $g$ can be solved in time $n^{\CO((\log g)^{c})}$ for some absolute constant $c$ where $n$ denotes the number of vertices of the input graphs.
\end{corollary}

While the isomorphism problem for graphs of bounded genus is also fixed-parameter tractable \cite{Kawarabayashi15,Neuen21}, the algorithm from the previous corollary is guaranteed to be faster as soon as the genus passes a threshold that is polylogarithmic in the number of vertices.
Also, I remark that the isomorphism problem for graphs of bounded genus can be solved in logarithmic space \cite{ElberfeldK14}.

\paragraph{Related Work.}
\label{par:related-work}

Another proof of Corollary \ref{cor:hypergraph-isomorphism-intro} has been independently obtained by Daniel Wiebking \cite{Wiebking20}.
Actually, Wiebking shows that canonization of arbitrary hereditarily finite objects (which may even include implicitly represented \emph{labeling cosets}) can be done in quasipolynomial time $(n+m)^{\CO((\log n)^{c})}$ where $n$ denotes the number of vertices and $m$ the size of the hereditarily finite objects.
However, the algorithmic framework of Wiebking \cite{Wiebking20} is unable to exploit any restrictions on the input groups (as in Theorem \ref{thm:main}).
Hence, the two results are incomparable with respect to the power of the algorithms obtained.

This is also highlighted by the applications.
While the results of this paper allow the design of an algorithm solving the isomorphism problem for graphs of Euler genus at most $g$ running in time $n^{\CO((\log g)^{c})}$,
Wiebking utilizes his algorithm to build an isomorphism test for graphs of tree-width at most $k$ running in time $n^{\CO((\log k)^{c})}$.

In two follow-up works \cite{GroheNW20,Neuen22}, we managed to combine and extend the results of \cite{Wiebking20} and the present paper to obtain isomorphism algorithms with similar running times for much larger classes of graphs.
In a first step, \cite{GroheNW20} gives an isomorphism test for all graphs excluding $K_h$ as a minor running in time $n^{\CO((\log h)^{c})}$.
This result is further generalized in \cite{Neuen22} to all graphs that exclude $K_h$ as a topological subgraph.
Both results crucially rely on the notion of $t$-CR-bounded graphs and the fast isomorphism test for such graphs based on Theorem \ref{thm:main}.
For a more high-level description on how these result connect to one another I also refer to the recent survey \cite{GroheN21}.

\paragraph{Structure of the Paper.}
\label{par:structer-paper}

The remainder of this work is structured as follows.
In the next section we give the necessarily preliminaries and review the normalization procedure for $\mgamma_d$-groups in Section \ref{sec:normalization-framework}.
Section \ref{sec:hypergraph-isomorphism} is devoted to the proof of Theorem \ref{thm:main}.
Afterwards, we discuss some consequences of Theorem \ref{thm:main}.
We define $t$-CR-bounded graphs in Section \ref{sec:t-cr-bounded} and present an isomorphism algorithm.
Finally, we obtain a new isomorphism test graphs of small genus in Section \ref{sec:genus}.

\section{Preliminaries}

\subsection{Graphs and other Combinatorial Objects}

\paragraph{Sets and Partitions.}
We write $[n]$ to denote the set $\{1,2,\dots,n\}$.
For a finite set $X$ the power set of $X$ is denoted by $2^{X} \coloneqq \{Y \mid Y \subseteq X\}$.
Note that $|2^{X}| = 2^{|X|}$.
Also, for $t \leq |X|$, the set of all $t$-element subsets of $X$ is denoted $\binom{X}{t} \coloneqq \{Y \subseteq X \mid |Y| = t\}$.
Again, observe that $|\binom{X}{t}| = \binom{|X|}{t}$.
Similarly, $\binom{X}{\leq t} \coloneqq \{Y \subseteq X \mid |Y| \leq t\}$ denotes the set of all subsets of $X$ of size at most $t$.
For three sets $X,X_1,X_2$, the set $X$ is the \emph{disjoint union} of $X_1$ and $X_2$, denoted by $X = X_1 \uplus X_2$, if $X = X_1 \cup X_2$ and $X_1 \cap X_2 = \emptyset$.
For a finite set $X$ a \emph{partition} of $X$ is a collection $\FB$ of subsets such that $B_1 \cap B_2 = \emptyset$ for all $B_1,B_2 \in \FB$ and $\bigcup_{B \in \FB} B = X$.
A partition $\FB$ is called an \emph{equipartition} if $|B_1| = |B_2|$ for all $B_1,B_2 \in \FB$.
For $S \subseteq X$ define $\FB[S] \coloneqq \{B \cap S \mid B \in \FB\colon B \cap S \neq \emptyset\}$ to be the \emph{induced partition} on $S$.
A partition $\FB_1$ of a set $X$ \emph{refines} another partition $\FB_2$ of $X$, denoted $\FB_1 \preceq \FB_2$, if for every $B_1 \in \FB_1$ there is some $B_2 \in \FB_2$ such that $B_1 \subseteq B_2$.
If additionally $\FB_1 \neq \FB_2$ we say that $\FB_1$ \emph{strictly refines} $\FB_2$ which is denoted by $\FB_1 \prec \FB_2$.

\paragraph{Graphs.}
A \emph{graph} is a pair $G = (V(G),E(G))$ with vertex set $V(G)$ and edge set $E(G)$.
Unless stated otherwise, all graphs are undirected and simple graphs, i.e., there are no loops or multiedges.
In this setting an edge is denoted as $vw$ where $v,w \in V(G)$.
The \emph{(open) neighborhood} of a vertex $v$ is denoted $N_G(v) \coloneqq \{w \in V(G) \mid vw \in E(G)\}$.
Also, for a set of vertices $X \subseteq V(G)$, the neighborhood of $X$ is defined as $N_G(X) \coloneqq \left(\bigcup_{v \in X} N_G(v)\right) \setminus X$.
The \emph{degree} of a vertex $v \in V(G)$, denoted $\deg_G(v)$, is the size of its neighborhood.
Usually, we omit the index $G$ if it is clear from the context and simply write $N(v)$, $N(X)$ and $\deg(v)$.
Let $v,w \in V(G)$.
A \emph{path} from $v$ to $w$ is a sequence of pairwise distinct vertices $v=u_0,u_1,\dots,u_{\ell-1},u_\ell=w$ such that $u_{i-1}u_i \in E(G)$ for all $i \in [\ell]$.
In this case $\ell$ is the \emph{length} of the path.
The \emph{distance} between $v$ and $w$, denoted $\dist_G(v,w)$, is the length of a shortest path from $v$ to $w$.
As before, the index $G$ is usually omitted.
For $X \subseteq V(G)$ the \emph{induced subgraph on $X$} is $G[X] \coloneqq (X,\{vw \mid v,w \in X, vw \in E(G)\})$.
Also, $G - X \coloneqq G[V(G) \setminus X]$ denotes the induced subgraph on the complement of $X$.

Two graphs $G_1$ and $G_2$ are \emph{isomorphic} if there is a bijection $\varphi\colon V(G_1) \rightarrow V(G_2)$ such that $vw \in E(G_1)$ if and only if $\varphi(v)\varphi(w) \in E(G_2)$ for all $v,w \in V(G_1)$.
In this case $\varphi$ is an \emph{isomorphism} from $G_1$ to $G_2$, which is also denoted by $\varphi\colon G_1 \cong G_2$.
Also, $\Iso(G_1,G_2)$ denotes the set of all isomorphisms from $G_1$ to $G_2$.
The automorphism group of $G_1$ is $\Aut(G_1) \coloneqq \Iso(G_1,G_1)$.
Observe that, if $\Iso(G_1,G_2) \neq \emptyset$, it holds that $\Iso(G_1,G_2) = \Aut(G_1)\varphi \coloneqq \{\gamma\varphi \mid \gamma \in \Aut(G_1)\}$ for every isomorphism $\varphi \in \Iso(G_1,G_2)$.
The Graph Isomorphism Problem takes as input two graphs $G_1$ and $G_2$ and asks whether they are isomorphic. 

A \emph{(vertex-)colored graph} is a tuple $G = (V(G),E(G),\chi_V)$ where $\chi_V\colon V(G) \rightarrow C$ is a mapping and $C$ is some finite set of colors.
For ease of notation I also often write $(G,\chi_V)$ in order to explicitly refer to the vertex-coloring $\chi_V$.
A \emph{vertex- and arc-colored graph} is a tuple $G = (V(G),E(G),\chi_V,\chi_E)$ where $\chi_V\colon V(G) \rightarrow C$ is a vertex-coloring and $\chi_E\colon \{(v,w) \mid vw \in E(G)\} \rightarrow C$ is an arc-coloring, where $C$ is again some finite set of colors.
Observe that the arc-coloring colors directed edges, i.e., for an undirected edge $vw \in E(G)$ we may assign the directed edge $(v,w)$ a different color than $(w,v)$.
Note that an uncolored graph can be interpreted as a vertex- and arc-colored graph where each vertex is assigned the same color as well as each (directed) edge is assigned the same color.
The \emph{vertex color classes} of a (colored) graph are the sets $\chi_V^{-1}(c)$ where $c \in C$.
Note that the color classes form a partition of the vertex set.
A vertex-coloring $\chi_V$ is \emph{discrete} if all color classes are singletons, i.e., $\chi_V(v) \neq \chi_V(w)$ for all distinct $v,w \in V(G)$.

\paragraph{Hypergraphs}
A \emph{hypergraph} is a pair $\CH = (V(\CH),\CE(\CH))$ with vertex set $V(\CH)$ and edge set $\CE(\CH) \subseteq 2^{V(\CH)}$.
Tow hypergraphs $\CH_1$ and $\CH_2$ are \emph{isomorphic} if there is a bijection $\varphi\colon V(\CH_1) \rightarrow V(\CH_2)$ such that $E \in \CE(\CH_1)$ if and only if $E^\varphi \coloneqq \{\varphi(v) \mid v \in E\} \in \CE(\CH_2)$ for all $E \in 2^{V(\CH_1)}$.
In this case $\varphi$ is an \emph{isomorphism} from $\CH_1$ to $\CH_2$, which is also denoted by $\varphi\colon \CH_1 \cong \CH_2$.
In this work, the \emph{Hypergraph Isomorphism Problem} takes as input two hypergraphs $\CH_1$ and $\CH_2$ over the same vertex set $V \coloneqq V(\CH_1) = V(\CH_2)$ and a permutation group $\Gamma \leq \Sym(V)$ (given by a set of generators),
and asks whether there is some $\gamma \in \Gamma$ such that $\gamma\colon \CH_1 \cong \CH_2$.

\subsection{Group Theory}

In this section we introduce the group-theoretic notions required in this work.
For a general background on group theory I refer to \cite{Rotman99} whereas background on permutation groups can be found in \cite{DixonM96}.

\subsubsection{Permutation groups}

A \emph{permutation group} acting on a set $\Omega$ is a subgroup $\Gamma \leq \Sym(\Omega)$ of the symmetric group.
The size of the permutation domain $\Omega$ is called the \emph{degree} of $\Gamma$ and, throughout this work, is denoted by $n = |\Omega|$.
If $\Omega = [n]$ then we also write $S_n$ instead of $\Sym(\Omega)$.
Also, $\Alt(\Omega)$ denotes the alternating group on the set $\Omega$ and, similar to the symmetric group, we write $A_n$ instead of $\Alt(\Omega)$ if $\Omega = [n]$.
For $\gamma \in \Gamma$ and $\alpha \in \Omega$ we denote by $\alpha^{\gamma}$ the image of $\alpha$ under the permutation $\gamma$.
The set $\alpha^{\Gamma} \coloneqq \{\alpha^{\gamma} \mid \gamma \in \Gamma\}$ is the \emph{orbit} of $\alpha$.
The group $\Gamma$ is \emph{transitive} if $\alpha^{\Gamma} = \Omega$ for some (and therefore every) $\alpha \in \Omega$.

For $\alpha \in \Omega$ the group $\Gamma_\alpha \coloneqq \{\gamma \in \Gamma \mid \alpha^{\gamma} = \alpha\} \leq \Gamma$ is the \emph{stabilizer} of $\alpha$ in $\Gamma$.
The group $\Gamma$ is \emph{semi-regular} if $\Gamma_\alpha = \{\id\}$ for all $\alpha \in \Omega$ ($\id$ denotes the identity element of the group).
If additionally $\Gamma$ is transitive then the group $\Gamma$ is called \emph{regular}.
For $A \subseteq \Omega$ and $\gamma \in \Gamma$ let $A^{\gamma} \coloneqq \{\alpha^{\gamma} \mid \alpha \in A\}$.
The \emph{pointwise stabilizer} of $A$ is the subgroup $\Gamma_{(A)} \coloneqq \{\gamma \in \Gamma \mid\forall \alpha \in A\colon \alpha^{\gamma}= \alpha \}$.
The \emph{setwise stabilizer} of $A$ is the subgroup $\Gamma_{A} \coloneqq \{\gamma \in \Gamma \mid A^{\gamma}= A\}$.
Observe that $\Gamma_{(A)} \leq \Gamma_{A}$.
Also, for a $\Gamma$-invariant set $A \subseteq \Omega$, we denote by $\Gamma[A]$ the induced action of $\Gamma$ on $A$, i.e., the group obtained from $\Gamma$ by restricting all permutations to $A$.

Let $\Gamma \leq \Sym(\Omega)$ be a transitive group.
A \emph{block} of $\Gamma$ is a nonempty subset $B \subseteq \Omega$ such that $B^{\gamma} = B$ or $B^{\gamma} \cap B = \emptyset$ for all $\gamma \in \Gamma$.
The trivial blocks are $\Omega$ and the singletons $\{\alpha\}$ for
$\alpha \in \Omega$.
The group $\Gamma$ is said to be \emph{primitive} if there are no non-trivial blocks.
If $\Gamma$ is not primitive it is called \emph{imprimitive}.
If $B \subseteq \Omega$ is a block of $\Gamma$ then $\mathfrak{B} = \{B^{\gamma} \mid \gamma \in \Gamma\}$ forms a \emph{block system} of $\Gamma$.
Note that $\mathfrak{B}$ is an equipartition of $\Omega$.
The group $\Gamma_{(\mathfrak{B})} \coloneqq \{\gamma \in \Gamma \mid \forall B \in \mathfrak{B}\colon B^{\gamma} = B\}$ denotes the subgroup stabilizing each block $B \in \mathfrak{B}$ setwise.
Observe that $\Gamma_{(\FB)}$ is a normal subgroup of~$\Gamma$.
We denote by $\Gamma[\FB] \leq \Sym(\FB)$ the natural induced action of $\Gamma$ on the block system $\FB$.
More generally, if $A$ is a set of objects on which $\Gamma$ acts naturally, we denote by $\Gamma[A] \leq \Sym(A)$ the induced action of $\Gamma$ on the set $A$.
A block system $\mathfrak{B}$ is \emph{minimal} if there is no non-trivial block system $\mathfrak{B}'$ such that $\mathfrak{B} \prec \mathfrak{B}'$.
A block system $\mathfrak{B}$ is minimal if and only if $\Gamma[\FB]$ is primitive.

Let $\Gamma \leq \Sym(\Omega)$ and $\Gamma' \leq \Sym(\Omega')$.
A \emph{homomorphism} is a mapping $\varphi\colon \Gamma \rightarrow \Gamma'$ such that $\varphi(\gamma)\varphi(\delta) = \varphi(\gamma\delta)$ for all $\gamma,\delta \in \Gamma$.
For $\gamma \in \Gamma$ we denote by $\gamma^{\varphi}$ the $\varphi$-image of $\gamma$.
Similarly, for $\Delta \leq \Gamma$ we denote by $\Delta^{\varphi}$ the $\varphi$-image of $\Delta$ (note that $\Delta^{\varphi}$ is a subgroup of $\Gamma'$).

A \emph{permutational isomorphism} from $\Gamma$ to $\Gamma'$ is a bijective mapping $f:\Omega\to\Omega'$
such that $\Gamma' = \{f^{-1}\gamma f \mid \gamma \in \Gamma\}$ where $f^{-1}\gamma f\colon \Omega' \rightarrow \Omega'$ is the unique map mapping~$f(\alpha)$ to~$f(\alpha^{\gamma})$ for all~$\alpha\in \Omega'$.
If there is a permutational isomorphism from $\Gamma$ to $\Gamma'$, we call $\Gamma$ and $\Gamma'$ \emph{permutationally equivalent}.
A \emph{permutational automorphism} of $\Gamma$ is a permutational isomorphism from $\Gamma$ to itself.

\subsubsection{Algorithms for permutation groups}

We review some basic facts about algorithms for permutation groups.
For detailed information I refer to \cite{Seress03}.

In order to perform computational tasks for permutation groups efficiently the groups are represented by generating sets of small size.
Indeed, most algorithms are based on so called strong generating sets,
which can be chosen of size quadratic in the size of the permutation domain of the group and can be computed in polynomial time given an arbitrary generating set (see, e.g., \cite{Seress03}).

\begin{theorem}[cf.\ \cite{Seress03}] 
 \label{thm:permutation-group-library}
 Let $\Gamma \leq \Sym(\Omega)$ and let $S$ be a generating set for $\Gamma$.
 Then the following tasks can be performed in time polynomial in $n$ and $|S|$:
 \begin{enumerate}
  \item compute the order of $\Gamma$,
  \item given $\gamma \in \Sym(\Omega)$, test whether $\gamma \in \Gamma$,
  \item compute the orbits of $\Gamma$,
  \item given $\Delta \subseteq \Omega$, compute a generating set for $\Gamma_{(\Delta)}$, and
  \item compute a minimal block system for $\Gamma$.
 \end{enumerate}
 For a second group $\Gamma' \leq \Sym(\Omega')$ with domain size $n' = |\Omega'|$, the following tasks can be solved in time polynomial in $n$, $n'$ and $|S|$:
 \begin{enumerate}
  \setcounter{enumi}{5}
  \item given a homomorphism $\varphi\colon \Gamma \rightarrow \Gamma'$ (given as a list of images for $\gamma \in S$),
  \begin{enumerate}
   \item compute a generating set for $\ker(\varphi) = \{\gamma \in \Gamma \mid \varphi(\gamma) = 1\}$, and
   \item given $\gamma'\in \Gamma'$, compute an element $\gamma \in \Gamma$ such that
     $\varphi(\gamma) = \gamma'$ (if it exists).
  \end{enumerate}
 \end{enumerate}
\end{theorem}

\subsubsection{Groups with restricted composition factors}

In this work we shall be interested in a particular subclass of permutation groups, namely groups with restricted composition factors.
Let $\Gamma$ be a group.
A \emph{subnormal series} is a sequence of subgroups $\Gamma = \Gamma_0 \trianglerighteq \Gamma_1 \trianglerighteq \dots \trianglerighteq \Gamma_k = \{\id\}$.
The length of the series is $k$ and the groups $\Gamma_{i-1} / \Gamma_{i}$ are the factor groups of the series, $i \in [k]$.
A \emph{composition series} is a strictly decreasing subnormal series of maximal length. 
For every finite group $\Gamma$ all composition series have the same family of factor groups considered as a multi-set (cf.\ \cite{Rotman99}).
A \emph{composition factor} of a finite group $\Gamma$ is a factor group of a composition series of $\Gamma$.

\begin{definition}
 For $d \geq 2$ let $\mgamma_d$ denote the class of all groups $\Gamma$ for which every composition factor of $\Gamma$ is isomorphic to a subgroup of $S_d$.
\end{definition}

I want to stress the fact that there are two similar classes of groups that have been used in the literature both typically denoted by $\Gamma_d$.
One of these is the class we define as~$\mgamma_d$ introduced by Luks~\cite{Luks82} while the other one used in~\cite{BabaiCP82} in particular allows composition factors that are simple groups of Lie type of bounded dimension.

\begin{lemma}[Luks \cite{Luks82}]
 \label{la:gamma-d-closure}
 Let $\Gamma \in \mgamma_d$. Then
 \begin{enumerate}
  \item $\Delta \in \mgamma_d$ for every subgroup $\Delta \leq \Gamma$, and
  \item $\Gamma^{\varphi} \in \mgamma_d$ for every homomorphism $\varphi\colon \Gamma \rightarrow \Delta$.
 \end{enumerate}
\end{lemma}

\subsection{String Isomorphism}
\label{subsec:string-isomorphism}

Next, we review several basic facts about the String Isomorphism Problem and the recursion mechanisms of Luks's algorithm \cite{Luks82}.
While this paper is primarily concerned with the Hypergraph Isomorphism Problem, we use similar recursion techniques and, at several points, need to solve instances of the String Isomorphism Problem as a subroutine.

In order to apply group-theoretic techniques in the context of the Graph Isomorphism Problem Luks \cite{Luks82} introduced the more general String Isomorphism Problem that allows to build recursive algorithms along the structure of the permutation groups involved.
For this paper, I follow the notation and terminology used by Babai \cite{Babai15,Babai16} for describing his quasipolynomial time algorithm for the Graph Isomorphism Problem that also employs such a recursive strategy.   

Recall that a \emph{string} is a mapping $\Fx\colon\Omega\rightarrow\Sigma$ where $\Omega$ is a finite set and $\Sigma$ is also a finite set called the \emph{alphabet}.
Let $\gamma \in \Sym(\Omega)$ be a permutation.
The permutation $\gamma$ can be applied to the string $\Fx$ by defining
\begin{equation}
 \Fx^{\gamma}\colon\Omega\rightarrow\Sigma\colon\alpha\mapsto\Fx\left(\alpha^{\gamma^{-1}}\right).
\end{equation}
Let $\Fy\colon\Omega\rightarrow\Sigma$ be a second string.
The permutation $\gamma$ is an isomorphism from $\Fx$ to $\Fy$, denoted $\gamma\colon \Fx \cong \Fy$, if $\Fx^{\gamma} = \Fy$.
Let $\Gamma \leq \Sym(\Omega)$.
A \emph{$\Gamma$-isomorphism} from $\Fx$ to $\Fy$ is a permutation $\gamma \in \Gamma$ such that $\gamma\colon \Fx \cong \Fy$.
The strings $\Fx$ and $\Fy$ are \emph{$\Gamma$-isomorphic}, denoted $\Fx \cong_\Gamma \Fy$, if there is a $\Gamma$-isomorphism from $\Fx$ to $\Fy$.
The \emph{String Isomorphism Problem} asks, given two strings $\Fx,\Fy\colon\Omega\rightarrow\Sigma$ and a generating set for $\Gamma \leq \Sym(\Omega)$, whether $\Fx$ and $\Fy$ are $\Gamma$-isomorphic.
The set of $\Gamma$-isomorphisms from $\Fx$ to $\Fy$ is denoted by
\begin{equation}
 \Iso_\Gamma(\Fx,\Fy) := \{\gamma \in \Gamma \mid \Fx^{\gamma} = \Fy\}.
\end{equation}
The set of \emph{$\Gamma$-automorphisms} of $\Fx$ is $\Aut_\Gamma(\Fx) \coloneqq \Iso_\Gamma(\Fx,\Fx)$.
Observe that $\Aut_\Gamma(\Fx)$ is a subgroup of $\Gamma$ and $\Iso_\Gamma(\Fx,\Fy) = \Aut_\Gamma(\Fx)\gamma$ for an arbitrary $\gamma \in \Iso_\Gamma(\Fx,\Fy)$ (or $\Iso_\Gamma(\Fx,\Fy) = \emptyset$).

The foundation for the group-theoretic approaches to the String Isomorphism Problem lays in two recursion mechanisms first exploited by Luks \cite{Luks82}.
As before, let $\Fx,\Fy\colon \Omega \rightarrow \Sigma$ be two strings.
For a set of permutations $K \subseteq \Sym(\Omega)$ and a \emph{window} $W \subseteq \Omega$ we define
\begin{equation}
 \Iso_K^{W}(\Fx,\Fy) := \{\gamma \in K \mid \forall \alpha \in W\colon \Fx(\alpha) = \Fy(\alpha^{\gamma})\}.
\end{equation}
In this work, the set $K$ is always a coset, i.e., $K = \Gamma\gamma$ for some group $\Gamma \leq \Sym(\Omega)$ and a permutation $\gamma \in \Sym(\Omega)$, and the set $W$ is $\Gamma$-invariant.
In this case it can be shown that $\Iso_K^{W}(\Fx,\Fy)$ is either empty or a coset of the group $\Aut_\Gamma^{W}(\Fx) \coloneqq \Iso_\Gamma^{W}(\Fx,\Fx)$.
Hence, the set $\Iso_K^{W}(\Fx,\Fy)$ can be represented by a generating set for $\Aut_\Gamma^{W}(\Fx)$ and a single permutation $\sigma \in \Iso_K^{W}(\Fx,\Fy)$.
Moreover, for $K = \Gamma\gamma$ it holds that
\begin{equation}
 \label{eq:string-isomorphism-alignment}
 \Iso_{\Gamma\gamma}^{W}(\Fx,\Fy) = \Iso_\Gamma^{W}(\Fx,\Fy^{\gamma^{-1}})\gamma.
\end{equation}
Using this identity, it is possible to restrict to the case where $K$ is actually a group.

For the first type of recursion suppose $K = \Gamma \leq \Sym(\Omega)$ is not transitive on $W$ and let $W_1,\dots,W_\ell$ be the orbits of $\Gamma[W]$.
Then the strings can be processed orbit by orbit as described in Algorithm \ref{alg:orbit-by-orbit}.
This recursion mechanism is referred to as \emph{orbit-by-orbit processing}.

\begin{algorithm}
 \caption{Orbit-by-Orbit processing}
 \label{alg:orbit-by-orbit}
 \SetKwInOut{Input}{Input}
 \SetKwInOut{Output}{Output}
 \Input{Strings $\Fx,\Fy\colon\Omega\rightarrow\Sigma$, a group $\Gamma \leq \Sym(\Omega)$, and a $\Gamma$-invariant set $W \subseteq \Omega$ such that $\Gamma[W]$ is not transitive.}
 \Output{$\Iso_\Gamma^{W}(\Fx,\Fy)$}
 \DontPrintSemicolon
 compute orbits $W_1,\dots,W_\ell$ of $\Gamma[W]$\;
 $K := \Gamma$\;
 \For{$i=1,\dots,\ell$}{
  $K := \Iso_K^{W_i}(\Fx,\Fy)$\;
 }
 \Return $K$\;
\end{algorithm}

The set $\Iso_K^{W_i}(\Fx,\Fy)$ can be computed making one recursive call to the String Isomorphism Problem over domain size $n_i \coloneqq |W_i|$.
Indeed, using Equation \eqref{eq:string-isomorphism-alignment}, it can be assumed that $K = \Gamma \leq \Sym(\Omega)$ is a group and $W_i$ is $\Gamma$-invariant.
Then
\begin{equation}
 \Iso_\Gamma^{W_i}(\Fx,\Fy) = \{\gamma \in \Gamma \mid \gamma[W_i] \in \Iso_{\Gamma[W_i]}(\Fx[W_i],\Fy[W_i])\}
\end{equation}
where $\Fx[W_i]$ denotes the induced substring of $\Fx$ on the set $W_i$, i.e., $\Fx[W_i] \colon W_i \rightarrow \Sigma\colon \alpha \mapsto \Fx(\alpha)$ (the string $\Fy[W_i]$ is defined analogously).
So overall, if $\Gamma$ is not transitive, the set $\Iso_\Gamma^{W}(\Fx,\Fy)$ can be computed making $\ell$ recursive calls over window size $n_i = |W_i|$, $i \in [\ell]$.

For the second recursion mechanism let $\Delta \leq \Gamma$ and let $T$ be a transversal for $\Delta$ in $\Gamma$.
Then
\begin{equation}
 \Iso_\Gamma^{W}(\Fx,\Fy) = \bigcup_{\delta \in T} \Iso_{\Delta\delta}^{W}(\Fx,\Fy) = \bigcup_{\delta \in T} \Iso_{\Delta}^{W}(\Fx,\Fy^{\delta^{-1}})\delta.
\end{equation}
Luks applied this type of recursion when $\Gamma$ is transitive (on the window $W$), $\FB$ is a minimal block system for $\Gamma$, and $\Delta = \Gamma_{(\FB)}$.
In this case $\Gamma[\FB]$ is a primitive group.
Let $t = |\Gamma[\FB]|$ be the size of a transversal for $\Delta$ in $\Gamma$.
Note that $\Delta$ is not transitive (on the window $W$).
Indeed, each orbit of $\Delta$ has size at most $n/b$ where $b = |\FB|$.
So by combining both types of recursion the computation of $\Iso_\Gamma^{W}(\Fx,\Fy)$ is reduced to $t \cdot b$ many instances of the String Isomorphism Problem over window size $|W|/b$.
This specific combination of types of recursion is referred to as \emph{standard Luks reduction}.
Observe that the time complexity of standard Luks reduction is determined by the size of the primitive group $\Gamma[\FB]$.

While this work considers a more general problem to be introduced in Section \ref{sec:hypergraph-isomorphism}, the basic notation and recursion mechanisms can also be applied in this more general setting and play an important role for the isomorphism test for hypergraphs designed in this paper.
Regarding the String Isomorphism Problem for $\mgamma_d$-groups, we shall require the following result which is also based on the recursion techniques described above as well as an extension of the Local Certificates Routine introduced by Babai \cite{Babai16} to the setting of $\mgamma_d$-groups.

\begin{theorem}[{\cite[Theorem VII.3]{GroheNS18}}]
 \label{thm:string-isomorphism-gamma-d}
 The String Isomorphism Problem for $\mgamma_d$-groups can be solved in time $n^{\CO((\log d)^{c})}$ for some constant $c$.
\end{theorem}

\begin{corollary}
 \label{cor:graph-isomorphism-gamma-d}
 Let $G_1$ and $G_2$ be two vertex- and arc-colored graphs over the same set of vertices $V \coloneqq V(G_1) = V(G_2)$ and let $\Gamma \leq \Sym(V)$ be a $\mgamma_d$-group.
 Then a representation for the set $\Iso_\Gamma(G_1,G_2) \coloneqq \{\gamma \in \Gamma \mid \gamma \colon G_1 \cong G_2\}$ can be computed in time $n^{\CO((\log d)^{c})}$ for some constant $c$.
\end{corollary}

Note that, as is the case for string isomorphisms, the set $\Iso_\Gamma(G_1,G_2)$ is either empty or it is a coset of $\Aut_\Gamma(G_1) \leq \Gamma$.
Hence, $\Iso_\Gamma(G_1,G_2)$ can be represented by a generating set for $\Aut_\Gamma(G_1)$ and a single isomorphism $\gamma \in \Iso_\Gamma(G_1,G_2)$.

\subsection{Recursion}

For the purpose of later analyzing our recursion, we record the following bound.

\begin{lemma}[{\cite[Lemma II.4]{GroheNS18}}]
  \label{la:recursion-bound}
 Let $k\in\mathbb N$ and $t\colon \mathbb{N} \rightarrow \mathbb{N}$ be a function such that $t(1) = 1$. Suppose that for every $n \in \mathbb{N}$ there are natural numbers~$n_1,\dots,n_\ell$ for which one of the following holds:
 \begin{enumerate}
  \item $t(n) \leq \sum_{i=1}^{\ell} t(n_i)$ where $\sum_{i=1}^{\ell} n_i \leq 2^{k}n$ and $n_i \leq n/2$ for all $i \in [\ell]$, or
  \item $t(n) \leq \sum_{i=1}^{\ell} t(n_i)$ where $\sum_{i=1}^{\ell} n_i \leq n$ and $\ell \geq 2$.
 \end{enumerate}
 Then $t(n) \leq n^{k+1}$.
\end{lemma}

\section{A Framework for Computing a Normalized Group Action}
\label{sec:normalization-framework}

Next, we review the normalization procedure introduced in \cite{GroheNS18}.
The purpose of the normalization is to take as input an isomorphism instance of the form $(\Gamma,\FX,\FY)$, where $\Gamma \leq \Sym(\Omega)$ is a $\mgamma_d$-group and $(\FX,\FY)$ are combinatorial objects (such as strings, hypergraphs, etc.) defined on $\Omega$,
and to return an equivalent instance (with respect to isomorphism testing) $(\Gamma^*,\FX^*,\FY^*)$, where $\Gamma^* \leq \Sym(\Omega^*)$ and $(\FX^*,\FY^*)$ are combinatorial objects on $\Omega^*$, such that additionally \emph{$\Gamma^*$ is $d$-normalized}.

The algorithm for the String Isomorphism Problem \cite{GroheNS18} applies such a normalization procedure once in the beginning to obtain a $d$-normalized permutation group and then maintains this property throughout the remaining algorithm.
In contrast, the algorithm presented in this work is not able to preserve the property that groups are $d$-normalized throughout its execution.
As a result, we constantly need to \emph{renormalize} instances by calling the normalization procedure in each step of the recursive algorithm.
Actually, since the generic normalization procedure from \cite{GroheNS18} may significantly increase the size of the domain $\Omega$, we cannot afford to apply this subroutine during each recursive call.

In order to fix this problem, the main insight is that during our recursive algorithm, the property of permutation groups being $d$-normalized is ``almost'' preserved.
This allows us restore $d$-normalization using a specialized renormalization subroutine that only slightly increases the domain size (which is then affordable in the desired time frame).

\subsection{Almost $d$-ary Sequences of Partitions}
\label{subsec:almost-d-ary}

To implement the renormalization subroutine, it turns out to be useful to take a different perspective compared to \cite{GroheNS18} that mostly builds on graph-theoretic ideas (see also \cite[Chapter 6.1]{Neuen19}) which is discussed below.

A \emph{rooted tree} is a pair $(T,v_0)$ where $T$ is a directed tree and $v_0 \in V(T)$ is the root of $T$ (all edges are directed away from the root).
Let $L(T)$ denote the set of leaves of $T$, i.e., vertices $v \in V(T)$ without outgoing edges.
For $v \in V(T)$ we denote by $T^{v}$ the subtree of $T$ rooted at vertex $v$.

\begin{definition}[Structure Tree]
 Let $\Gamma \leq \Sym(\Omega)$ be a permutation group.
 A \emph{structure tree} for $\Gamma$ is a rooted tree $(T,v_0)$ such that $L(T) = \Omega$ and $\Gamma \leq (\Aut(T))[\Omega]$.
\end{definition}

\begin{lemma}[{\cite[Lemma 6.1.2]{Neuen19}}]
 Let $\Gamma \leq \Sym(\Omega)$ be a transitive group and $(T,v_0)$ a structure tree for $\Gamma$.
 For every $v \in V(T)$ the set $L(T^{v})$ is a block of $\Gamma$.
 Moreover, $\{L(T^{w}) \mid w \in v^{\Aut(T)}\}$ forms a block system of the group $\Gamma$.
\end{lemma}

Let $\Gamma \leq \Sym(\Omega)$ be a transitive group.
The previous lemma implies that every structure tree $(T,v_0)$ gives a sequence of $\Gamma$-invariant partitions $\{\Omega\} = \FB_0 \succ \dots \succ \FB_k = \{\{\alpha\} \mid \alpha \in \Omega\}$.
On the other hand, every such sequence of partitions gives a structure tree $(T,v_0)$ with
\[V(T) \coloneqq \Omega \cup \bigcup_{i=0,\dots,k-1} \FB_i\]
and
\[E(T) \coloneqq \{(B,B') \mid B \in \FB_{i-1}, B' \in \FB_i, B' \subseteq B\} \cup \{(B,\alpha) \mid B \in \FB_{k-1}, \alpha \in B\}.\]
The root is $v_0 \coloneqq \Omega$.
Hence, there is a one-to-one correspondence between structure trees and sequences of $\Gamma$-invariant partitions of the form $\{\Omega\} = \FB_0 \succ \dots \succ \FB_k = \{\{\alpha\} \mid \alpha \in \Omega\}$.
In the following, both view points are used interchangeably depending on the current task.

\begin{definition}[Almost $d$-ary Sequences of Partitions]
 \label{def:d-normalized}
 Let $\Gamma \leq \Sym(\Omega)$ be a group and let $\{\Omega\} = \FB_0 \succ \dots \succ \FB_k = \{\{\alpha\} \mid \alpha \in \Omega\}$ be a sequence of $\Gamma$-invariant partitions.
 The sequence $\FB_0 \succ \dots \succ \FB_k$ is \emph{almost $d$-ary} if for every $i \in [k]$ and $B \in \FB_{i-1}$ it holds that
 \begin{enumerate}
  \item $|\mathfrak{B}_i[B]| \leq d$, or
  \item $\Gamma_B[\FB_i[B]]$ is semi-regular.
 \end{enumerate}
 Similarly, a structure tree $(T,v_0)$ for a group $\Gamma$ is \emph{almost $d$-ary} if the corresponding sequence of partitions is.
 
 A permutation group $\Gamma \leq \Sym(\Omega)$ is \emph{$d$-normalized} if it has an almost $d$-ary sequence of partitions.
\end{definition}

For a group $\Gamma \leq \Sym(\Omega)$ and a structure tree $(T,v_0)$ let $\varphi\colon \Gamma \rightarrow \Aut(T,v_0)$ be the unique homomorphism such that $(\gamma^{\varphi})[\Omega] = \gamma$ for all $\gamma \in \Gamma$
(this homomorphism is unique since $(\Aut(T,v_0))_{(\Omega)} = \{\id\}$).
Also let $N^{+}(v) \coloneqq \{w \in V(T) \mid (v,w) \in E(T)\}$ be the set of children of $v$ for every $v \in V(T)$.
Then $(T,v_0)$ is almost $d$-ary if for every $v \in V(T)$ it holds that $|N^{+}(v)| \leq d$ or $(\Gamma^{\varphi})_v[N^{+}(v)]$ is semi-regular.

A simple, but crucial observation is that such sequences are inherited by subgroups and restrictions of the action to invariant subsets.

\begin{observation}
 \label{obs:sequence-of-partitions}
 Let $\Gamma \leq \Sym(\Omega)$ be a group, and let $\{\Omega\} = \FB_0 \succ \dots \succ \FB_k = \{\{\alpha\} \mid \alpha \in \Omega\}$ be an almost $d$-ary sequence of $\Gamma$-invariant partitions.
 Moreover, let $\Delta \leq \Gamma$.
 Then $\FB_0 \succ \dots \succ \FB_k$ also forms an almost $d$-ary sequence of $\Delta$-invariant partitions.
 Additionally, for a $\Delta$-invariant subset $S \subseteq \Omega$ it holds that $\FB_0[S] \succeq \dots \succeq \FB_m[S]$ forms an almost $d$-ary sequence of $\Delta[S]$-invariant partitions.
\end{observation}

\subsection{Structure Graphs and Tree Unfoldings}
\label{subsec:unfoldings-structure-graph}

The key tool to compute a $d$-normalized action of a group are tree unfoldings of \emph{structure graphs}.
More precisely, to obtain an action of the group that has an almost $d$-ary structure tree, the basic idea is first to construct a \emph{structure graph} for the group and afterwards unfold the graph leading to the $d$-normalized action for which the unfolded graph forms a witnessing structure tree.

For a directed graph $G$ define $\widetilde{G}$ to be the underlying undirected graph.
A \emph{rooted simple acyclic graph} is a pair $(G,v_0)$ where $G$ is a directed graph and $v_0 \in V(G)$ such that
for every $(v,w) \in E(G)$ it holds that $\dist_{\widetilde{G}}(v_0,v) + 1 = \dist_{\widetilde{G}}(v_0,w)$.
Observe that the vertices of a rooted simple acyclic graph $(G,v_0)$ can be partitioned into \emph{layers} $V_i \coloneqq \{v \in V(G) \mid \dist_{\widetilde{G}}(v_0,v) = i\}$.
Then, by definition, every edge of $G$ starts in some layer $V_i$, $i \geq 0$, and ends in layer $V_{i+1}$ (see also Figure \ref{fig:structure-graph-johnson}).

Let $G$ be a rooted simple acyclic graph.
For $v \in V(G)$ define $N^{+}(v) \coloneqq \{w \in V(G) \mid (v,w) \in E(G)\}$ to be the set of outgoing neighbors of $v$.
The \emph{forward degree} of $v$ is $\deg^{+}(v) \coloneqq |N^{+}(v)|$.
A vertex is a \emph{leaf} of $G$ if it has no outgoing neighbors, i.e., $\deg^{+}(v) = 0$.
Let $L(G) \coloneqq \{v \in V(G) \mid \deg^{+}(v) = 0\}$ denote the set of leaves of $G$.
More generally, let $L(G,v)$ be those leaves of $G$ that are reachable via a directed path from $v$.

\begin{definition}[Structure Graph]
 \label{def:structure-graph-for-groups}
 Let $\Gamma \leq \Sym(\Omega)$ be a permutation group.
 A \emph{structure graph} for $\Gamma$ is a triple $(G,v_0,\varphi)$ where $(G,v_0)$ is a rooted simple acyclic graph such that $L(G) = \Omega$ and
 $\varphi\colon \Gamma \rightarrow \Aut(G)$ is a homomorphism such that $(\gamma^{\varphi})[\Omega] = \gamma$ for all $\gamma \in \Gamma$.
\end{definition}

Intuitively speaking, a triple $(G,v_0,\varphi)$ is a structure graph for $\Gamma$ if every $\gamma \in \Gamma$ extends to an automorphism of $(G,v_0)$ via $\varphi$ (the last condition in the definition ensures that $\gamma^\varphi$ is indeed an extension of $\gamma$).

Note that each structure tree can be viewed as a structure graph (for trees the homomorphism $\varphi$ is uniquely defined and can be easily computed).
As indicated above the strategy to define a $d$-normalized action is to consider the tree unfolding of a suitable structure graph.
The permutation domain of the $d$-normalized action then corresponds to the leaves of the tree unfolding for which there is a natural action of the group $\Gamma$ (via the homomorphism $\varphi$).

Let $(G,v_0)$ be a rooted simple acyclic graph.
A \emph{branch} of $(G,v_0)$ is a sequence $(v_0,v_1,\dots,v_k)$ such that $(v_{i-1},v_i) \in E(G)$ for all $i \in [k]$.
A branch $(v_0,v_1,\dots,v_k)$ is \emph{maximal} if it can not to extended to a longer branch, i.e., if $v_k$ is a leaf of $(G,v_0)$.
Let $\Br(G,v_0)$ denote the set of branches of $(G,v_0)$ and $\Br^{*}(G,v_0)$ denote the set of maximal branches.
Note that $\Br^{*}(G,v_0) \subseteq \Br(G,v_0)$.
Also, for a maximal branch $\bar v = (v_0,v_1,\dots,v_k)$ let $L(\bar v) \coloneqq v_k$.
Note that $L(\bar v) \in L(G)$.

For a rooted simple acyclic graph $(G,v_0)$ the \emph{tree unfolding} of $(G,v_0)$ is defined to be the rooted tree $\Unf(G,v_0)$
with vertex set $\Br(G,v_0)$ and edge set
\[E(\Unf(G,v_0)) = \{((v_0,\dots,v_k),(v_0,\dots,v_k,v_{k+1})) \mid (v_0,\dots,v_{k+1}) \in \Br(G,v_0)\}.\]
Note that $L(\Unf(G,v_0)) = \Br^{*}(G,v_0)$, i.e., the leaves of the tree unfolding of $(G,v_0)$ are exactly the maximal branches of $(G,v_0)$.

\begin{example}
 Consider the group $\Gamma = S_5^{(2)}$, i.e., the action of the symmetric group $S_5$ on the set $\binom{[5]}{2}$ of all two-element subsets of $[5]$.
 A structure graph $G$ for this group is displayed in Figure \ref{fig:structure-graph-johnson}.
 Observe that $V(G) = \binom{[5]}{\leq 2}$.
 Each permutation $\gamma \in \Gamma$ naturally extends to an automorphism of $G$ (resulting in the homomorphism $\varphi$ from the definition of a structure graph).
 
 Any maximal branch of $G$ has the form $(\emptyset,\{i\},\{i,j\})$ for distinct elements $i,j \in [5]$.
 It is not difficult to see that the maximal branches of $G$ are in one-to-one correspondence with the set of all tuples $(i,j)$ where $i,j \in [5]$ and $i \neq j$.
 Now, $\Gamma$ also acts on the set of all branches via the natural action of $S_5$.
 The tree unfolding of $G$ forms a structure tree of this action.
 
 \begin{figure}
  \centering
  \begin{tikzpicture}
   \draw[line width = 5pt, gray, opacity = 0.4] (5.4,3.2) -- (3,1.6);
   \draw[line width = 5pt, gray, opacity = 0.4] (3,1.6) -- (7.2,0);
   \node[normalvertex,label={above:$\emptyset$}] (r) at (5.4,3.2) {};
   \foreach \i in {1,...,5}{
    \node[normalvertex,label={right:$\{\i\}$}] (u\i) at (2.4*\i - 1.8,1.6) {};
    \draw[->,thick] (r) edge (u\i);
   }
   \foreach \a/\b [count=\i] in {1/2,1/3,2/3,1/4,2/4,1/5,2/5,3/4,3/5,4/5}{
    \node[normalvertex,label={below:$\{\a,\b\}$}] (l\i) at (1.2*\i - 1.2,0) {};
    \draw[->,thick] (u\a) edge (l\i);
    \draw[->,thick] (u\b) edge (l\i);
   }
   
  \end{tikzpicture}
  \caption{A structure graph for the group $S_5^{(2)}$ with root vertex $v_0 = \emptyset$. A maximal branch $(\emptyset,\{2\},\{2,5\})$ is highlighted in gray.}
  \label{fig:structure-graph-johnson}
 \end{figure}
\end{example}

\begin{lemma}[{\cite[Lemma 6.1.9]{Neuen19}}]
 \label{la:standard-action-on-branches}
 Let $\Gamma \leq \Sym(\Omega)$ be a permutation group and let $(G,v_0,\varphi)$ be a structure graph for $\Gamma$.
 Then there is an action $\psi\colon \Gamma \rightarrow \Sym(\Br^{*}(G,v_0))$ on the set of maximal branches of $(G,v_0)$ such that
 \begin{enumerate}
  \item $L(\bar v^{\psi(\gamma)}) = (L(\bar v))^{\gamma}$ for all $\bar v \in \Br^{*}(G,v_0)$ and $\gamma \in \Gamma$, and
  \item $\Unf(G,v_0)$ forms a structure tree for $\Gamma^{\psi}$.
 \end{enumerate}
 Moreover, given the group $\Gamma$ and the structure graph $(G,v_0,\varphi)$, the homomorphism $\psi$ can be computed in time polynomial in $|\Br^{*}(G,v_0)|$.
\end{lemma}

Let $\Gamma \leq \Sym(\Omega)$ and let $(G,v_0,\varphi)$ be a structure graph.
Following \cite{Neuen19}, we refer to the action defined in the last lemma as the \emph{standard action of $\Gamma$ on the set of maximal branches $\Br^{*}(G,v_0)$ (with respect to $\varphi$)}.
We call $(G,v_0,\varphi)$ \emph{almost $d$-ary} if the tree unfolding $\Unf(G,v_0)$ builds an almost $d$-ary structure tree for the standard action of $\Gamma$ on $\Br^{*}(G,v_0)$.

\begin{lemma}
 \label{la:action-normalization-from-structure-graph}
 Let $\Gamma \leq \Sym(\Omega)$ be a permutation group and let $(G,v_0,\varphi)$ be an almost $d$-ary structure graph for $\Gamma$.
 Then there is an algorithm computing a homomorphism $\psi\colon \Gamma \rightarrow \Sym(\Omega^{*})$, an almost $d$-ary structure tree $(T,v_0)$ for $\Gamma^{\psi}$,
 and a mapping $f\colon \Omega^{*} \rightarrow \Omega$ such that
 \begin{enumerate}
  \item $(f^{-1}(\alpha))^{\psi(\gamma)} = f^{-1}(\alpha^{\gamma})$ for every $\alpha \in \Omega$ and $\gamma \in \Gamma$, and
  \item $|\Omega^{*}| \leq |\Br^{*}(G)|$.
 \end{enumerate}
 Moreover, the algorithm runs in time polynomial in the size of the input and $|\Omega^{*}|$.
\end{lemma}

\begin{proof}
 Let $\Omega^{*} \coloneqq \Br^{*}(G,v_0)$ and let $\psi\colon\Gamma\rightarrow\Sym(\Br^{*}(G,v_0))$ be the standard action of $\Gamma$ on the set $\Br^{*}(G,v_0)$.
 Let $\Gamma^{*} \coloneqq \Gamma^{\psi}$.
 Then $\Unf(G,v_0)$ is an almost $d$-ary structure tree for $\Gamma^{*}$ by Lemma \ref{la:standard-action-on-branches}.
 Also, define
 \[f\colon \Omega^{*} \rightarrow \Omega\colon \bar v \mapsto L(\bar v).\]
 Let $\alpha \in \Omega$ and $\gamma \in \Gamma$.
 Then
 \begin{align*}
  (f^{-1}(\alpha))^{\psi(\gamma)} &= \{\bar v \in \Br^{*}(G) \mid L(\bar v) = \alpha\}^{\psi(\gamma)} = \{\bar v^{\psi(\gamma)} \in \Br^{*}(G) \mid L(\bar v) = \alpha\} \\
                                  &= \{\bar v \in \Br^{*}(G) \mid L(\bar v) = \alpha^{\gamma}\} = f^{-1}(\alpha^{\gamma})  
 \end{align*}
 using Lemma \ref{la:standard-action-on-branches}.
\end{proof}

The last lemma provides a generic tool to implement a normalization procedure.
Given an instance $(\Gamma,\FX,\FY)$, say of the String Isomorphism Problem or the Hypergraph Isomorphism Problem, and an almost $d$-ary structure graph for $\Gamma$, the lemma provides a $d$-normalized group $\Gamma^* \coloneqq \Gamma^\psi \leq \Sym(\Omega^*)$ that is closely connected to the original group $\Gamma$ via the function $f$.
In particular, using the function $f$, we can define the corresponding combinatorial objects $(\FX^*,\FY^*)$ in a generic way to obtain an equivalent instance $(\Gamma^*,\FX^*,\FY^*)$.

Hence, the main task is to find structure graphs of small size and a small number of branches.
Such structure graphs are already implicitly obtained in \cite{GroheNS18}.
For explicit constructions I refer to \cite{Neuen19}.

\begin{lemma}[{\cite[Lemma 6.1.17]{Neuen19}}]
 \label{la:construct-structure-graph-transitive}
 Let $\Gamma \leq \Sym(\Omega)$ be a transitive $\mgamma_d$-group.
 Then there is an almost $d$-ary structure graph $(G,v_0,\varphi)$ for $\Gamma$ such that $|\Br(G,v_0)| \leq n^{\funcnorm(d) + 1}$ where $\funcnorm(d) = \CO(\log d)$.
 Moreover, there is an algorithm computing such a structure graph in time polynomial in the size of $G$.
\end{lemma}

Here, we also introduce the function $\funcnorm$ as a function for measuring the normalization cost which will play an important role in remainder of this paper.
We also need the following simple generalization of Lemma \ref{la:construct-structure-graph-transitive} which removes the restriction that $\Gamma$ has to be transitive.

\begin{lemma}
 \label{la:construct-structure-graph}
 Let $\Gamma \leq \Sym(\Omega)$ be a $\mgamma_d$-group.
 Then there is an almost $d$-ary structure graph $(G,v_0,\varphi)$ for $\Gamma$ such that $|\Br(G,v_0)| \leq n^{\funcnorm(d) + 1}$.
 Moreover, there is an algorithm computing such a structure graph in time polynomial in the size of $G$.
\end{lemma}

\begin{proof}
 This can be achieved by combining structure graphs for all orbits of $\Gamma$ constructed in Lemma \ref{la:construct-structure-graph-transitive} into a single structure graph adding a new root vertex connected to the roots of all structure graphs constructed for the orbits.
\end{proof}

Additionally, we shall use another lemma that allows us to combine given structure graphs along a given block system of the group $\Gamma$ (see also \cite[Lemma 6.1.13]{Neuen19}).

\begin{lemma}
 \label{la:combine-structure-graphs-along-block-system}
 Let $\Gamma \leq \Sym(\Omega)$ be a group, $\FB$ be a block system of size $m \coloneqq |\FB|$ such that $\Gamma[\FB]$ is transitive and let $B \in \FB$ be a block of size $b \coloneqq |B|$.
 Suppose there is an almost $d$-ary structure graph $(G_{\FB},v_0,\varphi_{\FB})$ for $\Gamma[\FB]$ such that $|\Br(G_{\FB},v_0)| \leq m^{k}$.
 Also, suppose there is an almost $d$-ary structure graph $(H_B,w_0,\varphi_B)$ for $\Gamma_B[B]$ such that $|\Br(H_B,w_0)| \leq b^{\ell}$.
 
 Then there is an almost $d$-ary structure graph $(G,v_0,\varphi)$ for $\Gamma$ such that
 \[|\Br(G,v_0)| \leq m^{k} \cdot b^{\ell} \leq n^{\max\{k,\ell\}}.\]
 Moreover, for every $B \in \FB$ there is a node $v_B \in V(G)$ such that $L(G,v_B) = B$.
 Additionally, there is an algorithm computing such a structure graph in time polynomial in the size of $G$.
\end{lemma}

The main benefit of this lemma is that, in certain cases, it allows us to obtain improved bounds on the number of branches which turns out to be crucial for the renormalization routine.

\begin{figure}
 \centering
 \begin{tikzpicture}
  \node[emptyvertex,label={above:$v_0$}] (r) at (5.6,6.0) {};
  
  \fill[darkpastelgreen, opacity = 0.4, rounded corners] (0.6,2.2) rectangle (10.6,6.6);
  \node at (9.8,6.0) {$(G,v_0)$};
  \foreach \i in {0,1,2,3}{
   \node[emptyvertex] (x\i) at (2.6*\i + 1.7,4.8) {};
   \draw[->, thick] (r) edge (x\i);
  }
  \foreach \i in {0,1,2,3,4,5}{
   \node[emptyvertex] (y\i) at (1.8*\i + 1.1,3.6) {};
  }
  \foreach \i in {0,1,2,3}{
   \node[emptyvertex] (z\i) at (3.2*\i + 0.8,2.4) {};
  }
  \draw[->, thick] (x0) edge (y0);
  \draw[->, thick] (x1) edge (y0);
  \draw[->, thick] (x0) edge (y1);
  \draw[->, thick] (x2) edge (y1);
  \draw[->, thick] (x0) edge (y2);
  \draw[->, thick] (x3) edge (y2);
  \draw[->, thick] (x1) edge (y3);
  \draw[->, thick] (x2) edge (y3);
  \draw[->, thick] (x1) edge (y4);
  \draw[->, thick] (x3) edge (y4);
  \draw[->, thick] (x2) edge (y5);
  \draw[->, thick] (x3) edge (y5);
  
  \draw[->, thick] (y0) edge (z0);
  \draw[->, thick] (y1) edge (z0);
  \draw[->, thick] (y3) edge (z0);
  \draw[->, thick] (y0) edge (z1);
  \draw[->, thick] (y2) edge (z1);
  \draw[->, thick] (y4) edge (z1);
  \draw[->, thick] (y1) edge (z2);
  \draw[->, thick] (y2) edge (z2);
  \draw[->, thick] (y5) edge (z2);
  \draw[->, thick] (y3) edge (z3);
  \draw[->, thick] (y4) edge (z3);
  \draw[->, thick] (y5) edge (z3);
  
  \foreach \i in {0,1,2,3}{
   
   \fill[orange, opacity = 0.4, rounded corners] (-0.2 + 3.2*\i,-0.2) -- (-0.2 + 3.2*\i,0.95) -- (0.8 + 3.2*\i,1.9) -- (1.8 + 3.2*\i,0.95) -- (1.8 + 3.2*\i,-0.2) -- cycle;
   
   \node[emptyvertex] (t\i) at (3.2*\i + 0.8,1.6) {};
   \foreach \j in {0,1,2}{
    \node[emptyvertex] (v\i\j) at (3.2*\i + 0.8*\j,0.8) {};
    \draw[->, thick] (t\i) edge (v\i\j);
    \node[emptyvertex] (w\i\j) at (3.2*\i + 0.8*\j,0) {};
   }
   \draw[->, thick] (v\i0) edge (w\i0);
   \draw[->, thick] (v\i0) edge (w\i1);
   \draw[->, thick] (v\i1) edge (w\i0);
   \draw[->, thick] (v\i1) edge (w\i2);
   \draw[->, thick] (v\i2) edge (w\i1);
   \draw[->, thick] (v\i2) edge (w\i2);
  }
  
  \draw[->, thick] (z0) edge (t0);
  \draw[->, thick] (z1) edge (t1);
  \draw[->, thick] (z2) edge (t2);
  \draw[->, thick] (z3) edge (t3);
  
  \node at (-0.8,1.2) {$(H_B,w_0)$};
  
  \foreach \v[count = \i] in {1,2,3,,4,5,6,,7,8,9,,10,11,12}{
   \node at (0.8*\i - 0.8,-0.4) {$\v{}$};
  }
 \end{tikzpicture}
 \caption{Visualization for the proof of Lemma \ref{la:combine-structure-graphs-along-block-system} with $\Omega = \{1,\dots,12\}$ and $\FB = \{B_1,B_2,B_3,B_4\}$ where $B_i = \{3i-2,3i-1,3i\}$. We attach a copy of $(H_B,w_0)$ to every leaf of $(G,v_0)$.}
 \label{fig:combine-structure-trees}
\end{figure}

\begin{proof}
 The basic idea to construct the structure graph for $\Gamma$ is to attach a copy of $(H_B,w_0)$ to every leaf of $(G_{\FB},v_0)$ (see Figure \ref{fig:combine-structure-trees}).
 Towards this end suppose $\FB = \{B_1,\dots,B_m\}$ such that $B = B_1$.
 For every $i \in [m]$ pick $\gamma_{1 \rightarrow i} \in \Gamma$ such that $B_1^{\gamma_{1 \rightarrow i}} = B_i$.
 Also define the rooted simple acyclic graph $(H_B^{i},(w_0,i))$ with vertex set
 \[V(H_B^{i}) \coloneqq \{(v,i) \mid v \in V(H_B) \setminus B\} \cup \{\alpha^{\gamma_{1 \rightarrow i}} \mid \alpha \in B\}\]
 and edge set
 \[E(H_B^{i}) \coloneqq \{((v,i),(w,i)) \mid (v,w) \in E(H_B), w \notin B\} \cup \{((v,i),\alpha^{\gamma_{1 \rightarrow i}}) \mid (v,\alpha) \in E(H_B), \alpha \in B\}.\]
 Note that $(H_B^{i},(w_0,i))$ is an isomorphic copy of $(H_B,w_0)$ with leaf set $L(H_B^{i}) = B_i$.
 Let $\psi_{i}\colon H_B \cong H_B^{i}$ be the \emph{standard isomorphism} defined by $\psi_{i}(w) \coloneqq (w,i)$ for all $w \in V(H_B) \setminus B$ and $\psi_{i}(\alpha) \coloneqq \alpha^{\gamma_{1 \rightarrow i}}$ for $\alpha \in B$.
 Also let $\psi_{i,j}\colon H_B^{i} \cong H_b^{j}$ be defined by $\psi_{i,j} \coloneqq \psi_i^{-1}\psi_j$.
 Moreover, let
 \[\varphi_B^{i}\colon \Gamma_{B_i}[B_i] \rightarrow \Aut(H_B^{i},(w_0,i))\colon \gamma \mapsto \psi_i^{-1}\left(\varphi_B(\gamma_{1 \rightarrow i}\gamma\gamma_{1 \rightarrow i}^{-1})\right)\psi_i\]
 (where $\gamma_{1 \rightarrow i}$ and $\gamma_{1 \rightarrow i}^{-1}$ are restricted in a suitable manner).
 Then $(H_B^{i},(w_0,i),\varphi_B^{i})$ is an almost $d$-ary structure graph for $\Gamma_{B_i}[B_i]$.
 
 Now define the rooted simple acyclic graph $(G,v_0)$ where
 \[V(G) = V(G_{\FB}) \uplus \bigcup_{i \in [m]} V(H_B^{i})\]
 and
 \[E(G) = E(G_{\FB}) \cup \bigcup_{i \in [m]} E(H_B^{i}) \cup \{(B_i,(w_0,i)) \mid i \in [m]\}.\]
 Clearly, $L(G) = \Omega$ and the graph $(G,v_0)$ can be computed in time polynomial in the size of $G$.
 Also, every branch $\bar v \in \Br(G,v_0)$ is the concatenation of a branch $\bar u \in \Br(G_{\FB},v_0)$ and a (possibly empty) branch $\bar w \in \Br(H_B^{i},(w_0,i))$ for the unique $i \in [m]$ such that $L(\bar u) = B_i$.
 This means
 \[|\Br(G,v_0)| \leq |\Br(G_{\FB},v_0)| \cdot |\Br(H_B,w_0)| \leq m^{k} \cdot b^{\ell} \leq n^{\max\{k,\ell\}}.\]
 
 Next define the homomorphism $\varphi\colon \Gamma \rightarrow \Aut(G,v_0)$.
 For $\gamma \in \Gamma$ first let $\sigma \in S_m$ such that $B_i^{\gamma} = B_{\sigma(i)}$ for all $i \in [m]$.
 Then, for each $i \in [m]$, there is a $\delta_i \in \Gamma_{B_i}[B_i]$ such that $\gamma[B_i] = \delta_i \gamma_{i \rightarrow \sigma(i)}$ where $\gamma_{i \rightarrow j} \coloneqq \gamma_{1 \rightarrow i}^{-1}\gamma_{1 \rightarrow j}$.
 With this, for $v \in V(G)$, define
 \[v^{\varphi(\gamma)} = \begin{cases}
                          v^{\varphi_{\FB}(\gamma[\FB])}                              &\text{if } v \in V(G_{\FB})\\
                          \psi_{i \rightarrow \sigma(i)}(v^{\varphi_B^{i}(\delta_i)}) &\text{if } v \in V(H_{B}^{i})\\
                         \end{cases}.\]
 It can be checked that $(G,v_0,\varphi)$ is a structure graph for $\Gamma$.
 It remains to check that $(G,v_0,\varphi)$ is almost $d$-ary.
 Towards this end, let $v_t \in V(G)$ such that $\deg^{+}(v_t) > d$.
 Also let $(v_0,v_1,\dots,v_t) \in \Br(G,v_0)$ be a branch of $(G,v_0)$ that terminates in $v_t$.
 It suffices to argue $(\Gamma^{\varphi})_{(v_0,v_1,\dots,v_t)}[N^{+}(v_t)]$ is semi-regular.
 If $v \in V(G_{\FB})$ this immediately follows from the fact that $(G_{\FB},v_0,\varphi_{\FB})$ is almost $d$-ary.
 Otherwise $v \in V(H_B^{i})$ for some $i \in [m]$  and the statement follows from the fact that $(H_B^{i},(w_0,i),\varphi_B^{i})$ is almost $d$-ary.
\end{proof}

\section{Isomorphism for Hypergraphs}
\label{sec:hypergraph-isomorphism}

Having introduced the tools to compute $d$-normalized group actions we can now turn to the Hypergraph Isomorphism Problem for $\mgamma_d$-groups.
For the remainder of this section, let us fix some $d \geq 2$.
We start by formally introducing the relevant problems and stating some basic properties and algorithmic tools.

\subsection{The Generalized String Isomorphism Problem}

For the purpose of designing an algorithm for the Hypergraph Isomorphism Problem it is actually more convenient to consider the following equivalent problem.
The \emph{Set-of-Strings Isomorphism Problem} takes as input two sets $\FX = \{\Fx_1,\dots,\Fx_m\}$ and $\FY = \{\Fy_1,\dots,\Fy_m\}$ where $\Fx_i,\Fy_i \colon \Omega \rightarrow \Sigma$ are strings, and a group $\Gamma \leq \Sym(\Omega)$,
and asks whether there is some $\gamma \in \Gamma$ such that $\FX^{\gamma} \coloneqq \{\Fx_1^{\gamma},\dots,\Fx_m^{\gamma}\} = \FY$.

\begin{theorem}
 The Hypergraph Isomorphism Problem for $\mgamma_d$-groups is polynomial-time equivalent to the Set-of-Strings Isomorphism Problem for $\mgamma_d$-groups under many-one reductions.
\end{theorem}

\begin{proof}
 For the forward direction a hypergraph $\CH = (V,\CE)$ can be interpreted as a set of strings $\FX = \{\Fx_E \mid E \in \CE\}$ where $\Fx_E\colon V \rightarrow \{0,1\}$ is the characteristic function of $E$, i.e., $\Fx(v) = 1$ if and only if $v \in E$.
 
 In the other direction, a set of strings $\FX = \{\Fx_1,\dots,\Fx_m\}$, where $\Fx_i\colon \Omega \rightarrow \Sigma$ is a string, can be viewed as a hypergraph with vertex set $V \coloneqq \Omega \times \Sigma$ and edge set $\CE = \{\{(\alpha,\Fx_i(\alpha)) \mid \alpha \in \Omega\} \mid i \in [m]\}$.
\end{proof}

For the remainder of this work we can thus focus on the Set-of-Strings Isomorphism Problem.
Clearly, this problem generalizes the standard String Isomorphism Problem introduced above.
The basic strategy to tackle the Set-of-Strings Isomorphism Problem is to generalize the algorithm of \cite{GroheNS18} providing an algorithm for the String Isomorphism Problem for $\mgamma_d$-groups running in time $n^{\CO((\log d)^{c})}$ for some constant $c$.
Actually, for the purpose of building a recursive algorithm, we consider a slightly more general problem that crucially allows us to modify instances in a certain way exploited later on.

Let $\Gamma \leq \Sym(\Omega)$ be a group and let $\FP$ be a partition of the set $\Omega$.
A \emph{$\FP$-string} is a pair $(P,\Fx)$ where $P \in \FP$ and $\Fx\colon P \rightarrow \Sigma$ is a string over a finite alphabet $\Sigma$.
For $\sigma \in \Sym(\Omega)$ the string $\Fx^{\sigma}$ is defined by \[\Fx^{\sigma}\colon P^{\sigma} \rightarrow \Sigma\colon \alpha \mapsto \mathfrak{x}(\alpha^{\sigma^{-1}}).\]
A permutation $\sigma \in \Sym(\Omega)$ is a \emph{$\Gamma$-isomorphism} from $(P,\Fx)$ to a second $\FP$-string $(Q,\Fy)$ if $\sigma \in \Gamma$ and $(P^{\sigma},\Fx^{\sigma}) = (Q,\Fy)$.

\begin{definition}
 The \emph{Generalized String Isomorphism Problem} takes as input a permutation group $\Gamma \leq \Sym(\Omega)$,
 a $\Gamma$-invariant partition $\FP$ of the set $\Omega$,
 and $\FP$-strings $(P_1,\Fx_1),\dots,(P_m,\Fx_m)$ and $(Q_1,\Fy_1),\dots,(Q_m,\Fy_m)$,
 and asks whether there is some $\gamma \in \Gamma$ such that
 \begin{equation}
  \{(P_1^{\gamma},\Fx_1^{\gamma}),\dots,(P_m^{\gamma},\Fx_m^{\gamma})\} = \{(Q_1,\Fy_1),\dots,(Q_m,\Fy_m)\}.
 \end{equation}
\end{definition}

I remark that for a $\FP$-string $(P,\Fx)$ one may also omit the first component since $P$ is also the domain of $\Fx$.
In particular, for the Generalized String Isomorphism Problem, there is no difference between having pairs $(P,\Fx)$ in the sets or only the strings $\Fx$.
Here, we choose to work with pairs mainly because it turns out to be more convenient on a technical level.

We usually denote $\FX = \{(P_1,\Fx_1),\dots,(P_m,\Fx_m)\}$ and $\FY = \{(Q_1,\Fy_1),\dots,(Q_m,\Fy_m)\}$.
Clearly, the Generalized String Isomorphism Problem generalizes the Set-of-Strings Isomorphism Problem by choosing $\FP$ to be the trivial partition consisting of one block.
For the rest of this section we denote by $n \coloneqq |\Omega|$ the size of the domain, and $m$ denotes the size of $\FX$ and $\FY$ (we always assume $|\FX| = |\FY|$, otherwise the problem is trivial). 
Also, we can naturally lift all notations and basic properties described for the standard String Isomorphism Problem in Subsection \ref{subsec:string-isomorphism}.

The goal of this section is to provide an algorithm solving the Generalized String Isomorphism Problem for $\mgamma_d$-groups in time $(n+m)^{\CO((\log d)^{c})}$ for some absolute constant $c$.

We start by introducing some additional notation.
For every $P \in \FP$ we define $\Str_\FX(P) \coloneqq \{\Fx \colon P \rightarrow \Sigma \mid (P,\Fx) \in \FX\}$ to be the set of strings from $\FX$ associated with the block $P$.
Also let $m_{\FX}(P) \coloneqq |\Str_\FX(P)|$.\
Note that $m = \sum_{P \in \FP} m_\FX(P)$.
We say that $\FX$ is \emph{completely occupied} if $m_\FX(P) \geq 1$ for every $P \in \FP$.
Also, we say that $\FX$ is \emph{simple} if $m_\FX(P) \leq 1$ for every $P \in \FP$.
Moreover, we say that $\FX$ is \emph{balanced} if $m_\FX(P) = m_\FX(P')$ for every $P,P' \in \FP$.

First, we will assume throughout this work that all sets of $\FP$-strings encountered are completely occupied, also if not explicitly stated.
Note that an instance, which is not completely occupied, can be easily turned into one, that is completely occupied, by introducing additional $\FP$-strings.
Also note that the Generalized String Isomorphism Problem for simple sets of $\FP$-strings can be easily reduced to the standard String Isomorphism Problem.

Moreover, by employing an algorithm for the String Isomorphism Problem as a subroutine, one can achieve some balancing condition.
For a set $A \subseteq \Omega$ and a set of $\FP$-strings $\FX$ define
\[\FX[A] \coloneqq \{(P \cap A,\Fx[A \cap P]) \mid (P,\Fx) \in \FX, P \cap A \neq \emptyset\}.\]
Observe that $\FX[A]$ is a set of $\FP[A]$-strings.

\begin{lemma}
 \label{la:balance-orbits}
 Let $(\Gamma,\FP,\FX,\FY)$ be an instance of the Generalized String Isomorphism Problem where $\Gamma \in \mgamma_d$.
 There is an algorithm ${\sf BalanceOrbits}(\Gamma,\FP,\FX,\FY)$ running in time $n^{\CO((\log d)^{c})}$ for some constant $c$
 that computes a group $\Delta \leq \Gamma$ and $\delta \in \Gamma$ such that $\FX[A]$ is balanced for every orbit $A$ of $\Delta$ and $\Iso_\Gamma(\FX,\FY) \subseteq \Delta\delta$.
 
 Moreover, the output $\Delta\delta$ is isomorphism-invariant in the sense that, if $\gamma_1 \in \Iso_\Gamma(\FX,\FX')$, $\gamma_2 \in \Iso_\Gamma(\FY,\FY')$ and $\Delta'\delta'$ is the output of ${\sf BalanceOrbits}(\Gamma,\FP,\FX',\FY')$, then $\gamma_1^{-1}\Delta\delta\gamma_2 = \Delta'\delta'$.
\end{lemma}

\begin{proof}
 Let $\Iso_\Gamma(\FX,\FY) \subseteq \Delta\delta \subseteq \Gamma$ and suppose there is some orbit $A$ of $\Delta$ such that $\FX[A]$ is not balanced.
 Define the strings
 \[\Fx\colon \FP[A] \rightarrow [m]\colon P \mapsto m_{\FX[A]}(P)\]
 and
 \[\Fy\colon \FP[A^{\delta}] \rightarrow [m]\colon P \mapsto m_{\FY[A^{\delta}]}(P)\]
 and compute
 \[\Delta'\delta' \coloneqq \{\gamma \in \Delta\delta \mid \gamma[\FP[A]] \in \Iso(\Fx,\Fy)\}\]
 which boils down to solving an instance of the String Isomorphism Problem for $\mgamma_d$-groups which can be solved in time $n^{\CO((\log d)^{c})}$ for some constant $c$ by Theorem \ref{thm:string-isomorphism-gamma-d}.
 Note that $\Delta' < \Delta$.
 Since any sequence of subgroups can have length at most $n \log n$, after $n \log n$ many iterations, the algorithm terminates.

 The isomorphism-invariance of the output can be easily proved. Note that the order in which orbits $A$ are considered has no influence on the output of the algorithm.
\end{proof}

\subsection{Affected Orbits}

The main hurdle in generalizing the algorithm from \cite{GroheNS18} for the String Isomorphism Problem for $\mgamma_d$-groups to the setting of hypergraphs is an extension of the Local Certificates Routine originally introduced by Babai for his quasipolynomial time isomorphism test \cite{Babai16}.
Similar to \cite{GroheNS18}, our algorithm exploits a variant of the Unaffected Stabilizers Theorem for $\mgamma_d$-group which builds the theoretical foundation for the Local Certificates algorithm.

Recall that for a set $M$ we denote by $\Alt(M)$ the alternating group acting with its standard action on the set $M$.
Moreover, following Babai \cite{Babai16}, we refer to the groups $\Alt(M)$ and $\Sym(M)$ as the \emph{giants} where $M$ is an arbitrary finite set.
Let $\Gamma \leq \Sym(\Omega)$.
A \emph{giant representation} is a homomorphism $\varphi\colon \Gamma \rightarrow S_k$ such that $\Gamma^{\varphi} \geq A_k$.

Given a set $\FX$ of $\FP$-string, a group $\Gamma \leq \Sym(\Omega)$ and a giant representation $\varphi\colon \Gamma \rightarrow S_k$, the aim of the Local Certificates Routine is to determine whether $(\Aut_\Gamma(\FX)^{\varphi}) \geq A_k$ and to compute a meaningful certificate in both cases.
To achieve this goal the central tool is to split the set $\Omega$ into \emph{affected} and \emph{non-affected} points.

\begin{definition}[Affected Points, Babai \cite{Babai16}]
 \label{def:affected-points}
 Let $\Gamma \leq \Sym(\Omega)$ be a group and $\varphi\colon\Gamma \rightarrow S_k$ a giant representation.
 Then an element $\alpha \in \Omega$ is \emph{affected by $\varphi$} if $\Gamma_\alpha^{\varphi} \not\geq A_k$.
\end{definition}

\begin{remark}
 \label{rem:affected-orbit}
 Let $\varphi\colon \Gamma \rightarrow S_k$ be a giant representation and suppose $\alpha \in \Omega$ is affected by $\varphi$.
 Then every element in the orbit $\alpha^{\Gamma}$ is affected by $\varphi$.
 The set $\alpha^\Gamma$ is called an \emph{affected orbit} (with respect to $\varphi$).
\end{remark}

The correctness and the analysis of the running time of the Local Certificates Routine rest on the following two statements.
Recall the definition of a $d$-normalized permutation group (see Definition \ref{def:d-normalized}).

\begin{theorem}[{\cite[Theorem V.3]{GroheNS18}}]
 \label{thm:unaffected-stabilizer-tree}
 Let $\Gamma \leq \Sym(\Omega)$ be a $d$-normalized permutation group.
 Furthermore let $k > \max\{8,2 + \log_2d\}$ and $\varphi\colon \Gamma \rightarrow S_k$ be a giant representation.
 Let $D \subseteq \Omega$ be the set of elements not affected by $\varphi$.
 Then $\Gamma_{(D)}^{\varphi} \geq A_k$.
\end{theorem}

\begin{lemma}[{\cite[Theorem 6(b)]{Babai16}}]
 \label{thm:kernel-affected-orbits}
 Let $\Gamma \leq \Sym(\Omega)$ be a permutation group and suppose $\varphi\colon \Gamma \rightarrow S_k$ is a giant representation for $k \geq 5$.
 Suppose $A \subseteq \Omega$ is an affected orbit of $\Gamma$ (with respect to $\varphi$). Then every orbit of $\ker(\varphi)$ in $A$ has length at most $|A|/k$.
\end{lemma}

\subsection{Computing a Normalized Group Action}

Since Theorem \ref{thm:unaffected-stabilizer-tree} is only applicable to $d$-normalized groups, as the first step of our main algorithm, we invoke a normalization procedure that produces an equivalent instance of the Generalized String Isomorphism Problem where the input group is $d$-normalized.
Also, throughout the execution of the algorithm, even if the input group is $d$-normalized, we still need to recursively solve instances that are not $d$-normalized (i.e., the corresponding permutation group is not $d$-normalized).
Hence, we constantly need to \emph{renormalize} throughout the execution of the algorithm.
Here, it turns that the initial normalization procedure is too expensive, i.e., the updated instances created by this procedure are too large.
For this reason, we design a specialized renormalization procedure.
The basic idea of this subroutine is to exploit the fact that all permutation groups encountered during our main algorithm, while not necessarily being $d$-normalized, are ``close'' to being $d$-normalized.
This allows us to renormalize instances more efficiently.

Both the normalization and renormalization procedure rely on the tools described in Section \ref{sec:normalization-framework}.
In the next theorem, we first give the normalization procedure invoked at the beginning of our main algorithm.
Recall the definition of the function $\funcnorm$ (see Lemma \ref{la:construct-structure-graph-transitive}).

\begin{theorem}
 \label{thm:normalize-generalized-string-isomorphism-instance}
 Let $(\Gamma,\FP,\FX,\FY)$ be an instance of the Generalized String Isomorphism Problem where $\Gamma \leq \Sym(\Omega)$ is a $\mgamma_d$-group.
 
 Then there is a set $\Omega^{*}$,
 a monomorphism $\varphi\colon \Gamma \rightarrow \Sym(\Omega^{*})$,
 a sequence of partitions $\{\Omega^{*}\} = \FB_0^{*} \succ \FB_1^{*} \succ \dots \succ \FB_k^{*} = \{\{\alpha\} \mid \alpha \in \Omega^{*}\}$,
 and an instance $(\Gamma^{\varphi},\FP^{*},\FX^{*},\FY^{*})$ of the Generalized String Isomorphism Problem
 such that the following properties are satisfied:
 \begin{enumerate}
  \item $|\Omega^{*}| \leq n^{\funcnorm(d)+1}$,
  \item the sequence $\FB_0^{*} \succ \dots \succ \FB_k^{*}$ forms an almost $d$-ary sequence of $\Gamma^{\varphi}$-invariant partitions,
  \item there is some $i \in [k]$ such that $\FP^{*} = \FB_{i}^{*}$, and
  \item $\gamma \in \Iso_\Gamma(\FX,\FY)$ if and only if $\varphi(\gamma) \in \Iso_{\Gamma^{\varphi}}(\FX^{*},\FY^{*})$ for every $\gamma \in \Gamma$.
 \end{enumerate}
 Moreover, given $\Gamma \leq \Sym(\Omega)$, there is an algorithm computing the desired objects in time polynomial in the input and the size of $\Omega^{*}$.
\end{theorem}

\begin{proof}
 For the normalization we build on the framework established in Section \ref{sec:normalization-framework}.
 First, for every $P \in \FP$ we construct a structure graph $(G_P,v_0^{P},\psi_P)$ for $\Gamma_P[P]$ using Lemma \ref{la:construct-structure-graph}.
 Next, for every orbit $A$ of $\Gamma[\FP]$ we construct a structure graph $(H_A,w_0^{A},\tau_A)$ again using Lemma \ref{la:construct-structure-graph}.
 Using Lemma \ref{la:combine-structure-graphs-along-block-system} one can build a structure graph for $\Gamma[\bigcup_{P \in A} P]$ for every orbit $A$ of $\Gamma[\FP]$.
 Combining these structure graphs into a single structure graph $(G,v_0,\psi)$ adding a new root connected to the roots of the structure graphs for $\Gamma[\bigcup_{P \in A} P]$.
 
 Overall, it holds that $\Br(G,v_0) \leq n^{\funcnorm(d) + 1}$ and $G$ can be computed in time polynomial in the size of $G$.
 Additionally, for every $P \in \FP$ there is a vertex $v_P \in V(G)$ such that $L(G,v_P) = P$.
 
 Now let $\Omega^{*} \coloneqq \Br^{*}(G,v_0)$ and let $\varphi\colon\Gamma\rightarrow\Sym(\Br^{*}(G,v_0))$ be the standard action of $\Gamma$ on the set $\Br^{*}(G,v_0)$.
 Let $\Gamma^{*} \coloneqq \Gamma^{\psi}$.
 Then $\Unf(G,v_0)$ is an almost $d$-ary structure tree for $\Gamma^{*}$ by Lemma \ref{la:standard-action-on-branches} giving an almost $d$-ary sequence of partitions $\{\Omega^{*}\} = \FB_0^{*} \succ \FB_1^{*} \succ \dots \succ \FB_k^{*} = \{\{\alpha\} \mid \alpha \in \Omega^{*}\}$.
 Also, define
 \[f\colon \Omega^{*} \rightarrow \Omega\colon \bar v \mapsto L(\bar v).\]
 Let $\alpha \in \Omega$ and $\gamma \in \Gamma$.
 Then
 \begin{align*}
  (f(v_1,\dots,v_q))^{\gamma} = v_q^{\gamma} = L((v_1,\dots,v_q)^{\varphi(\gamma)}) = f((v_1,\dots,v_q)^{\varphi(\gamma)})  
 \end{align*}
 using Lemma \ref{la:standard-action-on-branches}.
 Now let
 \[M \coloneqq \{(v_1,\dots,v_\ell) \in \Br(G) \mid \exists P \in \FP\colon v_\ell = v_P\}.\]
 Also define
 \[\FP^{*} \coloneqq \{\{(v_1,\dots,v_q) \in \Br^{*}(G) \mid (v_1,\dots,v_\ell) = (w_1,\dots,w_\ell) \} \mid (w_1,\dots,w_\ell) \in M\}.\]
 Clearly, $\FP^{*} = \FB_i^{*}$ for some $i \in [k]$.
 Moreover, for every $P^{*} \in \FP^{*}$ it holds that $f(P^{*}) \in \FP$.
 Finally, let
 \[\FX^{*} \coloneqq \{(P^{*},\Fx^{*}\colon P^{*} \rightarrow \Sigma\colon \bar v \mapsto \Fx(L(\bar v)) \mid P^{*} \in \FP^{*}, (f(P^{*}),\Fx) \in \FX\}\]
 and similarly
 \[\FY^{*} \coloneqq \{(Q^{*},\Fy^{*}\colon Q^{*} \rightarrow \Sigma\colon \bar v \mapsto \Fy(L(\bar v)) \mid Q^{*} \in \FP^{*}, (f(Q^{*}),\Fy) \in \FY\}\]
 
 To finish the proof it remains to show that $\gamma \in \Iso_\Gamma(\FX,\FY)$ if and only if $\varphi(\gamma) \in \Iso_{\Gamma^{\varphi}}(\FX^{*},\FY^{*})$ for every $\gamma \in \Gamma$.
 First suppose $\gamma \in \Iso_\Gamma(\FX,\FY)$ and let $(P,\Fx) \in \FX$.
 Let $P^{*} \in \FP^{*}$ such that $f(P^{*}) = P$ and $\Fx^{*} \colon P^{*} \rightarrow \Sigma\colon \alpha^{*} \mapsto \Fx(f(\alpha^{*}))$.
 Let $Q^{*} \coloneqq (P^{*})^{\varphi(\gamma)}$.
 Note that $Q^{*} \in \FP^{*}$ since $\FP^{*}$ is $\Gamma^{*}$-invariant.
 Let $Q = f(Q^{*}) \in \FP$.
 Then
 \begin{align*}
  (P^{*},\Fx^{*})^{\varphi(\gamma)} &= \left(Q^{*},Q^{*} \rightarrow \Sigma\colon \alpha^{*} \mapsto \Fx(f((\alpha^{*})^{\varphi(\gamma)^{-1}}))\right) \\
                                    &= \left(Q^{*},Q^{*} \rightarrow \Sigma\colon \alpha^{*} \mapsto \Fx((f(\alpha^{*}))^{\gamma^{-1}})\right) \\
                                    &= \left(Q^{*},Q^{*} \rightarrow \Sigma\colon \alpha^{*} \mapsto \Fx^{\gamma}(f(\alpha^{*}))\right) \in \FY^{*}
 \end{align*}
 
 For the other direction suppose that $\varphi(\gamma) \in \Iso_{\Gamma^{*}}(\FX^{*},\FY^{*})$ and let $(P,\Fx) \in \FX$.
 Let $P^{*} \in \FP^{*}$ such that $f(P^{*}) = P$ and $\Fx^{*} \colon P^{*} \rightarrow \Sigma\colon \alpha^{*} \mapsto \Fx(f(\alpha^{*}))$.
 Then $(Q^{*},\Fy^{*}) \coloneqq (P^{*},\Fx^{*})^{\varphi(\gamma)} \in \FY^{*}$.
 Let $Q = f(Q^{*}) \in \FP$ and let $(Q,\Fy) \in \FY$ such that $\Fy^{*} \colon Q^{*} \rightarrow \Sigma\colon \alpha^{*} \mapsto \Fy(f(\alpha^{*}))$.
 Then $P^{\gamma} = f(P^{*})^{\gamma} = f((P^{*})^{\varphi(\gamma)}) = f(Q^{*}) = Q$.
 Let $\alpha \in P$ and let $\alpha^{*} \in f^{-1}(\alpha)$.
 Then
 \[\Fx(\alpha^{\gamma^{-1}}) = \Fx^{*}((\alpha^{*})^{\varphi(\gamma)^{-1}}) = \Fy^{*}(\alpha^{*}) = \Fy(\alpha).\]
 Hence, $\gamma \in \Iso_\Gamma(\Fx,\Fy)$.
\end{proof}

The next theorem describes the renormalization procedure.

\begin{theorem}
 \label{thm:renormalize-structure-tree}
 Let $\FP$ be a partition of $\Omega$ and $\Gamma \leq \Sym(\Omega)$ be a $\mgamma_d$-group.
 Also suppose $\{\Omega\} = \FB_0 \succ \FB_1 \succ \dots \succ \FB_\ell = \{\{\alpha\} \mid \alpha \in \Omega\}$ is a sequence of $\Gamma$-invariant partitions.
 Let $j \in [\ell]$ such that $\FB_j = \FP$ and assume that, for every $j \neq i \in [\ell]$ and every $B \in \FB_{i-1}$ it holds that $|\FB_i[B]| \leq d$ or $\Gamma_B[\FB_i[B]]$ is semi-regular.
 
 Then there is a set $\Omega^{*}$,
 a partition $\FP^{*}$,
 a monomorphism $\varphi\colon \Gamma \rightarrow \Sym(\Omega^{*})$,
 a mapping $f\colon \Omega^{*} \rightarrow \Omega$,
 a sequence of partitions $\{\Omega^{*}\} = \FB_0^{*} \succ \FB_1^{*} \succ \dots \succ \FB_k^{*} = \{\{\alpha\} \mid \alpha \in \Omega^{*}\}$,
 and a mapping $f'\colon \FP^{*} \rightarrow \FP$,
 such that the following properties are satisfied:
 \begin{enumerate}[label=(\Roman*)]
  \item\label{item:renormalize-structure-tree-1} the sequence $\FB_0^{*} \succ \dots \succ \FB_k^{*}$ forms an almost $d$-ary sequence of $\Gamma^{\varphi}$-invariant partitions,
  \item\label{item:renormalize-structure-tree-2} $(f(\alpha^{*}))^{\gamma} = f((\alpha^{*})^{\varphi(\gamma)})$ for every $\alpha^{*} \in \Omega^{*}$ and $\gamma \in \Gamma$,
  \item\label{item:renormalize-structure-tree-3} $\{f(\alpha^{*}) \mid \alpha^{*} \in P^{*}\} = f'(P^{*})$ for every $P^{*} \in \FP^{*}$,
  \item\label{item:renormalize-structure-tree-4} $|(f')^{-1}(P)| \leq |A|^{\funcnorm(d)}$ where $A$ is the orbit of $\Gamma_B[\FP[B]]$ containing $P$ where $P \subseteq B \in \FB_{j-1}$,
  \item\label{item:renormalize-structure-tree-5} $|P^{*}| = |f'(P^{*})|$ for every $P^{*} \in \FP^{*}$,
  \item\label{item:renormalize-structure-tree-6} there is some $i \in [k]$ such that $\FP^{*} = \FB_{i}^{*}$,
 \end{enumerate}
 Additionally, for every $\FX,\FY$ sets of $\FP$-strings, there are $\FX^{*},\FY^{*}$ sets of $\FP^{*}$-strings such that
 \begin{enumerate}[resume,label=(\Roman*)]
  \item\label{item:renormalize-structure-tree-9} $\FX^{*} = \{(P^{*},P^{*} \rightarrow \Sigma\colon \alpha^{*} \mapsto \Fx(f(\alpha^{*}))) \mid (f'(P^{*}),\Fx) \in \FX\}$ and
   $\FY^{*} = \{(Q^{*},Q^{*} \rightarrow \Sigma\colon \alpha^{*} \mapsto \Fy(f(\alpha^{*}))) \mid (f'(Q^{*}),\Fy) \in \FY\}$,
  \item\label{item:renormalize-structure-tree-7} $\gamma \in \Iso_\Gamma(\FX,\FY)$ if and only if $\varphi(\gamma) \in \Iso_{\Gamma^{\varphi}}(\FX^{*},\FY^{*})$ for every $\gamma \in \Gamma$, and
  \item\label{item:renormalize-structure-tree-8} for every $\Gamma$-invariant window $W \subseteq \Omega$ such that $\Gamma[W] \leq \Aut_\Gamma(\FX[W])$ it holds $\Gamma^{\varphi}[f^{-1}(W)] \leq \Aut_{\Gamma^{\varphi}}(\FX^{*}[f^{-1}(W)])$.
 \end{enumerate}
 Moreover, given $\Gamma \leq \Sym(\Omega)$, there is an algorithm ${\sf Renormalize}(\Gamma,\FP,\FX,\FY,\FB_0,\dots,\FB_\ell)$ computing the desired objects in time polynomial in the input and the size of $\Omega^{*}$.
\end{theorem}

Before diving into the proof, let us give some intuition on the theorem.
The input is an instance $(\Gamma,\FP,\FX,\FY)$ of the Generalized String Isomorphism Problem together with a sequence of $\Gamma$-invariant partitions $\{\Omega\} = \FB_0 \succ \FB_1 \succ \dots \succ \FB_\ell = \{\{\alpha\} \mid \alpha \in \Omega\}$ that is ``close'' to being almost $d$-ary, i.e., Definition \ref{def:d-normalized} is only violated on level $j$.
Now, instead of constructing a new structure graph for the entire group $\Gamma$, we only implant structure graphs at those places where Definition \ref{def:d-normalized} is violated.
A visualization is given in Figure \ref{fig:renormalize}.
This way, the domain $\Omega^*$ of the resulting $d$-normalized group $\Gamma^*$ is much smaller compared to the application of Theorem \ref{thm:normalize-generalized-string-isomorphism-instance}.
In the theorem this is reflected by Items \ref{item:renormalize-structure-tree-4} and \ref{item:renormalize-structure-tree-5} which together give an upper bound on the size of $\Omega^*$ that is typically much better than the corresponding bound in Theorem \ref{thm:normalize-generalized-string-isomorphism-instance}.

\begin{figure}
 \centering
 \begin{tikzpicture}
  \draw[thick,red] (0.55,2.55) rectangle (7.05,3.85);
  
  \node[emptyvertex] (r) at (3.8,4.6) {};
  \node[emptyvertex] (r-1) at (1.8,3.6) {};
  \node[emptyvertex] (r-2) at (5.8,3.6) {};
  
  \draw[lightcolor2] (-0.2,-0.2) rectangle (1.8,1.2);
  \draw[lightcolor3] (1.8,-0.2) rectangle (3.8,1.2);
  \draw[lightcolor4] (3.8,-0.2) rectangle (5.8,1.2);
  \draw[lightcolor5] (5.8,-0.2) rectangle (7.8,1.2);
  
  \draw[thick] (-0.2,-0.2) -- (7.8,-0.2);
  \draw[thick] (-0.2,1.2) -- (7.8,1.2);
  \foreach \i in {0,1,2,3,4}{
   \draw[thick] (-0.2 + 2*\i,-0.2) -- (-0.2 + 2*\i,1.2);
  }
  
  \node at (8.0,2.8) {$\FB_j = \FP$};
  \node at (-0.6,0.5) {$\FX$};
  
  \foreach \i in {0,1,2,3}{
   \draw (2*\i,0.0) -- (2*\i + 1.6,0.0);
   \draw (2*\i,0.4) -- (2*\i + 1.6,0.4);
   \draw (2*\i,0.6) -- (2*\i + 1.6,0.6);
   \draw (2*\i,1.0) -- (2*\i + 1.6,1.0);
   
   \foreach \j in {0,1,2,3}{
    \node[emptyvertex] (v\i\j) at (2*\i + 0.4*\j + 0.2,1.6) {};
   }
   \node[emptyvertex] (w\i-0) at (2*\i + 0.6 + 0.2,2.8) {};
   \node[emptyvertex] (w\i-1) at (2*\i + 0.2 + 0.2,2.2) {};
   \node[emptyvertex] (w\i-2) at (2*\i + 1.0 + 0.2,2.2) {};
   
   \draw[->, thick] (w\i-0) edge (w\i-1);
   \draw[->, thick] (w\i-0) edge (w\i-2);
   \draw[->, thick] (w\i-1) edge (v\i0);
   \draw[->, thick] (w\i-1) edge (v\i1);
   \draw[->, thick] (w\i-2) edge (v\i2);
   \draw[->, thick] (w\i-2) edge (v\i3);
   
   \foreach \j in {0,1,2,3,4}{
    \draw (2*\i + 0.4*\j,0.0) -- (2*\i + 0.4*\j,0.4);
    \draw (2*\i + 0.4*\j,0.6) -- (2*\i + 0.4*\j,1.0);
   }
  }
  
  \draw[->, thick] (r-1) edge (w0-0);
  \draw[->, thick] (r-1) edge (w1-0);
  \draw[->, thick] (r-2) edge (w2-0);
  \draw[->, thick] (r-2) edge (w3-0);
  \draw[->, thick] (r) edge (r-1);
  \draw[->, thick] (r) edge (r-2);
  
  \foreach \a [count=\i] in {a,b,b,a,,a,a,b,a,,b,b,a,a,,b,b,b,a}{
   \node at (-0.2 + \i*0.4,0.2) {\scriptsize $\a$};
  }
  \foreach \a [count=\i] in {b,a,b,a,,a,b,a,b,,a,a,a,b,,b,a,b,b}{
   \node at (-0.2 + \i*0.4,0.8) {\scriptsize $\a$};
  }
  
 \begin{scope}[yshift = -6.2cm]
 \draw[thick,red] (0.55,2.55) rectangle (7.05,4.45);
 \draw[gray!40, fill=gray!40, rounded corners] (0.6,2.6) rectangle (3.0,4.4);
 \draw[gray!40, fill=gray!40, rounded corners] (4.6,2.6) rectangle (7.0,4.4);
 
 \node[emptyvertex] (r) at (3.8,5.2) {};
 \node[emptyvertex] (r-1) at (1.8,4.2) {};
 \node[emptyvertex] (r-2) at (5.8,4.2) {};
 
 \node[emptyvertex] (s11) at (1.1,3.5) {};
 \node[emptyvertex] (s12) at (2.5,3.5) {};
 \node[emptyvertex] (s21) at (5.1,3.5) {};
 \node[emptyvertex] (s22) at (6.5,3.5) {};
 
 \draw[lightcolor2] (-0.2,-0.2) rectangle (1.8,1.2);
 \draw[lightcolor3] (1.8,-0.2) rectangle (3.8,1.2);
 \draw[lightcolor4] (3.8,-0.2) rectangle (5.8,1.2);
 \draw[lightcolor5] (5.8,-0.2) rectangle (7.8,1.2);
 
 \draw[thick] (-0.2,-0.2) -- (7.8,-0.2);
 \draw[thick] (-0.2,1.2) -- (7.8,1.2);
 \foreach \i in {0,1,2,3,4}{
  \draw[thick] (-0.2 + 2*\i,-0.2) -- (-0.2 + 2*\i,1.2);
 }
 
 \node at (-0.6,0.5) {$\FX$};
 
 \foreach \i in {0,1,2,3}{
  \draw (2*\i,0.0) -- (2*\i + 1.6,0.0);
  \draw (2*\i,0.4) -- (2*\i + 1.6,0.4);
  \draw (2*\i,0.6) -- (2*\i + 1.6,0.6);
  \draw (2*\i,1.0) -- (2*\i + 1.6,1.0);
  
  \foreach \j in {0,1,2,3}{
   \node[emptyvertex] (v\i\j) at (2*\i + 0.4*\j + 0.2,1.6) {};
  }
  \node[emptyvertex] (w\i-0) at (2*\i + 0.6 + 0.2,2.8) {};
  \node[emptyvertex] (w\i-1) at (2*\i + 0.2 + 0.2,2.2) {};
  \node[emptyvertex] (w\i-2) at (2*\i + 1.0 + 0.2,2.2) {};
  
  \draw[->, thick] (w\i-0) edge (w\i-1);
  \draw[->, thick] (w\i-0) edge (w\i-2);
  \draw[->, thick] (w\i-1) edge (v\i0);
  \draw[->, thick] (w\i-1) edge (v\i1);
  \draw[->, thick] (w\i-2) edge (v\i2);
  \draw[->, thick] (w\i-2) edge (v\i3);
  
  \foreach \j in {0,1,2,3,4}{
   \draw (2*\i + 0.4*\j,0.0) -- (2*\i + 0.4*\j,0.4);
   \draw (2*\i + 0.4*\j,0.6) -- (2*\i + 0.4*\j,1.0);
  }
 }
 
 \draw[->, thick] (s11) edge (w0-0);
 \draw[->, thick] (s12) edge (w1-0);
 \draw[->, thick] (s21) edge (w2-0);
 \draw[->, thick] (s22) edge (w3-0);
 \draw[->, thick] (s11) edge (w1-0);
 \draw[->, thick] (s12) edge (w0-0);
 \draw[->, thick] (s21) edge (w3-0);
 \draw[->, thick] (s22) edge (w2-0);
 \draw[->, thick] (r-1) edge (s11);
 \draw[->, thick] (r-1) edge (s12);
 \draw[->, thick] (r-2) edge (s21);
 \draw[->, thick] (r-2) edge (s22);
 \draw[->, thick] (r) edge (r-1);
 \draw[->, thick] (r) edge (r-2);
 
 \foreach \a [count=\i] in {a,b,b,a,,a,a,b,a,,b,b,a,a,,b,b,b,a}{
  \node at (-0.2 + \i*0.4,0.2) {\scriptsize $\a$};
 }
 \foreach \a [count=\i] in {b,a,b,a,,a,b,a,b,,a,a,a,b,,b,a,b,b}{
  \node at (-0.2 + \i*0.4,0.8) {\scriptsize $\a$};
 }
 \end{scope}
 
 \begin{scope}[yshift = -12.3cm]
  \draw[thick,red] (0.55,4.85) -- (8.8,4.85);
  \draw[thick,red] (0.55,2.55) -- (0.55,4.85);
  \draw[thick,red] (0.55,2.55) -- (8.8,2.55);
  \draw[gray!40, fill=gray!40, rounded corners] (0.6,2.6) rectangle (7.0,4.8);
  \node at (8.4,3.7) {$\dots$}; 
  
  \node[emptyvertex] (r) at (3.8,4.6) {};
  \draw[->, thick] (5.6,5.1) -- (r);
  \node[emptyvertex] (r-1) at (1.8,3.6) {};
  \node[emptyvertex] (r-2) at (5.8,3.6) {};
  
  \draw[lightcolor2] (-0.2,-0.2) rectangle (1.8,1.2);
  \draw[lightcolor3] (1.8,-0.2) rectangle (3.8,1.2);
  \draw[lightcolor2] (3.8,-0.2) rectangle (5.8,1.2);
  \draw[lightcolor3] (5.8,-0.2) rectangle (7.8,1.2);
  
  \draw[thick] (-0.2,-0.2) -- (8.8,-0.2);
  \draw[thick] (-0.2,1.2) -- (8.8,1.2);
  \foreach \i in {0,1,2,3,4}{
   \draw[thick] (-0.2 + 2*\i,-0.2) -- (-0.2 + 2*\i,1.2);
  }
  \node at (8.4,0.5) {$\dots$}; 
  
  \node at (8.0,2.8) {$\FB_i^* = \FP^*$};
  \node at (-0.6,0.5) {$\FX^*$};
  
  \foreach \i in {0,1,2,3}{
   \draw (2*\i,0.0) -- (2*\i + 1.6,0.0);
   \draw (2*\i,0.4) -- (2*\i + 1.6,0.4);
   \draw (2*\i,0.6) -- (2*\i + 1.6,0.6);
   \draw (2*\i,1.0) -- (2*\i + 1.6,1.0);
   
   \foreach \j in {0,1,2,3}{
    \node[emptyvertex] (v\i\j) at (2*\i + 0.4*\j + 0.2,1.6) {};
   }
   \node[emptyvertex] (w\i-0) at (2*\i + 0.6 + 0.2,2.8) {};
   \node[emptyvertex] (w\i-1) at (2*\i + 0.2 + 0.2,2.2) {};
   \node[emptyvertex] (w\i-2) at (2*\i + 1.0 + 0.2,2.2) {};
   
   \draw[->, thick] (w\i-0) edge (w\i-1);
   \draw[->, thick] (w\i-0) edge (w\i-2);
   \draw[->, thick] (w\i-1) edge (v\i0);
   \draw[->, thick] (w\i-1) edge (v\i1);
   \draw[->, thick] (w\i-2) edge (v\i2);
   \draw[->, thick] (w\i-2) edge (v\i3);
   
   \foreach \j in {0,1,2,3,4}{
    \draw (2*\i + 0.4*\j,0.0) -- (2*\i + 0.4*\j,0.4);
    \draw (2*\i + 0.4*\j,0.6) -- (2*\i + 0.4*\j,1.0);
   }
  }
  
  \draw[->, thick] (r-1) edge (w0-0);
  \draw[->, thick] (r-1) edge (w1-0);
  \draw[->, thick] (r-2) edge (w2-0);
  \draw[->, thick] (r-2) edge (w3-0);
  \draw[->, thick] (r) edge (r-1);
  \draw[->, thick] (r) edge (r-2);
  
  \foreach \a [count=\i] in {a,b,b,a,,a,a,b,a,,a,b,b,a,,a,a,b,a}{
   \node at (-0.2 + \i*0.4,0.2) {\scriptsize $\a$};
  }
  \foreach \a [count=\i] in {b,a,b,a,,a,b,a,b,,b,a,b,a,,a,b,a,b}{
   \node at (-0.2 + \i*0.4,0.8) {\scriptsize $\a$};
  }
 \end{scope}
 \end{tikzpicture}
 \caption{Visualization for Theorem \ref{thm:renormalize-structure-tree}.
  The top part shows the input structure tree together with a set $\FX$ of $\FP$-strings.
  The structure tree is almost $d$-ary expect for level $j$ indicated by the red area (for visualization purposes $|\FB_j[B]| = 2$ for all $B \in \FB_{j-1}$).
  To obtain a $d$-normalized action, we first implant structure graphs in the problematic places, as shown in the middle part.
  Afterwards, we obtain a $d$-normalized action by unfolding the resulting structure graph.
  The bottom part shows the left side of the structure tree obtained after the unfolding.}
 \label{fig:renormalize}
\end{figure}

\begin{proof}
 For every $P \in \FP$ the sequence $\FB_j[P] \succeq \dots \succeq \FB_\ell[P]$ forms an almost $d$-ary sequence of partitions for $\Gamma_P[P]$ (cf.\ Observation \ref{obs:sequence-of-partitions}).
 Hence, interpreting this sequence of partitions as a structure tree gives, for each $P \in \FP$, a structure graph $(G_P,v_0^{P},\psi_P)$ for the group $\Gamma_P[P]$.
 
 Now let $B \in \FB_{j-1}$ and let $A$ be an orbit of $\Delta(B) \coloneqq \Gamma_B[\FP[B]]$.
 Note that $\Delta(B)[A] \in \mgamma_d$ by Lemma \ref{la:gamma-d-closure}.
 Hence, by Lemma \ref{la:construct-structure-graph-transitive}, there is an almost $d$-ary structure graph $(H_{B,A},w_0^{B,A},\tau_{B,A})$ for the group $\Delta(B)[A]$ such that $\Br(H_{B,A},w_0^{B,A}) \leq |A|^{\funcnorm(d) + 1}$.
 
 \begin{claim}
  \label{claim:number-of-branches}
  For every $P \in A$ it holds that $|\{\bar v \in \Br^{*}(H_{B,A},w_0^{B,A}) \mid L(\bar v) = P\}| \leq |A|^{\funcnorm(d)}$.
 \end{claim}
 \begin{claimproof}
  Since $\Delta(B)[A]$ is transitive it holds that
  \[|\{\bar v \in \Br^{*}(H_{B,A},w_0^{B,A}) \mid L(\bar v) = P\}| = |\{\bar v \in \Br^{*}(H_{B,A},w_0^{B,A}) \mid L(\bar v) = P'\}|\]
  for all $P,P' \in A$.
  Hence,
  \[|\{\bar v \in \Br^{*}(H_{B,A},w_0^{B,A}) \mid L(\bar v) = P\}| \leq |\Br(H_{B,A},w_0^{B,A})|/|A| \leq |A|^{\funcnorm(d)}.\]
 \end{claimproof}
 
 Next, applying Lemma \ref{la:combine-structure-graphs-along-block-system} to the structure graphs $(H_{B,A},w_0^{B,A},\tau_{B,A})$ and $(G_P,v_0^{P},\psi_P)$, $P \in A$, results in an almost $d$-ary structure graph $(H'_{B,A},w_0^{B,A},\psi_{B,A})$ for the group $\Gamma_B[\bigcup_{P \in A} P]$.
 Moreover, adding a new root vertex connected to the root vertices $w_0^{B,A}$ of $(H'_{B,A},w_0^{B,A},\psi_{B,A})$, where $A$ is an orbit of $\Delta(B)$, gives an almost $d$-ary structure graph $(G_B,v_0^{B},\psi_B)$ of $\Gamma_B[B]$ for every $B \in \FB_{j-1}$.
 
 In order to build a structure graph for the complete group $\Gamma$ consider the group $\Gamma[\FB_{j-1}]$ and let $A'$ be an orbit of $\Gamma[\FB_{j-1}]$.
 Note that $(\Gamma[\FB_{j-1}])[A']$ has an almost $d$-ary structure tree which can be viewed as an almost $d$-ary structure graph $(H_{A'},u_0^{A'},\tau_{A'})$.
 Again applying Lemma \ref{la:combine-structure-graphs-along-block-system} to the structure graphs $(H_{A'},u_0^{A'},\tau_{A'})$ and $(G_B,v_0^{B},\psi_B)$, $B \in A'$, results in an almost $d$-ary structure graph $(H'_{A'},u_0^{A'},\psi_{A'})$ for the group $\Gamma[\bigcup_{B \in A'} B]$.
 Finally, adding a new root vertex connected to the root vertices $u_0^{A'}$ of $(H'_{A'},w_0^{A'},\psi_{A'})$, where $A'$ is an orbit of $\Gamma[\FB_{j-1}]$, gives an almost $d$-ary structure graph $(G,v_0,\psi)$ of $\Gamma$.
 
 Now let $\Omega^{*} \coloneqq \Br^{*}(G,v_0)$ and let $\varphi\colon\Gamma\rightarrow\Sym(\Br^{*}(G,v_0))$ be the standard action of $\Gamma$ on the set $\Br^{*}(G,v_0)$.
 Let $\Gamma^{*} \coloneqq \Gamma^{\psi}$.
 Then $\Unf(G,v_0)$ is an almost $d$-ary structure tree for $\Gamma^{*}$ by Lemma \ref{la:standard-action-on-branches} giving an almost $d$-ary sequence of partitions $\{\Omega^{*}\} = \FB_0^{*} \succ \FB_1^{*} \succ \dots \succ \FB_k^{*} = \{\{\alpha\} \mid \alpha \in \Omega^{*}\}$.
 In particular, this gives Property \ref{item:renormalize-structure-tree-1}.
 Also, define
 \[f\colon \Omega^{*} \rightarrow \Omega\colon \bar v \mapsto L(\bar v).\]
 Let $(v_1,\dots,v_q) \in \Br^{*}(G,v_0) = \Omega^{*}$ and $\gamma \in \Gamma$.
 Then
 \begin{align*}
  (f(v_1,\dots,v_q))^{\gamma} = v_q^{\gamma} = L((v_1,\dots,v_q)^{\varphi(\gamma)}) = f((v_1,\dots,v_q)^{\varphi(\gamma)})  
 \end{align*}
 using Lemma \ref{la:standard-action-on-branches}.
 So Property \ref{item:renormalize-structure-tree-2} also holds.
 In order to define the partition $\FP^{*}$ pick, for each $P \in \FP$, a vertex $v_P \in V(G)$ such that $L(G,v_P) = P$ (cf.\ Lemma \ref{la:combine-structure-graphs-along-block-system}).
 Note that, by the above construction, one can pick $v_P \in V(G)$ in such a way that the subgraph of $G$ rooted at $v_P$ is a tree.
 Without loss of generality one may assume that all these vertices have the same distance from the root.
 
 Now let
 \[M \coloneqq \{(v_1,\dots,v_\ell) \in \Br(G) \mid \exists P \in \FP\colon v_\ell = v_P\}\]
 and define
 \[\FP^{*} \coloneqq \{\{(v_1,\dots,v_q) \in \Br^{*}(G) \mid (v_1,\dots,v_\ell) = (w_1,\dots,w_\ell) \} \mid (w_1,\dots,w_\ell) \in M\}.\]
 Clearly, $\FP^{*} = \FB_i^{*}$ for some $i \in [k]$ showing Property \ref{item:renormalize-structure-tree-6}.
 Moreover, for every $P^{*} \in \FP^{*}$ there is a unique $(w_1,\dots,w_\ell) \in M$ such that
 \[\{(v_1,\dots,v_q) \in \Br^{*}(G) \mid (v_1,\dots,v_\ell) = (w_1,\dots,w_\ell)\}\]
 which gives a unique $P \in \FP$ such that $w_\ell = v_P$.
 We define $f'(P^{*}) \coloneqq P$.
 Since the subgraph of $G$ rooted at $v_P$ is a tree it follows that $|P^{*}| = |f'(P^{*})|$ for every $P^{*} \in \FP^{*}$ giving Property \ref{item:renormalize-structure-tree-5}.
 Moreover,
 \begin{align*}
  \{f(\alpha^{*}) \mid \alpha^{*} \in P^{*}\} &= \{v_q \mid \exists (v_1,\dots,v_q) \in \Br^{*}(G,v_0)\colon (v_1,\dots,v_\ell) = (w_1,\dots,w_\ell)\} \\
                                              &= \{v_q \mid \exists (v_1,\dots,v_q) \in \Br^{*}(G,v_0) \; \exists \ell \in [q]\colon v_\ell = v_P\} = (f')(P^{*})
 \end{align*}
 proving Property \ref{item:renormalize-structure-tree-3}.
 Also, proving Property \ref{item:renormalize-structure-tree-4},
 \[|(f')^{-1}(P)| = |\{(v_1,\dots,v_q) \in \Br^{*}(G,v_0) \mid \exists \ell \in [q] \colon v_\ell = v_P\}| \leq |A|^{\funcnorm(d)},\]
 where $A$ is the orbit of $\Gamma_B[\FP[B]]$ containing $P$ where $P \subseteq B \in \FB_{j-1}$, follows from the construction of the structure graph $(G,v_0,\varphi)$ and Claim \ref{claim:number-of-branches}.
 
 Next, define
 \[\FX^{*} \coloneqq \{(P^{*},P^{*} \rightarrow \Sigma\colon \alpha^{*} \mapsto \Fx(f(\alpha^{*}))) \mid (f'(P^{*}),\Fx) \in \FX\}\]
 and
 \[\FY^{*} \coloneqq \{(Q^{*},Q^{*} \rightarrow \Sigma\colon \alpha^{*} \mapsto \Fy(f(\alpha^{*}))) \mid (f'(Q^{*}),\Fy) \in \FY\}.\]
 In particular, this proves Property \ref{item:renormalize-structure-tree-9}.
 
 To prove Property \ref{item:renormalize-structure-tree-7} let $\gamma \in \Gamma$.
 First suppose $\gamma \in \Iso_\Gamma(\FX,\FY)$ and let $(P,\Fx) \in \FX$.
 Let $P^{*} \in (f')^{-1}(P)$ and $\Fx^{*} \colon P^{*} \rightarrow \Sigma\colon \alpha^{*} \mapsto \Fx(f(\alpha^{*}))$.
 Let $Q^{*} \coloneqq (P^{*})^{\varphi(\gamma)}$.
 Note that $Q^{*} \in \FP^{*}$ since $\FP^{*}$ is $\Gamma^{*}$-invariant by Properties \ref{item:renormalize-structure-tree-1} and \ref{item:renormalize-structure-tree-6} already proved above.
 Let $Q = f'(Q^{*}) \in \FP$.
 Then
 \begin{align*}
  (P^{*},\Fx^{*})^{\varphi(\gamma)} &= \left(Q^{*},Q^{*} \rightarrow \Sigma\colon \alpha^{*} \mapsto \Fx(f((\alpha^{*})^{\varphi(\gamma)^{-1}}))\right) \\
                                    &= \left(Q^{*},Q^{*} \rightarrow \Sigma\colon \alpha^{*} \mapsto \Fx((f(\alpha^{*}))^{\gamma^{-1}})\right) \\
                                    &= \left(Q^{*},Q^{*} \rightarrow \Sigma\colon \alpha^{*} \mapsto \Fx^{\gamma}(f(\alpha^{*}))\right) \in \FY^{*}
 \end{align*}
 using Properties \ref{item:renormalize-structure-tree-2} and \ref{item:renormalize-structure-tree-3}.
 
 For the other direction suppose that $\varphi(\gamma) \in \Iso_{\Gamma^{*}}(\FX^{*},\FY^{*})$ and let $(P,\Fx) \in \FX$.
 Let $P^{*} \in (f')^{-1}(P)$ and $\Fx^{*} \colon P^{*} \rightarrow \Sigma\colon \alpha^{*} \mapsto \Fx(f(\alpha^{*}))$.
 Then $(Q^{*},\Fy^{*}) \coloneqq (P^{*},\Fx^{*})^{\varphi(\gamma)} \in \FY^{*}$.
 Let $Q = f'(Q^{*}) = f(Q^{*}) \in \FP$ and let $(Q,\Fy) \in \FY$ such that $\Fy^{*} \colon Q^{*} \rightarrow \Sigma\colon \alpha^{*} \mapsto \Fy(f(\alpha^{*}))$.
 Then $P^{\gamma} = f(P^{*})^{\gamma} = f((P^{*})^{\varphi(\gamma)}) = f(Q^{*}) = Q$ by Properties \ref{item:renormalize-structure-tree-2} and \ref{item:renormalize-structure-tree-3}.
 Let $\alpha \in P$ and let $\alpha^{*} \in f^{-1}(\alpha)$.
 Then
 \[\Fx(\alpha^{\gamma^{-1}}) = \Fx^{*}((\alpha^{*})^{\varphi(\gamma)^{-1}}) = \Fy^{*}(\alpha^{*}) = \Fy(\alpha)\]
 again using Property \ref{item:renormalize-structure-tree-2}.
 Hence, $\gamma \in \Iso_\Gamma(\Fx,\Fy)$.
 
 Finally, it remains to show Property \ref{item:renormalize-structure-tree-8}.
 But this follows from repeating the above arguments restricted to the window $W$, that is, for every $\gamma \in \Gamma$ it holds that $\gamma[W] \in \Iso(\FX[W],\FY[W])$ if and only if $(\varphi(\gamma))[f^{-1}(W)] \in \Iso(\FX^{*}[f^{-1}(W)],\FY^{*}[f^{-1}(W)])$.
\end{proof}

\subsection{Recursion Mechanisms}

Having introduced methods to obtain $d$-normalized instances, we can start to formulate the basic recursion mechanisms to be performed by our main algorithm.
We start by extending the orbit-by-orbit processing and standard Luks reduction (see Subsection \ref{subsec:string-isomorphism}) to the setting of hypergraphs basically following \cite{Miller83b}.
However, for the recursion we shall measure the progress not in terms of the actual size of the instance, but in terms of a \emph{virtual size} to be introduced next.
The basic idea behind introducing the virtual size of an instance relates to the need of invoking the renormalization procedure when calling the main algorithm recursively.
Intuitively speaking, the virtual size of an instance already takes into account the cost of all potential renormalization steps performed during the algorithm.
Hence, while the actual size of an instance may grow when renormalizing, its virtual size does not increase.
This way, we can measure the progress of the algorithm even when the instances get significantly larger due to the renormalization procedure.

\subsubsection{The Virtual Size of an Instance}

Again, recall the definition of the function $\funcnorm$ (see Lemma \ref{la:construct-structure-graph-transitive}).

\begin{definition}
 Let $\FP$ be a partition of $\Omega$ and suppose $\FX$ is a set of $\FP$-strings.
 The \emph{$d$-virtual size} of $\FX$ is defined by
 \[s = \sum_{P \in \FP} |P| \cdot (m_\FX(P))^{\funcnorm(d) + 1}.\]
\end{definition}

Note that the $d$-virtual size $s$ of $\FX$ is bounded by $(m+n)^{\CO(\log d)}$.
It thus suffices to design an algorithm solving the Generalized String Isomorphism Problem that runs in time $s^{\CO((\log d)^{c})}$ for some constant $c$.
As already indicated above, for the recursive algorithm progress will be measured in terms of the virtual size of instances.
Towards this end, we first state some basic properties on the virtual size.

\begin{observation}
 \label{obs:virtual-size-vs-domain-size}
 Let $\FX$ be a completely occupied set of $\FP$-strings of $d$-virtual size $s$.
 Moreover, let $W \subsetneq W' \subseteq \Omega$.
 Then the $d$-virtual size of $\FX[W]$ is strictly smaller than the $d$-virtual size of $\FX[W']$.
\end{observation}

\begin{observation}
 \label{obs:virtual-size-balanced}
 Let $\FX$ be a balanced set of $\FP$-strings of $d$-virtual size $s$.
 Then $s = n \cdot m_\FX^{\funcnorm(d) + 1}$ where $m_\FX \coloneqq m_\FX(P)$ for some $P \in \FP$.
 Moreover, if $W \subseteq \Omega$ and $s'$ is the $d$-virtual size of $\FX[W]$, then $s' \leq \frac{|W|}{n} \cdot s$.
\end{observation}

\begin{lemma}
 \label{la:virtual-size-for-partition}
 Let $\FP$ be a partition of $\Omega$ and suppose $\FX$ is a set of $\FP$-strings.
 Also let $\Omega = W_1 \uplus W_2$.
 Let $s$ be the $d$-virtual size of $\FX$ and $s_i$ the $d$-virtual size of $\FX[W_i]$.
 Then $s_1 + s_2 \leq s$.
\end{lemma}

\begin{proof}
 Let $\FP_i \coloneqq \FP[W_i]$ and $\FX_i \coloneqq \FX[W_i]$ for $i \in \{1,2\}$.
 Then
 \begin{align*}
  s_1 + s_2 &= \sum_{P_1 \in \FP_1} |P_1| \cdot (m_{\FX_1}(P_1))^{\funcnorm(d) + 1} + \sum_{P_2 \in \FP_2} |P_2| \cdot (m_{\FX_2}(P_2))^{\funcnorm(d) + 1}\\
            &= \sum_{P \in \FP} |P \cap W_1| \cdot (m_{\FX_1}(P \cap W_1))^{\funcnorm(d) + 1} + \sum_{P \in \FP} |P \cap W_2| \cdot (m_{\FX_2}(P \cap W_2))^{\funcnorm(d) + 1}\\
            &\leq \sum_{P \in \FP} |P \cap W_1| \cdot (m_{\FX}(P))^{\funcnorm(d) + 1} + \sum_{P \in \FP} |P \cap W_2| \cdot (m_{\FX}(P))^{\funcnorm(d) + 1}\\
            &= \sum_{P \in \FP} (|P \cap W_1| + |P \cap W_2|) \cdot (m_{\FX}(P))^{\funcnorm(d) + 1}\\
            &= \sum_{P \in \FP} |P| \cdot (m_{\FX}(P))^{\funcnorm(d) + 1}\\
            &= s.
 \end{align*}
\end{proof}

\subsubsection{Orbit-by-Orbit Processing}
\label{subsec:orbit-by-orbit}

Next, we lift the two basic recursion mechanisms, orbit-by-orbit processing and standard Luks reduction, to the setting of the Generalized String Isomorphism Problem where progress is measured in terms of the virtual size.

For performing orbit-by-orbit processing on sets $\FX$ of $\FP$-strings we have the additional complication that, in contrast to the String Isomorphism Problem, it is not possible to consider orbits independently.
However, as already noted by Miller in \cite{Miller83b}, the problem can be solved by an additional call to the standard String Isomorphism Problem.

\begin{lemma}
 \label{la:combine-windows}
 Let $\Gamma \leq \Sym(\Omega)$ be a $\mgamma_d$-group, $\FP$ a $\Gamma$-invariant partition of $\Omega$, and $\FX,\FY$ two sets of $\FP$-strings.
 Also let $W_1,W_2 \subseteq \Omega$ be two $\Gamma$-invariant sets such that $\Gamma[W_i] \leq \Aut(\FX[W_i])$ for both $i \in \{1,2\}$.
 Let $W := W_1 \cup W_2$
 
 Then a representation for the set $\{\gamma \in \Gamma \mid \gamma[W] \in \Iso_{\Gamma[W]}(\FX[W],\FY[W])\}$ can be computed in time $(n+m)^{\CO((\log d)^{c})}$ where $n \coloneqq |\Omega|$ and $m \coloneqq |\FX[W]|$.
\end{lemma}

\begin{algorithm}
 \caption{\textsf{CombineWindows}}
 \label{alg:combine-windows}
 \DontPrintSemicolon
 \SetKwInOut{Input}{Input}
 \SetKwInOut{Output}{Output}
 \Input{A group $\Gamma \leq \Sym(\Omega)$,
        a $\Gamma$-invariant partition $\FP$,
        two sets $\FX$ and $\FY$ of $\FP$-strings,
        and two $\Gamma$-invariant sets $W_1$ and $W_2$ such that $\Gamma[W_i] \leq \Aut(\FX[W_i])$ for both $i \in \{1,2\}$.}
 \Output{A representation for the set $\{\gamma \in \Gamma \mid \gamma[W] \in \Iso_{\Gamma[W]}(\FX[W],\FY[W])\}$ where $W \coloneqq W_1 \cup W_2$.}
 \BlankLine
 \If{$\FX[W_1] \neq \FY[W_1] \textup{ \bfseries or } \FX[W_2] \neq \FY[W_2]$}{
  \Return $\emptyset$\;
 }
 $E(\FX) := \bigl\{((P \cap W_1,\Fx[W_1]),(P \cap W_2,\Fx[W_2])) \;\bigm|\; (P,\Fx) \in \FX, P \cap W_i \neq \emptyset \text{ for both } i \in \{1,2\}\bigr\}$\;
 $G(\FX) := (\FX[W_1] \uplus \FX[W_2],E(\FX))$\;
 $E(\FY) := \bigl\{((Q \cap W_1,\Fy[W_1]),(Q \cap W_2,\Fy[W_2])) \;\bigm|\; (Q,\Fy) \in \FY, Q \cap W_i \neq \emptyset \text{ for both } i \in \{1,2\}\bigr\}$\;
 $G(\FY) := (\FY[W_1] \uplus \FY[W_2],E(\FY))$\;
 $\varphi \colon \Gamma[W] \rightarrow \Sym(\FX[W_1] \uplus \FX[W_2])$ natural homomorphism\;
 \Return $\{\gamma \in \Gamma \mid \varphi(\gamma[W]) \in \Iso_{\Gamma[W]^{\varphi}}(G(\FX),G(\FY))\}$\;
\end{algorithm}

\begin{proof}
 Consider Algorithm \ref{alg:combine-windows}.
 First suppose $\FX[W_i] \neq \FY[W_i]$ for some $i \in \{1,2\}$.
 Then \[\Iso_{\Gamma[W]}(\FX[W],\FY[W]) = \emptyset\] since $\FX[W_i]^{\gamma} = \FX[W_i]$ for every $\gamma \in \Gamma$.
 
 So assume $\FX[W_i] = \FY[W_i]$ for both $i \in \{1,2\}$.
 Consider the graphs $G(\FX)$ and $G(\FY)$ defined in Line 5 and 7 of Algorithm \ref{alg:combine-windows}.
 The group $\Gamma$ naturally acts on the vertex sets $V(G(\FX)) = V(G(\FY))$ via the homomorphism $\varphi \colon \Gamma[W] \rightarrow \Sym(\FX[W_1] \uplus \FX[W_2])$.
 We need to argue that $\Iso_{\Gamma[W]}(\FX[W],\FY[W]) = \{\gamma \in \Gamma[W] \mid \varphi(\gamma) \in \Iso_{\Gamma^{\varphi}}(G(\FX),G(\FY))\}$.
 
 Let $\gamma \in \Iso_{\Gamma[W]}(\FX[W],\FY[W])$ and let $((P \cap W_1,\Fx[W_1]),(P \cap W_2,\Fx[W_2])) \in E(\FX)$ for some $(P,\Fx) \in \FX$.
 Then $(P \cap W,\Fx[W])^{\gamma} \in \FY[W]$ and hence, $((P^\gamma \cap W_1,\Fx[W_1]^\gamma),(P^\gamma \cap W_2,\Fx[W_2]^\gamma)) \in E(\FY)$.
 
 On the other hand, let $\gamma \in \Gamma[W]$ such $\varphi(\gamma) \in \Iso_{\Gamma[W]^{\varphi}}(G(\FX),G(\FY))$.
 Also let $(P \cap W,\Fx[W]) \in \FX[W]$.
 Then $(P^{\gamma} \cap W_i,\Fx[W_i]^{\gamma}) \in \FX[W_i] = \FY[W_i]$ for $i \in \{1,2\}$.
 Moreover, $((P^\gamma \cap W_1,\Fx[W_1]^\gamma), (P^\gamma \cap W_2,\Fx[W_2]^\gamma)) \in E(\FY)$.
 But this implies $(P^{\gamma} \cap W,\Fx[W]^{\gamma}) \in \FY[W]$.
 
 Finally, $\Iso_{\Gamma[W]^{\varphi}}(G(\FX),G(\FY))$ can be computed in time $(n+m)^{\CO((\log d)^{c})}$ by Corollary \ref{cor:graph-isomorphism-gamma-d}.
\end{proof}

Let $\Gamma \leq \Sym(\Omega)$ be a $\mgamma_d$-group, $\FP$ a $\Gamma$-invariant partition of $\Omega$, and $\FX,\FY$ two sets of $\FP$-strings.
Also, suppose $\Gamma$ is not transitive.
Then the set $\Iso_\Gamma(\FX,\FY)$ can be recursively computed as follows.
Let $W_1,\dots,W_r$ be the orbits of $\Gamma$.
Let $\Gamma_0 \coloneqq \Gamma$ and $\gamma_0 \coloneqq \id$.
We inductively compute
\[\Gamma_{i+1}\gamma_{i+1} = \{ \gamma \in \Gamma_i\gamma_i \mid \gamma[W_i] \in \Iso_{\Gamma_i\gamma_i[W_i]}(\FX[W_{i+1}],\FY[W_{i+1}]).\]
Then $\Iso_\Gamma(\FX,\FY) = \Iso_{\Gamma_r\gamma_r}(\FX,\FY) = \Iso_{\Gamma_r}(\FX,\FY^{\gamma_r^{-1}})$.
Let $\FX' \coloneqq \FX$ and $\FX' \coloneqq \FY^{\gamma_r^{-1}}$.
Then $\Gamma_r[W_i] \leq \Aut(\FX'[W_i])$ for all $i \in [r]$.
Hence, $\Iso_\Gamma(\FX,\FY)$ can be computed using Lemma \ref{la:combine-windows}.

To analyze the recursion let $s$ be the $d$-virtual size of $\FX$ and $s_i$ be the $d$-virtual size of $\FX[W_i]$, $i \in [r]$.
Then $\sum_{i = 1}^{r} s_i \leq s$ by Lemma \ref{la:virtual-size-for-partition}.
Hence, we obtain the same recurrence as for standard orbit-by-orbit processing on strings with respect to the actual size of the instance.

\subsubsection{Standard Luks Reduction}
\label{subsec:luks-recursion}

Next, we consider standard Luks reduction.
As before, let $\Gamma \leq \Sym(\Omega)$ be a $\mgamma_d$-group, $\FP$ a $\Gamma$-invariant partition of $\Omega$, and $\FX,\FY$ two sets of $\FP$-strings.
Suppose $\Gamma$ is transitive and let $\FB$ be a minimal block system for $\Gamma$ of size $b \coloneqq |\FB|$.
Let $T$ be a transversal for $\Gamma_{(\FB)}$ in $\Gamma$ and $t \coloneqq |T|$.
Then
\[\Iso_\Gamma(\FX,\FY) = \bigcup_{\delta \in T} \Iso_{\Gamma_{(\FB)}\delta}(\FX,\FY).\]
Since each orbit of $\Gamma_{(\FB)}$ has size at most $n/b$ one can apply orbit-by-orbit processing.
In total, this results in $t \cdot b$ recursive calls over windows of size $n/b$.

While this results in a good recurrence for the String Isomorphism Problem, the situation for the algorithm presented in this work looks different as progress is measured with respect to the virtual size.
Indeed, while each orbit $A$ of $\Gamma_{(\FB)}$ has size at most $n/b$ it is not possible to prove such a bound for the virtual size.
This is due to the fact that the set of $\FP$-strings $\FX$ might be unbalanced, i.e., the numbers $m_\FX(P)$ might differ drastically for $P \in \FP$.
The solution to this problem is to first balance the input using Lemma \ref{la:balance-orbits}.
For a transitive group, the algorithm first applies Lemma \ref{la:balance-orbits}.
If the updated group is not transitive, we can apply orbit-by-orbit processing.
Otherwise, the group is transitive and $\FX$ and $\FY$ are balanced.
Then, for $B \in \FP$, the $d$-virtual size of $\FX[B]$ is at most $s/b$ where $s$ is the $d$-virtual size of $\FX$ by Observation \ref{obs:virtual-size-balanced}.
This gives the desired recursion.

\subsubsection{Splitting Partition Classes}

Besides orbit-by-orbit processing and standard Luks reduction, we also need to provide a way of modifying instances in a certain way.
This type of modification is specifically designed for the setting of hypergraphs and forms one of the key conceptual contributions of this paper in order to generalize the group-theoretic techniques developed by Babai \cite{Babai16} to hypergraphs.

The main obstacle for lifting Babai's algorithm to hypergraphs lies in adapting the Local Certificates Routine.
As usual, let $\Gamma \leq \Sym(\Omega)$ be a $\mgamma_d$-group, $\FP$ a $\Gamma$-invariant partition of $\Omega$, and $\FX,\FY$ two sets of $\FP$-strings.
In a nutshell, the Local Certificates Routine considers a $\Gamma$-invariant window $W \subseteq \Omega$ such that $\Gamma[W] \leq \Aut(\FX[W])$ (i.e., the group $\Gamma$ respects $\FX$ restricted to the window $W$) and aims at creating automorphisms of $\FX$.
In order to compute automorphisms of the entire set of $\FP$-strings $\FX$ the Local Certificates Routine considers the group $\Gamma_{(\Omega \setminus W)}$ fixing every point outside of $W$.
Intuitively speaking, the Unaffected Stabilizers Theorem guarantees that the group $\Gamma_{(\Omega \setminus W)}$ is large.
However, in contrast to the String Isomorphism Problem, the group $\Gamma_{(\Omega \setminus W)}$ is not a subgroup of the automorphism group $\Aut(\FX)$ since the action on the set of strings is may not be consistent within and outside the window $W$.
In other words, in order to compute $\Aut(\FX)$, it is not possible to consider $\FX[W]$ and $\FX[\Omega \setminus W]$ independently (similar to orbit-by-orbit processing).

To solve this problem, first consider the special case that $\FX[W]$ is simple, i.e., for every $P \in \FP$ all strings from $\Str_\FX(P)$ are identical on $W \cap P$.
In this case it is possible to consider $\FX[W]$ and $\FX[\Omega \setminus W]$ independently as the next lemma indicates.

\begin{lemma}
 \label{la:extending-window-automorphisms}
 Let $\Gamma \leq \Sym(\Omega)$ be a permutation group, let $\FP$ be a $\Gamma$-invariant partition of $\Omega$ and $\FX$ a set of $\FP$-strings.
 Also suppose $W \subseteq \Omega$ is a $\Gamma$-invariant window such that $\FX[W]$ is simple and $\Gamma[W] \leq \Aut(\FX[W])$.
 Then $\Gamma_{(\Omega \setminus W)} \leq \Aut(\FX)$.
\end{lemma}

\begin{proof}
 Let $\gamma \in \Gamma_{(\Omega \setminus W)}$ and $(P,\Fx) \in \FX$.
 It suffices to show that $(P^{\gamma},\Fx^{\gamma}) \in \FX$.
 First suppose $P \subseteq W$. Then $(P,\Fx) \in \FX[W]$ and $(P^{\gamma},\Fx^{\gamma}) \in \FX[W]$.
 Moreover, $P^{\gamma} \subseteq W$ since $W$ is $\Gamma$-invariant.
 Hence, $(P^{\gamma},\Fx^{\gamma}) \in \FX$.
 
 Otherwise $P \cap (\Omega \setminus W) \neq \emptyset$.
 Since $\alpha^{\gamma} = \alpha$ for all $\alpha \in \Omega \setminus W$ it follows that $P^{\gamma} \cap P \neq \emptyset$.
 Using the fact that $\FP$ is $\Gamma$-invariant this implies that $P^{\gamma} = P$.
 To complete the proof it is argued that $\Fx^{\gamma} = \Fx$.
 Let $\alpha \in P$.
 If $\alpha \notin W$ then $\Fx^{\gamma}(\alpha) = \Fx^{\gamma}(\alpha^{\gamma}) = \Fx(\alpha^{\gamma\gamma^{-1}}) = \Fx(\alpha)$.
 So assume $\alpha \in W$.
 Since $\gamma[W] \in \Aut(\FX[W])$ it holds that $((P \cap W)^{\gamma},(\Fx[P \cap W])^{\gamma}) = (P \cap W,(\Fx[P \cap W])^{\gamma}) \in \FX[W]$.
 But this means $(\Fx[P \cap W])^{\gamma} = \Fx[P \cap W]$ since $\FX[W]$ is simple.
 Hence $\Fx^{\gamma}(\alpha) = \Fx(\alpha)$.
\end{proof}

Let $W \subseteq \Omega$ be a $\Gamma$-invariant set such that $\Gamma[W] \leq \Aut(\FX[W])$ (this is the case during the Local Certificates Routine).
In order to solve the problem described above in general, we modify the instance in such a way that $\FX[W]$ becomes simple, which allows us to apply Lemma \ref{la:extending-window-automorphisms}.
Consider a set $P \in \FP$.
In order to ``simplify'' the instance we define an equivalence relation on the set $\Str_\FX(P)$ of all strings contained in $P$ where two strings are equivalent if they are identical on the window $W \cap P$.
For each equivalence class we create a new block $P'$ containing exactly the strings from the equivalence class.
Since the group $\Gamma$ respects the induced subinstance $\FX[W]$ it naturally acts on the equivalence classes and we can define an equivalent instance $(\Gamma',\FP',\FX',\FY')$ of the Generalized String Isomorphism Problem together with a corresponding window $W'$ such that $\FX'[W']$ is simple.
This process is visualized in Figure \ref{fig:simplify-on-window}.

\begin{figure}
  \centering
  \begin{tikzpicture}
   \node at (-4.0,0.8) {$\FX$};
   
   \draw[fill,gray!50] (-0.15,-0.15) rectangle (1.2,1.8);
   \draw[fill,gray!50] (3.15,-0.15) rectangle (4.2,1.8);
   
   \node at (1.35,2.4) {$P_1$};
   \node at (4.65,2.4) {$P_2$};
   
   \draw[thick] (-0.3,-0.3) -- (6.3,-0.3);
   \draw[thick] (-0.3,1.95) -- (6.3,1.95);
   \draw[thick] (-0.3,-0.3) -- (-0.3,1.95);
   \draw[thick] (3,-0.3) -- (3,1.95);
   \draw[thick] (6.3,-0.3) -- (6.3,1.95);
   
   \draw (0,0) -- (2.7,0);
   \draw (0,0.3) -- (2.7,0.3);
   \draw (0,0.45) -- (2.7,0.45);
   \draw (0,0.75) -- (2.7,0.75);
   \draw (0,0.9) -- (2.7,0.9);
   \draw (0,1.2) -- (2.7,1.2);
   \draw (0,1.35) -- (2.7,1.35);
   \draw (0,1.65) -- (2.7,1.65);
   \foreach \i in {0,...,9}{
    \draw (\i*0.3,0) -- (\i*0.3,0.3);
    \draw (\i*0.3,0.45) -- (\i*0.3,0.75);
    \draw (\i*0.3,0.9) -- (\i*0.3,1.2);
    \draw (\i*0.3,1.35) -- (\i*0.3,1.65);
   }
   \foreach \a [count=\i] in {a,b,a,a,a,b,b,a,b}{
    \node at (-0.15 + \i*0.3,0.15) {\scriptsize $\a$};
   }
   \foreach \a [count=\i] in {a,a,a,b,a,b,a,a,a}{
    \node at (-0.15 + \i*0.3,0.6) {\scriptsize $\a$};
   }
   \foreach \a [count=\i] in {a,a,a,b,a,a,a,b,a}{
    \node at (-0.15 + \i*0.3,1.05) {\scriptsize $\a$};
   }
   \foreach \a [count=\i] in {a,b,a,a,a,b,a,b,a}{
    \node at (-0.15 + \i*0.3,1.5) {\scriptsize $\a$};
   }
   
   \draw (3.3,0) -- (6,0);
   \draw (3.3,0.3) -- (6,0.3);
   \draw (3.3,0.45) -- (6,0.45);
   \draw (3.3,0.75) -- (6,0.75);
   \draw (3.3,0.9) -- (6,0.9);
   \draw (3.3,1.2) -- (6,1.2);
   \draw (3.3,1.35) -- (6,1.35);
   \draw (3.3,1.65) -- (6,1.65);
   \foreach \i in {0,...,9}{
    \draw (3.3+\i*0.3,0) -- (3.3+\i*0.3,0.3);-0.3
    \draw (3.3+\i*0.3,0.45) -- (3.3+\i*0.3,0.75);
    \draw (3.3+\i*0.3,0.9) -- (3.3+\i*0.3,1.2);
    \draw (3.3+\i*0.3,1.35) -- (3.3+\i*0.3,1.65);
   }
   \foreach \a [count=\i] in {a,a,b,b,b,b,b,a,a}{
    \node at (3.15 + \i*0.3,0.15) {\scriptsize $\a$};
   }
   \foreach \a [count=\i] in {a,b,a,a,a,a,a,b,b}{
    \node at (3.15 + \i*0.3,0.6) {\scriptsize $\a$};
   }
   \foreach \a [count=\i] in {a,b,a,a,a,b,a,a,b}{
    \node at (3.15 + \i*0.3,1.05) {\scriptsize $\a$};
   }
   \foreach \a [count=\i] in {a,b,a,a,b,a,a,b,b}{
    \node at (3.15 + \i*0.3,1.5) {\scriptsize $\a$};
   }
   
   \draw[thick,->, thick] (1,-0.4) -- (-1.95,-1.4);
   \draw[thick,->, thick] (2.3,-0.4) -- (1.25,-1.4);
   \draw[thick,->, thick] (3.7,-0.4) -- (4.55,-1.4);
   \draw[thick,->, thick] (5,-0.4) -- (7.85,-1.4);
   
   \node at (-4.0,-2.4) {$\FX'$};
   
   \draw[fill,gray!50] (-3.45,-3.15) rectangle (-2.1,-2.1);
   \draw[fill,gray!50] (-0.15,-3.15) rectangle (1.2,-2.1);
   \draw[fill,gray!50] (3.15,-3.15) rectangle (4.2,-1.65);
   \draw[fill,gray!50] (6.45,-3.15) rectangle (7.5,-2.55);
   
   \node at (-1.95,-3.75) {$P_1 \times \{abaa\}$};
   \node at (1.35,-3.75) {$P_1 \times \{aaab\}$};
   \node at (4.65,-3.75) {$P_2 \times \{aba\}$};
   \node at (7.95,-3.75) {$P_2 \times \{aab\}$};
   
   \draw[thick] (-3.6,-3.3) -- (9.6,-3.3);
   \draw[thick] (-3.6,-1.5) -- (9.6,-1.5);
   \draw[thick] (-3.6,-3.3) -- (-3.6,-1.5);
   \draw[thick] (-0.3,-3.3) -- (-0.3,-1.5);
   \draw[thick] (3,-3.3) -- (3,-1.5);
   \draw[thick] (6.3,-3.3) -- (6.3,-1.5);
   \draw[thick] (9.6,-3.3) -- (9.6,-1.5);
   
   \draw (0,-3) -- (2.7,-3);
   \draw (0,-2.7) -- (2.7,-2.7);
   \draw (0,-2.55) -- (2.7,-2.55);
   \draw (0,-2.25) -- (2.7,-2.25);
   \draw (-3.3,-3) -- (-0.6,-3);
   \draw (-3.3,-2.7) -- (-0.6,-2.7);
   \draw (-3.3,-2.55) -- (-0.6,-2.55);
   \draw (-3.3,-2.25) -- (-0.6,-2.25);
   \foreach \i in {0,...,9}{
    \draw (-3.3+\i*0.3,-3) -- (-3.3+\i*0.3,-2.7);
    \draw (-3.3+\i*0.3,-2.55) -- (-3.3+\i*0.3,-2.25);
    \draw (\i*0.3,-3) -- (\i*0.3,-2.7);
    \draw (\i*0.3,-2.55) -- (\i*0.3,-2.25);
   }
   \foreach \a [count=\i] in {a,b,a,a,a,b,b,a,b}{
    \node at (-3.45 + \i*0.3,-2.85) {\scriptsize $\a$};
   }
   \foreach \a [count=\i] in {a,a,a,b,a,b,a,a,a}{
    \node at (-0.15 + \i*0.3,-2.85) {\scriptsize $\a$};
   }
   \foreach \a [count=\i] in {a,a,a,b,a,a,a,b,a}{
    \node at (-0.15 + \i*0.3,-2.4) {\scriptsize $\a$};
   }
   \foreach \a [count=\i] in {a,b,a,a,a,b,a,b,a}{
    \node at (-3.45 + \i*0.3,-2.4) {\scriptsize $\a$};
   }
   
   \draw (3.3,-3) -- (6,-3);
   \draw (3.3,-2.7) -- (6,-2.7);
   \draw (3.3,-2.55) -- (6,-2.55);
   \draw (3.3,-2.25) -- (6,-2.25);
   \draw (3.3,-2.1) -- (6,-2.1);
   \draw (3.3,-1.8) -- (6,-1.8);
   \draw (6.6,-3) -- (9.3,-3);
   \draw (6.6,-2.7) -- (9.3,-2.7);
   \foreach \i in {0,...,9}{
    \draw (3.3+\i*0.3,-3) -- (3.3+\i*0.3,-2.7);
    \draw (3.3+\i*0.3,-2.55) -- (3.3+\i*0.3,-2.25);
    \draw (3.3+\i*0.3,-2.1) -- (3.3+\i*0.3,-1.8);
    \draw (6.6+\i*0.3,-3) -- (6.6+\i*0.3,-2.7);
   }
   \foreach \a [count=\i] in {a,a,b,b,b,b,b,a,a}{
    \node at (6.45 + \i*0.3,-2.85) {\scriptsize $\a$};
   }
   \foreach \a [count=\i] in {a,b,a,a,a,a,a,b,b}{
    \node at (3.15 + \i*0.3,-2.85) {\scriptsize $\a$};
   }
   \foreach \a [count=\i] in {a,b,a,a,a,b,a,a,b}{
    \node at (3.15 + \i*0.3,-2.4) {\scriptsize $\a$};
   }
   \foreach \a [count=\i] in {a,b,a,a,b,a,a,b,b}{
    \node at (3.15 + \i*0.3,-1.95) {\scriptsize $\a$};
   }
   
  \end{tikzpicture}
  \caption{Visualization for the first step of the proof of Theorem \ref{thm:simplify-on-window}.
   A set $\FX$ of $\FP$-strings is given in the top and the ``simplified'' instance $\FX'$ is given below. The window $W$ is marked in gray.
   Note that $\FX'[W']$ is simple where $W'$ denotes the window marked in gray in the bottom part of the figure.}
  \label{fig:simplify-on-window}
\end{figure}

While this simplification allows us to treat $\FX[W]$ and $\FX[W \setminus \Omega]$ independently, it creates a number of different problems that that need to be addressed.
First of all, $\Gamma'$ may not be $d$-normalized.
However, by lifting an almost $d$-ary sequence of partitions from $\Gamma$ to $\Gamma'$, we obtain a sequence of $\Gamma'$-invariant partitions that is ``close'' to being almost $d$-ary in the sense of Theorem \ref{thm:renormalize-structure-tree}.
Hence, we can invoke the renormalization procedure from Theorem \ref{thm:renormalize-structure-tree} to obtain a $d$-normalized instance again.

However, this creates another problem.
The output of the Local Certificates Routine must be isomorphism-invariant regarding all pairs of test sets considered in the aggregation of the local certificates.
But the renormalization procedure is not isomorphism-invariant.
Luckily, this problem can be circumvented by performing the Local Certificates Routine on all pairs of test sets in parallel, simplifying at the same time in the same way.
Towards this end, we need to be able to simplify a list of sets of $\FP$-strings instead of only a single pair.
All of this is achieved by the next theorem.

\begin{theorem}
 \label{thm:simplify-on-window}
 Let $\Gamma \leq \Sym(\Omega)$ be a $\mgamma_d$-group and $\FP$ be a $\Gamma$-invariant partition of $\Omega$.
 Also suppose $\{\Omega\} = \FB_0 \succ \FB_1 \succ \dots \succ \FB_\ell = \{\{\alpha\} \mid \alpha \in \Omega\}$ forms an almost $d$-ary sequence of $\Gamma$-invariant partitions such that $\FP = \FB_i$ for some $i \in [\ell]$.
 Let $(\FX_i)_{i \in [p]}$ be a list of sets of $\FP$-strings.
 Also let $W \subseteq \Omega$ be a $\Gamma$-invariant set such that $\Gamma[W] \leq \Aut_{\Gamma[W]}(\FX_i[W])$ for all $i \in [p]$.
 Moreover, assume that $\FX_i[W] = \FX_j[W]$ for all $i,j \in [p]$.
 
 Then there is an equivalence relation $\sim$ on the set $[p]$ such that $\FX_i \cong_\Gamma \FX_j$ implies $i \sim j$ for all $i,j \in [p]$,
 and for each equivalence class $A \subseteq [p]$, we get the following:
 
 A set $\Omega^{*}$,
 a group $\Gamma^{*} \leq \Sym(\Omega)$,
 elements $\lambda_{i} \in \Gamma$ for all $i \in A$,
 a monomorphism $\varphi\colon \Gamma^{*} \rightarrow \Gamma$,
 a window $W^{*} \subseteq \Omega^{*}$,
 a sequence of partitions $\{\Omega^{*}\} = \FB_0^{*} \succ \FB_1^{*} \succ \dots \succ \FB_k^{*} = \{\{\alpha\} \mid \alpha \in \Omega^{*}\}$,
 and a list of $\FP^{*}$-strings $(\FX_i^{*})_{i \in A}$,
 such that the following properties are satisfied:
 \begin{enumerate}[label=(\Alph*)]
  \item\label{item:simplify-on-window-1} the sequence $\FB_0^{*} \succ \dots \succ \FB_k^{*}$ forms an almost $d$-ary sequence of $\Gamma^{*}$-invariant partitions,
  \item\label{item:simplify-on-window-2} $W^{*}$ is $\Gamma^{*}$-invariant and $\Gamma^{*}[W^{*}] \leq \Aut_{\Gamma^{*}[W^{*}]}(\FX_i^{*}[W^{*}])$ for all $i \in A$,
  \item\label{item:simplify-on-window-3} there is some $i \in [k]$ such that $\FP^{*} = \FB_{i}^{*}$,
  \item\label{item:simplify-on-window-4} $\Iso_\Gamma(\FX_i,\FX_j) = \lambda_i^{-1}(\Iso_{\Gamma^{*}}(\FX_i^{*},\FX_j^{*}))^{\varphi}\lambda_j$ for all $i,j \in A$,
  \item\label{item:simplify-on-window-5} $\FX_i^{*}[W^{*}]$ is simple for all $i \in A$,
  \item\label{item:simplify-on-window-6} for $s$ the virtual size of $\FX_i$ and $s^{*}$ the virtual size of $\FX_i^{*}$ it holds that $s^{*} \leq s$, and
  \item\label{item:simplify-on-window-7} for $s$ the virtual size of $\FX_i[\Omega \setminus W]$ and $s^{*}$ the virtual size of $\FX_i^{*}[\Omega^{*} \setminus W^{*}]$ it holds that $s^{*} \leq s$.
 \end{enumerate}
 Moreover, there is an algorithm computing the desired objects in time $p^{2} \cdot (n + m)^{\CO((\log d)^{c})}$ for some absolute constant $c$.
\end{theorem}

\begin{proof}
 For the proof we may assume that $\FX_i[W]$ is not simple (otherwise nothing needs to be done).
 The proof essentially works in two steps. The first step is to ensure Property \ref{item:simplify-on-window-5} which is the crucial property added by this theorem.
 The group $\Gamma' \leq \Sym(\Omega')$ obtained from the first step may not have an almost $d$-ary sequence of partitions.
 Hence, we have to renormalize the instances which builds the second step.
 
 Consider the natural homomorphism $\psi \colon \Gamma \rightarrow \Sym(\FX[W])$.
 Now let
 \[\Omega' \coloneqq \bigcup_{P \in \FP} P \times \{\Fx[W \cap P] \mid (P,\Fx) \in \FX\}\]
 and
 \[\FP' \coloneqq \{P \times \{\Fx[W \cap P]\} \mid (P,\Fx) \in \FX\}\]
 Also let
 \[\FX_i' \coloneqq \left\{\Bigl(P \times \{\Fx[W \cap P]\},\Fx'\colon P \times \{\Fx[W \cap P]\} \rightarrow \Sigma\colon (\alpha,\Fx[W \cap P]) \mapsto \Fx(\alpha)\Bigr) \;\;\Big|\;\; (P,\Fx) \in \FX_i\right\}\]
 Note that $\FX_i'$ is a set of $\FP'$-strings for every $i \in [p]$.
 Finally, the group $\Gamma$ faithfully acts on the set $\Omega'$ via
 \[(\alpha,\Fz)^{\gamma} = (\alpha^{\gamma},\Fz^{\psi(\gamma)})\]
 yielding a monomorphism $\psi' \colon \Gamma \rightarrow \Sym(\Omega')$.

 Now let $i,j \in [p]$, $\FX = \FX_i$, $\FX' = \FX_i'$, $\FY = \FX_j$ and $\FY' = \FX_j'$.
 
 \begin{claim}
  \label{claim:preserve-isomorphisms}
  For every $\gamma \in \Gamma$ it holds that $\FX^{\gamma} = \FY$ if and only if $(\FX')^{\psi'(\gamma)} \cong \FY'$.
 \end{claim}
 \begin{claimproof}
  First suppose that $\FX^{\gamma} = \FY$.
  Let $(P \times \{\Fx[W \cap P]\},\Fx') \in \FX'$ where $\Fx'\colon P \times \{\Fx[W \cap P]\}\colon (\alpha,\Fx[W \cap P]) \mapsto \Fx(\alpha)$.
  Then $(P \times \{\Fx[W \cap P]\},\Fx')^{\psi'(\gamma)} = (P^{\gamma} \times \{\Fx[W \cap P]^{\gamma}\},(\Fx')^{\gamma}) = (P^{\gamma} \times \{\Fy[W \cap P^{\gamma}]\},\Fy')$
  where $(P^{\gamma},\Fy) = (P,\Fx)^{\gamma} \in \FY$ and $\Fy'\colon P^{\gamma} \times \{\Fy[W \cap P^{\gamma}]\}\colon (\alpha,\Fy[W \cap P^{\gamma}]) \mapsto \Fy(\alpha)$.
  
  For the other direction suppose $(\FX')^{\psi'(\gamma)} \cong \FY'$ and let $(P,\Fx) \in \FX$.
  Then $(P \times \{\Fx[W \cap P]\},\Fx') \in \FX'$ where $\Fx'\colon P \times \{\Fx[W \cap P]\}\colon (\alpha,\Fx[W \cap P]) \mapsto \Fx(\alpha)$.
  Hence, $(P \times \{\Fx[W \cap P]\},\Fx')^{\psi'(\gamma)} \in \FY'$.
  So $(P \times \{\Fx[W \cap P]\},\Fx')^{\psi'(\gamma)} = (Q \times \{\Fy[W \cap Q]\},\Fy')$ where $(Q,\Fy) \in \FY$ and $\Fy'\colon P^{\gamma} \times \{\Fy[W \cap P^{\gamma}]\}\colon (\alpha,\Fy[W \cap P^{\gamma}]) \mapsto \Fy(\alpha)$.
  But then $(P,\Fx)^{\gamma} = (Q,\Fy) \in \FY$.
 \end{claimproof}

 Now let $W' \coloneqq \{(\alpha,\Fz) \in \Omega' \mid \alpha \in W\} \subseteq \Omega'$. Moreover let $\Gamma' \coloneqq \Gamma^{\psi'}$ denote the image of $\psi'$.
 
 \begin{claim}
  \label{claim:properties-updated-window}
  The set $W'$ is $\Gamma'$-invariant and $\Gamma'[W'] \leq \Aut(\FX'[W'])$.
 \end{claim}
 \begin{claimproof}
  The first part of the claim is immediately clear.
  For the second part let $\gamma \in \Gamma$ and let $\gamma' = \psi'(\gamma)$.
  We need to argue that $\gamma'[W'] \in \Aut(\FX'[W'])$.
  Towards this end let $((P \cap W) \times \{\Fx[W \cap P]\},\Fx'[W']) \in \FX'[W']$ where $(P,\Fx) \in \FX$ and $\Fx'\colon P \times \{\Fx[W \cap P]\}\colon (\alpha,\Fx[W \cap P]) \mapsto \Fx(\alpha)$.
  We have $(P \cap W,\Fx[W])^{\gamma} \eqqcolon (Q \cap W,\Fy[W]) \in \FX[W]$ since $\Gamma[W] \leq \Aut[W]$.
  Hence, $((Q \cap W) \times \{\Fy[W \cap Q]\},\Fy'[W']) = ((P \cap W) \times \{\Fx[W \cap P]\},\Fx'[W'])^{\gamma'} \in \FX'[W']$.
 \end{claimproof}
 
 \begin{claim}
  \label{claim:window-simple}
  $\FX'[W']$ is simple.
 \end{claim}
 \begin{claimproof}
  Let $P' \in \FP'$ such that $P' \cap W' \neq \emptyset$.
  Then there is some $(P,\Fx) \in \FX$ such that $P' = P \times \{\Fx[W \cap P]\}$.
  Now let $(P',\Fx_1'),(P',\Fx_2') \in \FX'$.
  But then $\Fx_1'(\alpha,\Fx[W \cap P]) = \Fx_1'(\alpha,\Fx[W \cap P])$ for all $\alpha \in W$ which implies that $(P' \cap W',\Fx_1'[P' \cap W']) = (P' \cap W',\Fx_2'[P' \cap W'])$.
  Hence, $m_{\FX'[W']}(P') \leq 1$.
 \end{claimproof}
 
 Next, consider the almost $d$-ary sequence of $\Gamma$-invariant partitions $\FB_0 \succ \FB_1 \succ \dots \succ \FB_\ell$ and also recall that $\FP = \FB_i$.
 For $0 \leq j \leq i$ define
 \[\FB_j' \coloneqq \{B \times \{\Fx \mid P \in \FP, P \subseteq B, (P \cap W,\Fx) \in \FX[W]\} \mid B \in \FB_j\}.\]
 Moreover, let $\FB_{i+1}' \coloneqq \FP'$, and for $i+1 \leq j \leq \ell$ define
 \[\FB_{j+1}' \coloneqq \{B \times \{\Fx\} \mid B \in \FB_{j}, P \in \FP, P \supseteq B, (P \cap W,\Fx) \in \FX[W] \}.\]
 It is to see that $\FB_j'$ is $\Gamma'$-invariant partition of $\Omega'$ for all $0 \leq j \leq \ell + 1$.
 Moreover, $\FB_0' \succ \FB_1' \succ \dots \succ \FB_{\ell+1}'$ since $\FX[W]$ is not simple.
 
 Now let $1 \leq j \leq i$ and $B \in \FB_{j-1}$.
 Also let $B' \coloneqq B \times \{\Fx \mid P \in \FP, P \subseteq B, (P \cap W,\Fx) \in \FX[W]\} \in \FB_{j-1}'$.
 Then $\Gamma_B[\FB_j[B]]$ is permutationally equivalent to $(\Gamma')_{B'}[\FB_j'[B']]$.
 
 For $i+1 \leq j \leq \ell$ and $B \in \FB_{j-1}$ it holds that $(\Gamma')_{B'}[\FB_{j+1}'[B']]$ is permutationally equivalent to a subgroup of $\Gamma_B[\FB_j[B]]$ for all $B' \in \{B \times \{\Fx\} \mid P \in \FP, P \supseteq B, (P \cap W,\Fx) \in \FX[W] \} \subseteq \FB_{j}'$.
 
 As a result, for every $i+1 \neq j \in [\ell+1]$ and every $B' \in \FB_{j-1}'$ it holds that $|\FB'_j[B']| \leq d$ or $(\Gamma')_{B'}[\FB'_j[B']]$ is semi-regular.
 
 Before applying Theorem \ref{thm:renormalize-structure-tree} to renormalize the current instance we need to ensure one additional property required later on.
 Define strings
 \[\Fx_i\colon \FP' \rightarrow [m]^{2}\colon P' \mapsto (m_{\FX_i'}(P'),m_{\FX_i'[\Omega'\setminus W']}(P' \setminus W'))\]
 and compute
 \[\Gamma_i''\lambda_{i \mapsto j}' \coloneqq \{\gamma' \in \Gamma' \mid \gamma'[\FP'] \in \Iso_{\Gamma'[\FP']}(\Fx_i,\Fx_j)\}.\]
 This can be done in time $p^{2} \cdot n^{\CO((\log d)^{c})}$ for some constant $c$ by Theorem \ref{thm:string-isomorphism-gamma-d}.
 
 Now let $i \sim j$ if and only if $\Gamma_i''\lambda_{i \mapsto j}' \neq \emptyset$.
 It is easy to see that $\sim$ defines an equivalence relation.
 Fix an equivalence class $A \subseteq [p]$ of $\sim$.
 
 Without loss of generality assume $1,2 \in A$ (if $|A| = 1$ we may use the only element twice) and fix $\Gamma'' = \Gamma''_1$.
 Moreover let $\varphi'\colon \Gamma'' \rightarrow \Gamma \colon \gamma'' \mapsto (\psi')^{-1}(\gamma'')$.
 Note that $\varphi'$ is a homomorphism.
 Finally, let $\FX''_i \coloneqq (\FX'_i)^{(\lambda_{1 \mapsto i}')^{-1}}$, $i \in A$, and let $\lambda_{i} \coloneqq (\psi')^{-1}(\lambda_{1 \mapsto i}')$.
 
 \begin{claim}
  \label{claim:preserve-isomorphisms-2}
  $\Iso_\Gamma(\FX_i,\FX_j) = \lambda_{i}^{-1}(\Iso_{\Gamma''}(\FX_i'',\FX_j''))^{\varphi'} \lambda_j$.
 \end{claim}
 \begin{claimproof}
  This follows from Claim \ref{claim:preserve-isomorphisms} and the definitions given above.
 \end{claimproof}
 
 \begin{claim}
  \label{claim:respect-block-sizes}
  Let $i \in A$, $P,P' \in \FP'$ and $\gamma \in \Gamma''$ such that $P^{\gamma} = P'$. Then $m_{\FX_i''}(P) = m_{\FX_i''}(P')$ and $m_{\FX''[\Omega' \setminus W']}(P \setminus W') = m_{\FX''[\Omega' \setminus W']}(P' \setminus W')$.
 \end{claim}
 \begin{claimproof}
  This follows directly from the definition of the group $\Gamma''$.
 \end{claimproof}
 
 Now Theorem \ref{thm:renormalize-structure-tree} is applied to the tuple $(\Gamma'',\FP',\FX_1'',\FX_2'')$ and the sequence $\FB_0' \succ \FB_1' \succ \dots \succ \FB_{\ell+1}'$ giving
 a set $\Omega^{*}$,
 a partition $\FP^{*}$ of $\Omega^{*}$,
 a monomorphism $\varphi^{*}\colon \Gamma'' \rightarrow \Sym(\Omega^{*})$,
 a mapping $f\colon \Omega^{*} \rightarrow \Omega'$,
 a sequence of partitions $\{\Omega^{*}\} = \FB_0^{*} \succ \FB_1^{*} \succ \dots \succ \FB_k^{*} = \{\{\alpha\} \mid \alpha \in \Omega^{*}\}$,
 and a mapping $f'\colon \FP^{*} \rightarrow \FP'$,
 such that the following properties are satisfied:
 \begin{enumerate}[label=(\Roman*)]
  \item the sequence $\FB_0^{*} \succ \dots \succ \FB_k^{*}$ forms an almost $d$-ary sequence of $(\Gamma'')^{\varphi^{*}}$-invariant partitions,
  \item $(f(\alpha^{*}))^{\gamma} = f((\alpha^{*})^{\varphi^{*}(\gamma)})$ for every $\alpha^{*} \in \Omega^{*}$ and $\gamma \in \Gamma''$,
  \item $\{f(\alpha^{*}) \mid \alpha^{*} \in P^{*}\} = f'(P^{*})$ for every $P^{*} \in \FP^{*}$,
  \item $|(f')^{-1}(P)| \leq |A|^{\funcnorm(d)}$ where $A$ is the orbit of $(\Gamma'')_B[\FP'[B]]$ containing $P$ where $P \subseteq B \in \FB_{i}'$,
  \item $|P^{*}| = |f'(P^{*})|$ for every $P^{*} \in \FP^{*}$,
  \item there is some $i \in [k]$ such that $\FP^{*} = \FB_{i}^{*}$,
 \end{enumerate}
 Additionally, for every $i \in A$, define
 \[\FX_i^{*} = \{(P^{*},P^{*} \rightarrow \Sigma\colon \alpha^{*} \mapsto \Fx(f(\alpha^{*}))) \mid (f'(P^{*}),\Fx) \in \FX_i''\}.\]
 Then, for all $i,j \in A$,
 \begin{enumerate}[label=(\Roman*)]
  \setcounter{enumi}{7}
  \item $\gamma \in \Iso_{\Gamma''}(\FX_i'',\FX_j'')$ if and only if $\varphi^{*}(\gamma) \in \Iso_{(\Gamma'')^{\varphi^{*}}}(\FX_i^{*},\FX_j^{*})$ for every $\gamma \in \Gamma''$, and
  \item for a $\Gamma''$-invariant window $W' \subseteq \Omega'$ such that $\Gamma''[W'] \leq \Aut_{\Gamma'}(\FX_i''[W'])$ it holds that $(\Gamma'')^{\varphi^{*}}[f^{-1}(W')] \leq \Aut_{(\Gamma'')^{\varphi^{*}}}(\FX_i^{*}[f^{-1}(W')])$.
 \end{enumerate}

 Now let $\Gamma^{*} \coloneqq (\Gamma'')^{\varphi^{*}}$.
 Also let $\varphi \coloneqq (\varphi^{*})^{-1} \circ \varphi'$ which is a monomorphism from $\Gamma^{*}$ to $\Gamma$.
 Finally, let $W^{*} = f^{-1}(W')$.
 
 Property \ref{item:simplify-on-window-1} follows directly from Item \ref{item:renormalize-structure-tree-1}.
 Next, $W^{*}$ is $\Gamma^{*}$-invariant since $W'$ is $\Gamma'$-invariant, $\Gamma'' \leq \Gamma'$ and Item \ref{item:renormalize-structure-tree-2} holds.
 Also, $\Gamma^{*}[W^{*}] \leq \Aut_{\Gamma^{*}[W^{*}]}(\FX^{*}[W^{*}])$ by Claim \ref{claim:properties-updated-window} and Item \ref{item:renormalize-structure-tree-8}.
 Property \ref{item:simplify-on-window-3} follows immediately from Item \ref{item:renormalize-structure-tree-6}.
 Moreover, Property \ref{item:simplify-on-window-4} holds by Claim \ref{claim:preserve-isomorphisms}, Claim \ref{claim:preserve-isomorphisms-2} and Item \ref{item:renormalize-structure-tree-7}.
 Property \ref{item:simplify-on-window-5} follows from Claim \ref{claim:window-simple} and the definition of the set $\FX_i^{*}$.
 
 To prove Property \ref{item:simplify-on-window-6} assume without loss of generality that $i = 1$ and let $s$ the virtual size of $\FX_1$ and $s^{*}$ the virtual size of $\FX_1^{*}$.
 Also note that $\FX_1' = \FX_1''$.
 We have
 \begin{align*}
  s^{*} = &\sum_{P^{*} \in \FP^{*}} |P^{*}| \cdot (m_{\FX_1^{*}}(P^{*}))^{\funcnorm(d) + 1} & \\
        = &\sum_{P' \in \FP'} \sum_{P^{*} \in (f')^{-1}(P')} |P^{*}| \cdot (m_{\FX_1^{*}}(P^{*}))^{\funcnorm(d) + 1} & \text{Item \ref{item:renormalize-structure-tree-5}}\\
        = &\sum_{P' \in \FP'} |P'| \sum_{P^{*} \in (f')^{-1}(P')} (m_{\FX_1^{*}}(P^{*}))^{\funcnorm(d) + 1} & \text{Item \ref{item:renormalize-structure-tree-9} and \ref{item:renormalize-structure-tree-3}}\\
        = &\sum_{P' \in \FP'} |P'| \cdot (m_{\FX'}(P'))^{\funcnorm(d) + 1} \cdot |(f')^{-1}(P')| & \\
        = &\sum_{B \in \FB_i'} \sum_{\substack{A \text{ orbit of}\\(\Gamma'')_B[\FP'[B]]}} \sum_{P' \in A} |P'| \cdot (m_{\FX_1'}(P'))^{\funcnorm(d) + 1} \cdot |(f')^{-1}(P')| & \text{Item \ref{item:renormalize-structure-tree-4}}\\
     \leq &\sum_{B \in \FB_i'} \sum_{\substack{A \text{ orbit of}\\(\Gamma'')_B[\FP'[B]]}} \sum_{P' \in A} |P'| \cdot (m_{\FX_1'}(P'))^{\funcnorm(d) + 1} \cdot |A|^{\funcnorm(d)} & \\
        = &\sum_{B \in \FB_i'} \sum_{\substack{A \text{ orbit of}\\(\Gamma'')_B[\FP'[B]]}} \sum_{P' \in A} \frac{1}{|A|} \cdot |P'| \cdot (m_{\FX_1'}(P') \cdot |A|)^{\funcnorm(d) + 1} & \\
 \end{align*}
 For every $P \in \FP$ let $B_P = P \times \{\Fx[W \cap P] \mid (P,\Fx) \in \FX_1\}$.
 Note that $\FB_i' = \{B_P \mid P \in \FP\}$.
 Moreover,
 \[m_{\FX_1}(P) = \sum_{\substack{A \text{ orbit of}\\(\Gamma'')_{B_P}[\FP'[B_P]]}} \sum_{P' \in A} m_{\FX_1'}(P').\]
 Additionally, for each $A$ orbit of $(\Gamma'')_{B_P}[\FP'[B_P]]$ and $P',P'' \in A$ it holds that $|P| = |P'| = |P''|$ and $m_{\FX_1'}(P') = m_{\FX_1'}(P'') \eqqcolon m_{\FX_1'}(A)$ by Claim \ref{claim:respect-block-sizes}.
 Overall, this gives
 \begin{align*}
  s^{*} \leq &\sum_{B \in \FB_i'} \sum_{\substack{A \text{ orbit of}\\(\Gamma'')_B[\FP'[B]]}} \sum_{P' \in A} \frac{1}{|A|} \cdot |P'| \cdot (m_{\FX_1'}(P') \cdot |A|)^{\funcnorm(d) + 1} & \\
           = &\sum_{P \in \FP} \sum_{\substack{A \text{ orbit of}\\(\Gamma'')_{B_P}[\FP'[B_P]]}} \sum_{P' \in A} \frac{1}{|A|} \cdot |P'| \cdot (m_{\FX_1'}(P') \cdot |A|)^{\funcnorm(d) + 1} & \\
           = &\sum_{P \in \FP} |P| \sum_{\substack{A \text{ orbit of}\\(\Gamma'')_{B_P}[\FP'[B_P]]}} |A| \cdot \frac{1}{|A|} \cdot (m_{\FX_1'}(A) \cdot |A|)^{\funcnorm(d) + 1} & \\
           = &\sum_{P \in \FP} |P| \sum_{\substack{A \text{ orbit of}\\(\Gamma'')_{B_P}[\FP'[B_P]]}} (m_{\FX_1'}(A) \cdot |A|)^{\funcnorm(d) + 1} & \\
        \leq &\sum_{P \in \FP} |P| \left(\sum_{\substack{A \text{ orbit of}\\(\Gamma'')_{B_P}[\FP'[B_P]]}} m_{\FX_1'}(A) \cdot |A|\right)^{\funcnorm(d) + 1} & \\
           = &\sum_{P \in \FP} |P| \cdot (m_\FX(P))^{\funcnorm(d) + 1} = s.
 \end{align*}
 Finally, Property \ref{item:simplify-on-window-7} can be proved analogously.
\end{proof}

\subsection{The Local Certificates Routine}

In this subsection the Local Certificates Routine originally introduced in \cite{Babai16} is lifted to the Generalized String Isomorphism Problem for $\mgamma_d$-group which builds the crucial step of extending the group-theoretic techniques of Babai's quasipolynomial time isomorphism test to the setting of this paper.

\subsubsection{Algorithm}

Let $\Gamma \leq \Sym(\Omega)$ be a permutation group and let $(\Gamma,\FP,\FX,\FY)$ be an instance of the Generalized String Isomorphism Problem.
Furthermore let $\varphi\colon \Gamma \rightarrow S_k$ be a giant representation.
For the description of the Local Certificates Routine we extend the notation of set- and point-wise stabilizers for the group $\Gamma$ to the action on the set $[k]$ defined via the giant representation $\varphi$.
For a set $T \subseteq [k]$ let $\Gamma_T \coloneqq \varphi^{-1}((\Gamma^{\varphi})_T)$ and $\Gamma_{(T)} \coloneqq \varphi^{-1}((\Gamma^{\varphi})_{(T)})$.

The basic approach of the Local Certificates Routine is to consider \emph{test sets} $T \subseteq [k]$ of logarithmic size.

\begin{definition}
 \label{def:full-test-sets}
 A test set $T \subseteq [k]$ is \emph{full} if $(\Aut_{\Gamma_T}(\FX))^{\varphi}[T] \geq \Alt(T)$.
 A \emph{certificate of fullness} is a subgroup $\Delta \leq \Aut_{\Gamma_T}(\FX)$ such that $\Delta^{\varphi}[T] \geq \Alt(T)$.
 A \emph{certificate of non-fullness} is a non-giant $\Lambda \leq \Sym(T)$ such that $(\Aut_{\Gamma_T}(\FX))^{\varphi}[T] \leq \Lambda$.
\end{definition}

The central part of the algorithm is to determine for each test set $T \subseteq [k]$ (of a certain size $t = |T|$ to be determined later) whether $T$ is full and,
depending on the outcome, compute a certificate of fullness or a certificate of non-fullness.
Actually, in order to decide isomorphism, non-fullness certificates are also required for pairs of test sets.
All of this is achieved by the following theorem.

Let $W \subseteq \Omega$ be $\Gamma$-invariant.
Recall that $\Iso_\Gamma^{W}(\FX,\FY) = \{\gamma \in \Gamma \mid \gamma[W] \in \Iso_{\Gamma[W]}(\FX[W],\FY[W])\}$ and $\Aut_\Gamma^{W}(\FX) = \Iso_\Gamma^{W}(\FX,\FX)$.

For $\Delta \leq \Gamma$ define $\Aff(\Delta,\varphi) = \{\alpha \in \Omega \mid \Delta_\alpha^{\varphi} \not\geq A_k\}$ to be the set of points affected by $\varphi$ for the group $\Delta$.
Note that for $\Delta_1 \leq \Delta_2 \leq \Gamma$ it holds $\Aff(\Delta_1,\varphi) \supseteq \Aff(\Delta_2,\varphi)$.

\begin{theorem}
 \label{thm:local-certificates-all-pairs}
 Let $\FP$ be a partition of $\Omega$, $\FX$ and $\FY$ two sets of $\FP$-strings.
 Let $s$ be the $d$-virtual size of $\FX$ and $\FY$.
 Also let $\Gamma \leq \Sym(\Omega)$ a $\mgamma_d$-group that has an almost $d$-ary sequence of partitions $\FB_0 \succ \dots \succ \FB_\ell$ such that $\FP = \FB_i$ for some $i \in [\ell]$.  
 Furthermore suppose there is a giant representation $\varphi\colon \Gamma \rightarrow S_k$ and let $k \geq t > \max\{8,2 + \log_2d\}$.
 Define
 \[\CT = \{(\FX,T),(\FY,T) \mid T \subseteq [k], |T| = t\}.\]
 For every $(\FZ_1,T_1),(\FZ_2,T_2) \in \CT$ one can compute
 \begin{enumerate}[label = (\roman*)]
  \item\label{item:local-certificates-all-pairs-output-1} if $T_1$ is full with respect to $\FZ_1$, a group $\Delta \coloneqq \Delta(\FZ_1,T_1) \leq \Aut_{\Gamma_{T_1}}(\FZ_1)$ such that $\Delta^{\varphi}[T_1] \geq \Alt(T_1)$, or
  \item\label{item:local-certificates-all-pairs-output-2} if $T_1$ is not full with respect to $\FZ_1$, a non-giant group $\Lambda \coloneqq \Lambda(T_1,\FZ_1,T_2,\FZ_2) \leq \Sym(T_1)$ and a bijection $\lambda \coloneqq \lambda(T_1,\FZ_1,T_2,\FZ_2)\colon T_1 \rightarrow T_2$ such that 
   \[\left\{\gamma^\varphi|_{T_1} \,\middle|\, \gamma \in \Iso_\Gamma(\FZ_1,\FZ_2) \wedge T_1^{\left(\gamma^{\varphi}\right)} = T_2\right\} \subseteq \Lambda\lambda.\]
 \end{enumerate}
 Moreover, for every two pairs $(\FZ_1,T_1),(\FZ_2,T_2) \in \CT$ and $(\FZ_1',T_1'),(\FZ_2',T_2') \in \CT$ and isomorphisms $\gamma_i \in \Iso_\Gamma((\FZ_i,T_i),(\FZ_i',T_i'))$, $i \in \{1,2\}$, it holds that
 \begin{enumerate}[label = (\alph*)]
  \item\label{item:local-certificates-all-pairs-property-1} if $T_1$ is full with respect to $\FZ_1$, then $(\Delta(T_1,\FZ_1))^{\gamma_1} = \Delta(T_1',\FZ_1')$, and
  \item\label{item:local-certificates-all-pairs-property-2} if $T_1$ is not full with respect to $\FZ_1$, then
  \[\varphi(\gamma_1^{-1}) \Lambda(T_1,\FZ_1,T_2,\FZ_2)\lambda(T_1,\FZ_1,T_2,\FZ_2) \varphi(\gamma_2) = \Lambda(T_1',\FZ_1',T_2',\FZ_2')\lambda(T_1',\FZ_1',T_2',\FZ_2').\]
 \end{enumerate}
 Moreover, there are numbers $s_1,\dots,s_r \leq s/2$ such that $\sum_{i=1}^{r}s_i \leq 4k^{2t} \cdot t! \cdot s$ and,
 for each $i \in [r]$ using a recursive call to the Generalized String Isomorphism Problem for instances of $d$-virtual size at most $s_i$,
 and $k^{\CO(t)} \cdot t!\cdot (n+m)^{\CO((\log d)^{c})}$ additional computation steps,
 an algorithm can compute all desired objects.
\end{theorem}

\begin{algorithm}
 \footnotesize
 \caption{\textsf{LocalCertificates}}
 \label{alg:local-certificates-all-pairs}
 \DontPrintSemicolon
 \BlankLine
 pick $(\FZ,T) \in \CT$ arbitrary \label{line:fix-test-object} and
 $\psi \colon \Gamma_{T} \rightarrow \Sym(T)$ the homomorphism obtained from $\varphi$ by restricting the image to $T$\;
 $\CU \coloneqq \emptyset$ and $\CS \coloneqq \emptyset$ \tcc*[r]{\small list of unfinished pairs}
 \For{$(\FZ_1,T_1) \in \CT$}{
   compute $\lambda_i \in \Gamma$ such that $T^{\lambda_i^{\varphi}} = T_i$\;
   $\FZ_i' \coloneqq \FZ_i^{\lambda_i^{-1}}$ and $\CR(\FZ_1') \coloneqq (\lambda_1;T_1,\FZ_1)$\;
   add $\FZ_1'$ to $\CS$\;
 }
 $\CU \coloneqq \CS^{2}$ and $\CI(\CS) \coloneqq (\Omega,\FP,\Gamma_{T},\emptyset,\psi,\id,\FB_0,\dots,\FB_\ell)$\;
 \While{$\CU \neq \emptyset$}{ \label{line:invariants-1}
   \For{$\CS$ equivalence class of $\CU$}{
     $(\Omega,\FP,\Lambda,W,\psi,\tau,\FB_0,\dots,\FB_\ell) \coloneqq \CI(\CS)$ \label{line:extend-window-start}\;
     $W' \coloneqq \Aff(\Lambda,\psi)$\;
     \For{$(\FZ_1',\FZ_2') \in \CS^{2}$}{
       $\Lambda'\lambda' \coloneqq {\sf BalanceOrbits}(\Lambda,\FP,\FZ_1',\FZ_2')$ \label{line:balance-orbits}\;
       $\Lambda'\lambda' \coloneqq \{\gamma \in \Lambda'\lambda' \mid \gamma[W'] \in \Iso_{\Lambda'\lambda'[W']}(\FZ_1[W'],\FZ_2[W'])\}$ \label{line:recursive-call}\;
       $\Lambda'\lambda' := {\sf CombineWindows}(\Lambda'\lambda',\FP,\FZ_1,\FZ_2,W,W')$\;\label{line:extend-window-end}
       \eIf{$(\Lambda')^{\psi} \not\geq \Alt(T)$}{ 
         $(\lambda_1;T_1,\FZ_1) \coloneqq \CR(\FZ_1')$ and $(\lambda_2;T_2,\FZ_2) \coloneqq \CR(\FZ_2')$\;
         $\Lambda(\FZ_1,T_1,\FZ_2,T_2) \coloneqq (\lambda_1^{-1}(\Lambda')^{\tau}\lambda_1)^{\varphi}[T_1]$ and $\lambda(\FZ_1,T_1,\FZ_2,T_2) \coloneqq (\lambda_1^{-1}(\lambda')^{\tau}\lambda_2)^{\varphi}[T_1]$ \label{line:output-1}\;
         remove $(\FZ_1',\FZ_2')$ from $\CU$\;
       }{
         $\CJ(\FZ_1',\FZ_2') \coloneqq (\Lambda',\lambda',W')$\;
       }
     }
     \For{$\CS' \subseteq \CS$ equivalence class of $\CU$}{
       pick $\FZ_1' \in \CS'$ and set $(\Lambda,\id,W') \coloneqq \CJ(\FZ_1',\FZ_1')$\;
       \For{$\FZ_2',\FZ_3' \in \CS'$}{
         $(\Lambda,\lambda_2,W') \coloneqq \CJ(\FZ_1',\FZ_2')$ and $(\Lambda,\lambda_3,W') \coloneqq \CJ(\FZ_1',\FZ_3')$\;
         $\FZ_i'' = \FZ_i'^{\lambda_i^{-1}}$, $i \in \{2,3\}$ \label{line:realign}\;
         $(\delta_2;T_2,\FZ_2) \coloneqq \CR(\FZ_2')$ and $(\delta_3;T_3,\FZ_3) \coloneqq \CR(\FZ_3')$\;
         $\CR(\FZ_2'') \coloneqq (\lambda_2\delta_2; T_2,\FZ_2)$ and $\CR(\FZ_3'') \coloneqq (\lambda_3\delta_3 ; T_3,\FZ_3)$\;
         replace $(\FZ_2',\FZ_3')$ by $(\FZ_2'',\FZ_3'')$ to $\CU$\;
       }
       $\CI(\CS') \coloneqq (\Omega,\FP,\Lambda,W \cup W',\psi,\tau|_\Lambda,\FB_0,\dots,\FB_\ell)$ \label{line:extend-w}\;
       ${\sf Simplify}(\CS',\CI,\CR)$ \label{line:simplify}\;
     }
     
     \For{$(\FZ_1',\FZ_2') \in \CU$}{ \label{line:invariants-2}
       $\CS$ equivalence class containing $\FZ_1'$\;
       $(\Omega,\FP,\Lambda,W,\psi,\tau,\FB_0,\dots,\FB_\ell) \coloneqq \CI(\CS)$\;
       $(\lambda_1;T_1,\FZ_1) \coloneqq \CR(\FZ_1')$ and $(\lambda_2;T_2,\FZ_2) \coloneqq \CR(\FZ_2')$\;
       \If{$\Lambda^{\psi} \not\geq \Alt(T)$}{
         $\Lambda(\FZ_1,T_1,\FZ_2,T_2) \coloneqq (\lambda_1^{-1}\Lambda^{\tau}\lambda_1)^{\varphi}[T_1]$ and $\lambda(\FZ_1,T_1,\FZ_2,T_2) \coloneqq (\lambda_1^{-1}\lambda_2)^{\varphi}[T_1]$ \label{line:output-2}\;
         remove $(\FZ_1',\FZ_2')$ from $\CU$\;
       }
       \If{$\Aff(\Lambda,\psi) \subseteq W$}{
         $\Delta(\FZ_1,T_1) \coloneqq \lambda_1^{-1}(\Lambda_{(\Omega \setminus W)})^{\tau}\lambda_1$ and remove $(\FZ_1',\FZ_2')$ from $\CU$ \label{line:output-3}\;
       }
     }
   }
 }
\end{algorithm}

\begin{proof}
 Consider Algorithm \ref{alg:local-certificates-all-pairs}.
 The algorithm essentially performs the standard Local Certificates Routine for all pairs of elements of $\CT$.
 However, a main difference is the simplification routine called in Line \ref{line:simplify}.
 The purpose of this method is to modify all elements of $\CU$ in order to guarantee the applicability of Lemma \ref{la:extending-window-automorphisms}.
 
 Throughout the algorithm we maintain the property that $\CU$ defines an equivalence relation.
 Here, $\CU$ contains all pairs for which the algorithm has not yet computed an output.
 Moreover, the partition into equivalence classes of $\CU$ coarsens the partition into isomorphism classes.
 More precisely, $\CU$ defines an equivalence relation on a subset $\CS$ of $\CT$.
 Pairs of input objects from $\CS$ not contained in $\CU$ where distinguished during the algorithm and for those pairs the empty set is returned as an output for Option \ref{item:local-certificates-all-pairs-output-2}.
 For all other pairs, the algorithm already produced an output. 
 In particular, this means that all equivalence classes of $\CU$ can be treated independently while still assuring Properties \ref{item:local-certificates-all-pairs-property-1} and \ref{item:local-certificates-all-pairs-property-2}.
 
 Actually, the algorithm constantly modifies all input pairs within the equivalence relation $\CU$.
 In order to translate the obtained results back to the original input pairs the algorithm utilizes the function $\CR$.
 Also, for each equivalence class $\CS$ of $\CU$ the algorithm maintains a tuple $\CI(\CS)$ storing information on the current state of computation within the equivalence class.
  
 In order to formalize the proof we first state properties that, whenever the algorithm executes Line \ref{line:invariants-1} or \ref{line:invariants-2}, are supposed to be satisfied for every pair $(\FZ_1',\FZ_2') \in \CU$
 where $\CS$ is the equivalence class containing $\FZ_1',\FZ_2'$,
 \[(\Omega,\FP,\Lambda,W,\psi,\tau,\FB_0,\dots,\FB_\ell) \coloneqq \CI(\CS),\]
 $(\lambda_1;T_1,\FZ_1) \coloneqq \CR(\FZ_1')$ and $(\lambda_2;T_2,\FZ_2) \coloneqq \CR(\FZ_2')$:
 \begin{enumerate}[label=(I.\arabic*)]
  \item\label{item:invariant-local-certificates-all-pairs-0} $\FZ_1'$ and $\FZ_2'$ are sets of $\FP$-strings over ground set $\Omega$, $\Lambda \leq \Sym(\Omega)$, $W \subseteq \Omega$, and $\psi\colon \Lambda \rightarrow \Sym(T)$ is a homomorphism,
  \item\label{item:invariant-local-certificates-all-pairs-1} $W$ is $\Lambda$-invariant,
  \item\label{item:invariant-local-certificates-all-pairs-2} $\Lambda[W] \leq \Aut(\FZ_1'[W])$ and $\FZ_1'[W] = \FZ_2'[W]$,
  \item\label{item:invariant-local-certificates-all-pairs-3} $\FZ_1'[W]$ is simple,
  \item\label{item:invariant-local-certificates-all-pairs-4} for each pair $(\FZ_1,T_1),(\FZ_2,T_2) \in \CT$ for which the desired output has not been computed there is some $(\FZ_1',\FZ_2') \in \CU$ such that $(\delta_1;\FZ_1,T_1) = \CR(\FZ_1')$ and $(\delta_2;\FZ_2,T_2) = \CR(\FZ_2')$,
  \item\label{item:invariant-local-certificates-all-pairs-5} $\tau \colon \Lambda \rightarrow \Gamma$ is a monomorphism such that 
   \[\lambda_1^{-1} (\Iso_{\Lambda}(\FZ_1',\FZ_2'))^{\tau} \lambda_2 = \Iso_{\Gamma}((\FZ_1,T_1),(\FZ_2,T_2)),\]
  \item\label{item:invariant-local-certificates-all-pairs-6} $T$ is $(\Lambda^{\tau})^{\varphi}$-invariant, $(\Lambda^{\tau})^{\varphi}[T] = \Lambda^{\psi}$ and $T^{(\lambda_i^{\varphi})} = T_i$,
  \item\label{item:invariant-local-certificates-all-pairs-7} $\{\Omega\} = \FB_0 \succ \FB_1 \succ \dots \succ \FB_\ell = \{\{\alpha\} \mid \alpha \in \Omega\}$ is an almost $d$-ary sequence of $\Lambda$-invariant partitions such that $\FP = \FB_i$ for some $i \in [\ell]$, and
  \item\label{item:invariant-local-certificates-all-pairs-8} $\FZ_1'$ and $\FZ_2'$ have $d$-virtual size at most $s$.
 \end{enumerate}
 
 We first prove the output of the algorithm is correct provided the properties described above hold throughout the algorithm.
 Let $(\FZ_1',\FZ_2') \in \CU$ and suppose the algorithm produces an output for this pair.
 First suppose this happens in Line \ref{line:output-1}.
 Let $\Lambda$ denote the group stored in $\CI(\CS)$ and $\Lambda'\lambda' \subseteq \Lambda$ be the coset computed before.
 Then $\Iso_{\Lambda}(\FZ_1',\FZ_2') = \Iso_{\Lambda'\lambda'}(\FZ_1',\FZ_2')$ and thus,
 \[\Iso_\Gamma((\FZ_1,T_1),(\FZ_2,T_2)) = \lambda_1^{-1} (\Iso_{\Lambda'\lambda'}(\FZ_1',\FZ_2'))^{\tau} \lambda_2 \subseteq \lambda_1^{-1}\Lambda'\lambda_1\lambda_1^{-1}\lambda'\lambda_2\]
 using Property \ref{item:invariant-local-certificates-all-pairs-5}.
 Also, $((\Lambda')^\tau)^{\varphi}[T] = \Lambda^{\psi} \not\geq \Alt(T)$ is hence, $(\lambda_1^{-1}(\Lambda')^{\tau}\lambda_1)^{\varphi}[T_1]$ is not a giant.
 Hence, Option \ref{item:local-certificates-all-pairs-output-2} is satisfied.
 A similar argument shows the correctness for the output produced in Line \ref{line:output-2}.
 
 So it remains to consider the output produced in Line \ref{line:output-3}.
 Then $W \supseteq \Aff(\Lambda,\psi)$ and thus, $(\Lambda_{(\Omega \setminus W)})^{\psi} \geq \Alt(T)$ by Theorem \ref{thm:unaffected-stabilizer-tree}.
 Note that Theorem \ref{thm:unaffected-stabilizer-tree} is applicable by Property \ref{item:invariant-local-certificates-all-pairs-7}.
 Moreover, $\Lambda_{(\Omega \setminus W)} \leq \Aut(\FZ_1')$ by Properties \ref{item:invariant-local-certificates-all-pairs-1}, \ref{item:invariant-local-certificates-all-pairs-2} and \ref{item:invariant-local-certificates-all-pairs-3} and Lemma \ref{la:extending-window-automorphisms}.
 This implies that
 \[\lambda_1^{-1}(\Lambda_{(\Omega \setminus W)})^{\tau}\lambda_1 \leq \Aut((\FZ_1,T_1))\]
 by Condition \ref{item:invariant-local-certificates-all-pairs-5}.
 Hence, Option \ref{item:local-certificates-all-pairs-output-1} is satisfied by an application of Property \ref{item:invariant-local-certificates-all-pairs-6}.
 
 Next, we verify that Properties \ref{item:local-certificates-all-pairs-property-1} and \ref{item:local-certificates-all-pairs-property-2} are satisfied.
 Let $(\FZ_1,T_1),(\FZ_2,T_2) \in \CT$ and $(\FZ_3,T_3),(\FZ_4,T_4) \in \CT$ and pick isomorphisms $\gamma_i \in \Iso_\Gamma((\FZ_i,T_i),(\FZ_{i+2},T_{i+2}))$, $i \in \{1,2\}$.
 The output for both pairs is generated within the same inner for-loop and elements are within the same equivalence class of $\CU$.
 First suppose the output is generated in Line \ref{line:output-2}.
 For $i \in [4]$ let $(\lambda_i;T_i,\FZ_i) \coloneqq \CR(\FZ_i')$.
 Then
 \[\gamma_1^{-1}\lambda_1^{-1}\Lambda^{\tau}\lambda_2\gamma_2 = \lambda_3^{-1}\Lambda^{\tau}\lambda_4\]
 because $\gamma_1 \in \lambda_1^{-1}\Lambda^{\tau}\lambda_3$ and $\gamma_2 \in \lambda_2^{-1}\Lambda^{\tau}\lambda_4$ by Property \ref{item:invariant-local-certificates-all-pairs-5}.
 Similar arguments can be used to deal with the other cases also exploiting the fact that the output produced by Lemma \ref{la:balance-orbits} in Line \ref{line:balance-orbits} is isomorphism-invariant. 
 
 The next step is to verify that all properties described above are maintained throughout the execution of the algorithm.
 First, it is easy to show that all properties hold at the first execution of Line \ref{line:invariants-1} where $\CU$ consists of a single equivalence class $\CS$.
 We need to show that all properties are again satisfied in Line \ref{line:invariants-2}
 (after that, the algorithm only produces an output for some pairs which does not affect any of the mentioned properties).
 Let $\CS$ be an equivalence class, $(\Omega,\FP,\Lambda,W,\psi,\tau,\FB_0,\dots,\FB_\ell) \coloneqq \CI(\CS)$ and $W' \coloneqq \Aff(\Lambda,\psi)$.
 For each pair we compute the subset $\Lambda'\lambda'$ obtained from balancing orbits (Lemma \ref{la:balance-orbits}) and respecting the window $W \cup W'$.
 During this process, some pairs are detected to be non-isomorphic, for those pairs we return the empty set in Line \ref{line:output-1}.
 This creates a refined equivalence relation.
 For each isomorphism class with respect to the window $W'$ the algorithm \emph{realigns} the sets of $\FP$-strings in Line \ref{line:realign} such that all equivalent sets of $\FP$-strings are identical on the window $W \cup W'$ (ensuring Property \ref{item:invariant-local-certificates-all-pairs-2}).
 It is easy to check that, before the execution of Line \ref{line:simplify}, all properties except for Property \ref{item:invariant-local-certificates-all-pairs-3} are satisfied.
 Hence, the main purpose of the ${\sf Simplify}$-method to be described next is to restore Property \ref{item:invariant-local-certificates-all-pairs-3} while maintaining all other properties.
 The basic idea of the ${\sf Simplify}$-method is to modify all instances based on Theorem \ref{thm:simplify-on-window}.
 In doing so, we may note that some pairs $(\FZ_1',\FZ_2') \in \CU$ cannot be isomorphic, in this case the algorithm outputs the empty set for the corresponding pair and removes it from $\CU$.
 This creates a finer equivalence relation.
 For each equivalence class the subroutine then updates the information associated with this class stored in the table $\CI$ and the table $\CR$.
 The next claim formalizes the properties of the ${\sf Simplify}$-method.
 
 \begin{claim}
  \label{claim:simplify}
  Let $\CS$ be an equivalence class of $\CU$ and let $(\Omega,\FP,\Lambda,W,\psi,\tau,\FB_0,\dots,\FB_\ell) \coloneqq \CI(\CS)$.
  There is a subroutine ${\sf Simplify}$ that updates $\CU$ and computes, for each equivalence class $\CS' \subseteq \CS$, a tuple
  \[(\Omega^{*},\FP^{*},\Lambda^{*},W^{*},\psi^{*},\tau^{*},\FB_0^{*},\dots,\FB^{*}_{\ell^{*}}) = \CI(\CS^{*})\]
  and replaces each pair $(\FZ_1',\FZ_2') \in (\CS')^{2}$, where $(\lambda_1;T_1,\FZ_1) \coloneqq \CR(\FZ_1')$ and $(\lambda_2;T_2,\FZ_2) \coloneqq \CR(\FZ_2')$,
  with a new pair $(\FZ_1^{*},\FZ_2^{*})$ with
  \[(\lambda_1^{*};T_1,\FZ_1) = \CR(\FZ_1^{*}) \text{ and } (\lambda_2^{*};T_2,\FZ_2) = \CR(\FZ_2^{*})\]
  resulting in the class $\CS^{*}$.
  
  After the update all required properties are satisfied.
  Additionally, for $s$ the virtual size of $\FZ_1'[\Omega \setminus W]$ and $s^{*}$ the virtual size of $\FZ_1^{*}[\Omega^{*} \setminus W^{*}]$ it holds that $s^{*} \leq s$.
  
  Moreover, this subroutine can be implemented in time $k^{\CO(t)} \cdot (n+m)^{\CO((\log d)^{c})}$ for some constant $c$.
 \end{claim}
 \begin{claimproof}
  We fix some numbering
  \[\CS = \{\FZ_1',\dots,\FZ_p'\}\]
  of the set $\CS$.
  Let $(\lambda_i;T_i,\FZ_i) \coloneqq \CR(\FZ_i')$ for all $i \in [p]$.
  Consider the tuple $(\Lambda,\FP,W,(\FZ_i')_{i \in [p]})$ together with the sequence $\FB_0,\dots,\FB_\ell$ to which we can apply Theorem \ref{thm:simplify-on-window} exploiting Properties \ref{item:invariant-local-certificates-all-pairs-1}, \ref{item:invariant-local-certificates-all-pairs-2} and \ref{item:invariant-local-certificates-all-pairs-7}.
  
  First, this gives an equivalence relation $\sim$ on the set $[p]$ (which can also be viewed as an equivalence relation on $\CS$ in the natural way).
  For each $i,j \in [p]$ such that $i \not\sim j$ the algorithm sets $\Lambda(\FZ_i,T_i,\FZ_j,T_j)\lambda(\FZ_i,T_i,\FZ_j,T_j) \coloneqq \emptyset$ since the corresponding objects can not be isomorphic.
  Moreover, we can consider all equivalence classes independently.
  Let $\CS' \subseteq \CS$ be an equivalence class of $\sim$.
  
  For the class $\CS'$ we obtain a set $\Omega^{*}$,
  a group $\Lambda^{*} \leq \Sym(\Omega^{*})$,
  elements $\lambda_{i}' \in \Lambda$ for all $i \in \CS'$,
  a monomorphism $\tau'\colon \Lambda^{*} \rightarrow \Lambda$,
  a window $W^{*} \subseteq \Omega^{*}$,
  a sequence of partitions $\{\Omega^{*}\} = \FB_0^{*} \succ \FB_1^{*} \succ \dots \succ \FB_{\ell^{*}}^{*} = \{\{\alpha\} \mid \alpha \in \Omega^{*}\}$,
  and a list of $\FP^{*}$-strings $(\FZ_i^{*})_{i \in \CS'}$.
  
  We set $\CS^{*} = \{\FZ_i^{*} \mid i \in \CS'\}$.
  Also
  \[\CI(\CS^{*}) \coloneqq (\Omega^{*},\FP^{*},\Lambda^{*},W^{*},\tau' \circ \psi,\tau' \circ \tau,\FB_0^{*},\dots,\FB^{*}_{\ell^{*}}).\]
  Moreover, for $i \in \CS'$, define
  \[\CR(\FZ_i^{*}) = ((\lambda_i')^{\tau}\lambda_i;T_i,\FZ_i).\]
  
  We need to verify that the desired properties are satisfied.
  Property \ref{item:invariant-local-certificates-all-pairs-0} is immediately clear.
  Properties \ref{item:invariant-local-certificates-all-pairs-1} and \ref{item:invariant-local-certificates-all-pairs-2} follow from Theorem \ref{thm:simplify-on-window}, Property \ref{item:simplify-on-window-2}.
  Also, \ref{item:invariant-local-certificates-all-pairs-3} follows directly from Theorem \ref{thm:simplify-on-window}, Property \ref{item:simplify-on-window-5}.
  Property \ref{item:invariant-local-certificates-all-pairs-4} is directly clear from the construction.
  For Property \ref{item:invariant-local-certificates-all-pairs-5} it holds that
  \[\Iso_{\Gamma}((\FZ_i,T_i),(\FZ_j,T_j)) = \lambda_i^{-1} (\Iso_{\Lambda}(\FZ_i',\FZ_j'))^{\tau} \lambda_j = \lambda_i^{-1} \left((\lambda_i')^{-1}(\Iso_{\Lambda^{*}}(\FZ_i^{*},\FZ_j^{*}))^{\tau'}\lambda_j'\right)^{\tau} \lambda_j\]
  since Property \ref{item:invariant-local-certificates-all-pairs-5} holds before the update and exploiting Theorem \ref{thm:simplify-on-window}, Property \ref{item:simplify-on-window-4}.
  Property \ref{item:invariant-local-certificates-all-pairs-6} is not affected by the update, so it still holds afterwards.
  Next, Property \ref{item:invariant-local-certificates-all-pairs-7} follows from Theorem \ref{thm:simplify-on-window}, Property \ref{item:simplify-on-window-1}.
  Finally, Property \ref{item:invariant-local-certificates-all-pairs-8} and the additional property stated in the claim follow from Theorem \ref{thm:simplify-on-window}, Properties \ref{item:simplify-on-window-6} and \ref{item:simplify-on-window-7}.
  
  To finish the proof, the bound on the running time directly follows from the corresponding bound in Theorem \ref{thm:simplify-on-window}.
  Note that $|\CS'| \leq |\CT| = k^{\CO(t)}$.
 \end{claimproof}
 
 Finally, we concern ourselves with the termination of the while loop and the running time of the algorithm.
 We first argue that the while-loop terminates after at most $s$ iterations.
 Observe that this also implies that an output is computed for each pair in $\CT$ by Property \ref{item:invariant-local-certificates-all-pairs-4}.
 
 Towards this end, consider some pair $(T_1,\FZ_1),(T_2,\FZ_2) \in \CT$ and let $(\FZ_1',\FZ_2') \in \CU$ such that $(\delta_i;\FZ_i,T_i) = \CR(\FZ_i')$ for $i \in \{1,2\}$ (in some iteration of the algorithm).
 Also let $\CS$ be the equivalence class containing $\FZ_1'$ and $\FZ_2'$ and $(\Omega,\FP,\Lambda,W,\psi,\tau,\FB_0,\dots,\FB_\ell) \coloneqq \CI(\CS)$.
 The $d$-virtual size of $\FZ_1'$ is at most $s$ throughout the algorithm by Property \ref{item:invariant-local-certificates-all-pairs-8}.
 On the other hand let $s'$ be the $d$-virtual size of $\FZ_1'[\Omega \setminus W]$.
 Since $W' \neq \emptyset$ it holds that $s'$ decreases in Line \ref{line:extend-w} by Observation \ref{obs:virtual-size-vs-domain-size}.
 Moreover, when updating all structures in Line \ref{line:simplify} it holds that $s'$ does not increase by Claim \ref{claim:simplify}.
 So overall, $s'$ decreases in each iteration and initially $s' = s$.
 This implies the bound on the number of iterations.
 In particular, the algorithm terminates and produces an output for each pair of elements from $\CT$. 
 
 Moreover, each iteration runs in the desired time by Theorem \ref{thm:permutation-group-library}, Lemma \ref{la:balance-orbits} and \ref{la:combine-windows}, and Claim \ref{claim:simplify}.
 Note that the $d$-virtual size of any set $\FP$-strings involved in the algorithm never exceeds $s$ by Property \ref{item:invariant-local-certificates-all-pairs-8}.
 Also note that $|\CU| = k^{\CO(t)}$ throughout the algorithm.
 
 So it remains to analyze the recursive calls the algorithm makes.
 First note that the only place where recursive calls are made is in Line \ref{line:recursive-call}.
 Let $(\FZ_1,T_1),(\FZ_2,T_2) \in \CT$.
 We prove that there are numbers $s_1,\dots,s_r \leq s/2$ such that $\sum_{i \in [r]} s_i \leq t! \cdot s$
 and $s_1,\dots,s_r$ are (upper bounds on) the $d$-virtual sizes of the instances occurring in recursive calls made for pairs $(\FZ_1',\FZ_2')$ such that $\CR(\FZ_i') = (\delta_i;T_i,\FZ_i)$, $i \in \{1,2\}$.
 Since $|\CT^{2}| \leq (2k^{t})^{2} = 4k^{2t}$ this implies the desired bound on the recursive calls.
 
 Let $s_i'$ be the $d$-virtual of $\FZ_1'[W']$ in iteration $i$, $i \in [q]$.
 We first argue that $\sum_{i \in [q]} s_i' \leq s$.
 Towards this end, let $t_i$ be the $d$-virtual size of $\FZ_1'[\Omega \setminus W]$ at the beginning of the $i$-th iteration of the while-loop.
 We prove by induction on $p \in \{0,\dots,q\}$ that $\sum_{i \in [p]} s_i + t_{p+1} \leq s$.
 The base case is immediately clear.
 So suppose $p \geq 1$.
 It suffices to prove that $s_{p+1} + t_{p+2} \leq t_{p+1}$.
 Let $t'$ be the $d$-virtual sizes virtual of $\FZ_1'[\Omega \setminus (W \cup W')]$ in iteration $p$ (before the execution of Line \ref{line:simplify}).
 Then $s_{p+1} + t' \leq t_{p+1}$ by Lemma \ref{la:virtual-size-for-partition}.
 Moreover, $t_{p+2} \leq t'$ by Claim \ref{claim:simplify}.
 
 Now let $i \in [q]$.
 In order to realize Line \ref{line:recursive-call} the algorithm actually performs multiple recursive calls.
 Let $W_1',\dots,W_k'$ be the partition of $W'$ into orbits.
 Let $s_{i,j}'$ be the $d$-virtual size of $\FZ_1'[W_j']$ (in iteration $i$).
 Note that $\sum_{j \in [k]}s_{i,j}' \leq s_i'$ be Lemma \ref{la:virtual-size-for-partition}.
 Using Lemma \ref{la:combine-windows} it suffices to consider each window $W_{j}'$ independently.
 Let $N \trianglelefteq \Lambda$ be the kernel of the homomorphism $\psi$.
 Note that $|\Lambda : N| \leq t!$.
 Hence, computing $\Iso_{\Lambda[W_j']}(\FZ_1'[W_j'],\FZ_2'[W_j'])$ amounts solving at most $t!$ instances of the Generalized String Isomorphism Problem over domain size $s_{i,j}'$ and input group $N[W_j']$.
 Also, by Theorem \ref{thm:kernel-affected-orbits}, each orbit $A_p$ of $N[W_j']$ has size at most $|W_j'|/t \leq |W_j'|/2$.
 Hence, it suffices to perform at most $t!$ recursive calls to the Generalized String Isomorphism Problem with domain $A_p$.
 Moreover, $\FZ_1'[W_j']$ is balanced by Lemma \ref{la:balance-orbits}.
 Thus, the $d$-virtual size of $\FZ_1'[A_p]$ is bounded by $\frac{|A_p|}{|W_j'|} s_{i,j}' \leq \frac{1}{2} s_{i,j}' \leq \frac{1}{2} s$ by Observation \ref{obs:virtual-size-balanced}.
 Overall, this proves the desired bounds on the recursive calls of the algorithm and completes the proof.
\end{proof}

\subsubsection{Aggregating Local Certificates}

Having generalized the Local Certificates Routine to the setting of hypergraphs, the rest of the section providing an algorithm for the Generalized String Isomorphism Problem for $\mgamma_d$-groups is analogous to the corresponding algorithm for the String Isomorphism Problem presented in \cite{GroheNS18} (see \cite{GroheNS18-full} for a full version of this paper).
This first step is to aggregate the local certificates computed above.

Let $\FP$ be a partition of $\Omega$, $\FX,\FY$ two sets of $\FP$-strings, $\Gamma \leq \Sym(\Omega)$ a group that has an almost $d$-ary sequence of partitions $\FB_0 \succ \dots \succ \FB_\ell$ such that $\FB_i = \FP$ for some $i \in [\ell]$, and $\varphi\colon \Gamma \rightarrow S_k$ a giant representation.
In this situation the goal is either to find a group $\Delta \leq \Aut_\Gamma(\FX)$ such that $\Delta^{\varphi}$ is large or computing (a small set of) isomorphism-invariant relational structures $\FA_i$, $i \in \{1,2\}$, defined on the set $[k]$ that are far away from being symmetric.
To achieve this goal, an algorithm computes local certificates for all test sets and all pairs of test sets of logarithmic size.
The certificates of fullness can be combined into a group of automorphisms whereas certificates of non-fullness may be combined into a relational structure that has only few automorphisms.
The next definition formalizes the target for a relational structure to only have few automorphisms.

\begin{definition}[Symmetry Defect]
 Let $\Gamma \leq \Sym(\Omega)$ be a group.
 The \emph{symmetry defect} of $\Gamma$ is the minimal $t \in [n]$ such that there is a set $M \subseteq \Omega$ of size $|M| = n-t$ such that $\Alt(M) \leq \Gamma$ (the group $\Alt(M)$ fixes all elements of $\Omega\setminus M$).
 In this case the \emph{relative symmetry defect} of $\Gamma$ is $t/n$.
 
 For any relational structure $\FA$ we define the \emph{(relative) symmetry defect} of $\FA$ to be the (relative) symmetry defect of its automorphism group $\Aut(\FA)$.
\end{definition}

A crucial property is that groups of large symmetry defect are much smaller than the giants.

\begin{theorem}[cf.\;\cite{DixonM96}, Theorem 5.2 A,B]
 \label{thm:symmetry-defect-subgroup-small-index}
 Let $A_n \leq S \leq S_n$ and suppose $n > 9$.
 Let $\Gamma \leq S$ and $r < n/2$.
 Suppose that $|S : \Gamma| < \binom{n}{r}$.
 Then the symmetry defect of $\Gamma$ is strictly less than $r$.
\end{theorem}

Actually, the following corollary turns out to be sufficient for us.

\begin{corollary}
 \label{cor:index-subgroup-large-symmetry-defect}
 Let $A_n \leq S \leq S_n$ be a giant group and suppose $n \geq 24$.
 Let $\Gamma \leq S$ and suppose the relative symmetry defect of $\Gamma$ is at least $1/4$.
 Then $|S:\Gamma| \geq (4/3)^{n}$.
\end{corollary}

\begin{proof}
 Let $r = \lfloor n/4\rfloor$.
 Then the symmetry defect of $\Gamma$ is at least $r$.
 Hence, $|S : \Gamma| \geq \binom{n}{r}$ by Theorem \ref{thm:symmetry-defect-subgroup-small-index}.
 Moreover,
 \[\binom{n}{\lfloor n/4\rfloor} \geq \left(\frac{n}{\lfloor n/4\rfloor}\right)^{\lfloor n/4\rfloor} \geq 4^{(n/4) - 1} = \frac{1}{4} \cdot \sqrt{2}^{n}.\]
 Since $n \geq 24$ it holds that $\frac{1}{4} \cdot \sqrt{2}^{n} \geq (4/3)^{n}$.
\end{proof}

\begin{lemma}
 \label{la:aggregate-local-certificates}
 Let $\FP$ be a partition of $\Omega$, $\FX_1,\FX_2$ two sets of $\FP$-strings of $d$-virtual size $s$, and $\Gamma \leq \Sym(\Omega)$ a group that has an almost $d$-ary sequence of partitions $\FB_0 \succ \dots \succ \FB_\ell$ such that $\FB_i = \FP$ for some $i \in [\ell]$.
 Furthermore suppose there is a giant representation $\varphi\colon \Gamma \rightarrow S_k$.
 Let $\max\{8,2 + \log_2d\} < t < k/10$.
 
 Then there are natural numbers $r \in \mathbb{N}$ and $s_1,\dots,s_r \leq s/2$ such that $\sum_{i=1}^{r}s_i \leq k^{\mathcal{O}(t)}s$ and,
 for each $i \in [r]$ using a recursive call to the Generalized String Isomorphism Problem for instances of $d$-virtual size at most $s_i$, and $k^{\mathcal{O}(t)}n^{c}$ additional computation,
 one obtains for $i=1,2$ one of the following:
 \begin{enumerate}
  \item\label{item:aggregate-local-certificates-1} a family of $r\le k^6$ many $t$-ary relational structures $\mathfrak{A}_{i,j}$, for $j\in[r]$, associated with $\FX_i$,
   each with domain $V_{i,j}\subseteq[k]$ of size $|V_{i,j}| \geq \frac{3}{4}k$ and with relative symmetry defect at least $\frac{1}{4}$ such that
   \[\left\{\mathfrak{A}_{1,1},\dots,\mathfrak{A}_{1,r}\right\}^{\varphi(\gamma)} = \left\{\mathfrak{A}_{2,1},\dots,\mathfrak{A}_{2,r}\right\} \text{ for every } \gamma \in \Iso_\Gamma(\FX_1,\FX_2),\]
   or
  \item\label{item:aggregate-local-certificates-2} a subset $M_i \subseteq [k]$ associated with $\FX_i$ of size $|M_i| \geq \frac{3}{4}k$ and $\Delta_i \leq \Aut_{\Gamma_{M_i}}(\FX_i)$
   such that $(\Delta_i^{\varphi})[M_i] \geq \Alt(M_i)$ and
   \[M_1^{\varphi(\gamma)} = M_2 \text{ for every } \gamma \in \Iso_\Gamma(\FX_1,\FX_2).\]
 \end{enumerate}
\end{lemma}

The proof is completely analogous to the proof of \cite[Theorem 24]{Babai16} (and also \cite[Lemma VI.3]{GroheNS18}) replacing the methods to compute the local certificates.
In order to present the proof we first require some additional background.

\begin{definition}[Degree of Transitivity]
 A permutation group $\Gamma \leq \Sym(\Omega)$ is \emph{$t$-transitive} if its natural induced action on the set of $n(n-1)\dots(n-t+1)$ ordered $t$-tuples of distinct elements is transitive.
 The \emph{degree of transitivity} $d(\Gamma)$ is the largest $t$ such that $\Gamma$ is $t$-transitive.
\end{definition}

\begin{theorem}[CFSG]
 \label{thm:degree-of-transitivity}
 Let $\Gamma \leq \Sym(\Omega)$ be a non-giant group.
 Then $d(\Gamma) \leq 5$.
\end{theorem}

A slightly weaker statement, namely $d(\Gamma) \leq 7$ for all non-giants permutation groups, can be shown using only Schreier's Hypothesis (see \cite[Theorem 7.3A]{DixonM96}).

A graph $G$ is \emph{regular} if every vertex has the same degree $d$.
A regular graph on $n$ vertices is \emph{non-trivial} if $0 < d < n-2$, i.e., the graph $G$ is not the complete graph and contains at least one edge.

\begin{lemma}[cf.\;\cite{Babai15}, Corollary 2.4.13]
 \label{la:symmetry-defect-regular-graph}
 Let $G$ be a non-trivial regular graph. Then the relative symmetry defect of $G$ is at least $1/2$.
\end{lemma}

\begin{proof}[Proof of Lemma \ref{la:aggregate-local-certificates}]
 Let $\CT = \{(\FX_1,T),(\FX_2,T) \mid T \subseteq [k], |T| = t\}$.
 For every pair of elements of $\CT$ the algorithm computes a certificate of fullness or non-fullness using Theorem \ref{thm:local-certificates-all-pairs}.
 Let $\Lambda_i \leq \Sym(\Omega)$ be the group generated by the fullness-certificates for all full subsets $T \subseteq [k]$ with respect to $\FX_i$.
 Note that $\Lambda_i \leq \Aut_\Gamma(\FX_i)$ for both $i \in [2]$.
 Also observe that the group $\Lambda_i$ is defined in an isomorphism-invariant manner meaning that $\gamma^{-1}\Lambda_1\gamma = \Lambda_2$ for all $\gamma \in \Iso_\Gamma(\FX_1,\FX_2)$.
 Let $S_i \subseteq [k]$ be the support of $\Lambda_i^{\varphi}$, i.e., $S_i \coloneqq \{\alpha \in [k] \mid |\alpha^{\left(\Lambda_i^{\varphi}\right)}| \geq 2\}$.
 Once again, the sets $S_i$ are isomorphism-invariant meaning that $S_1^{\varphi(\gamma)} = S_2$ for every $\gamma \in \Iso_\Gamma(\FX_1,\FX_2)$.
 In particular, we may assume that $|S_1| = |S_2|$ (otherwise $\FX_1 \not\cong_\Gamma \FX_2$ and the algorithm outputs some trivial non-isomorphic relational structures).
 Now we distinguish between three cases depending on the size of $S_i$.
 
 \begin{cs}
  \case{$\frac{1}{4}k \leq |S_i| \leq \frac{3}{4}k$}
   This case is simple by setting $r \coloneqq 1$, $\FA_{1,1} \coloneqq ([k],S_1)$, and $\FA_{2,1} \coloneqq ([k],S_2)$.
   It is easy to verify this satisfies Option \ref{item:aggregate-local-certificates-1} of the Lemma.
  \case{$|S_i| > \frac{3}{4}k$}
   We further distinguish between three subcases.
   First assume every orbit of $\Lambda_i^{\varphi}$ has size at most $\frac{3}{4}k$.
   Then the partition into the orbits of $\Lambda_i^{\varphi}$ gives a canonical structure $\FA_{i,1}$ with domain $[k]$ and relative symmetry defect at least $\frac{1}{4}$.
   More precisely, set $r \coloneqq 1$ and define
   \[\FA_{i,1} \coloneqq (S_i,\{(\alpha,\beta) \mid \alpha \in S_i, \beta \in \alpha^{\left(\Lambda_i^{\varphi}\right)}\}).\]
   
   So suppose there is a (unique) orbit $M_i \subseteq [k]$ of size $M_i \geq \frac{3}{4}k$.
   If $\Lambda_i^{\varphi}[M_i] \geq \Alt(M_i)$ then the second option of the Lemma is satisfied.
   
   Hence suppose $\Lambda_i^{\varphi}[M_i]$ is not a giant.
   By Theorem \ref{thm:degree-of-transitivity} the degree of transitivity satisfies $d(\Lambda_i^{\varphi}[M_i]) \leq 5$.
   Let $F_i \subseteq M_i$ be an arbitrary set of size $d(\Lambda_i^{\varphi}[M_i]) - 1$ and individualize the elements of $F_i$.
   Then $(\Lambda_i^{\varphi})_{(F_i)}[M_i']$ is transitive, but not $2$-transitive, where $M_i' = M_i \setminus F_i$.
   Note that the number of possible choices for the set $F_i$ is at most $k^{4}$.
   Now let $\FG_i = (M_i',R_{i,1},\dots,R_{i,p})$ be the \emph{orbital configuration} of $(\Lambda_i^{\varphi})_{(F_i)}$ on the set $M_i'$, that is, the relations $R_{i,j}$ are the orbits of $(\Lambda_i^{\varphi})_{(F_i)}$ in its natural action on $M_i' \times M_i'$.
   Note that $p \geq 3$ since $(\Lambda_i^{\varphi})_{(F_i)}[M_i']$ is not $2$-transitive.
   Also observe that the numbering of the $R_{i,j}$, $j \in [p]$, is not canonical (isomorphisms may permute the $R_{i,j}$).
   Without loss of generality suppose that $R_{i,1}$ is the diagonal.
   Now individualize one of the $R_{i,j}$ for $j \geq 2$ at a multiplicative cost of $p-1 \leq k-1$.
   If $R_{i,j}$ is undirected (i.e., $R_{i,j} = R_{i,j}^{-1}$) then it defines a non-trivial regular graph.
   Since the symmetry defect of this graph is at least $1/2$ (see Lemma \ref{la:symmetry-defect-regular-graph}) this gives us the desired structure.
   Otherwise $R_{i,j}$ is directed.
   If the out-degree of a vertex is strictly less $(|M_i'|-1)/2$ then the undirected graph $G_i = (M_i',R_{i,j} \cup R_{i,j}^{-1})$ is again a non-trivial regular graph.
   Otherwise, by individualizing one vertex (at a multiplicative cost of $|M_i'| \leq k$), one obtains a coloring of symmetry defect at least $1/2$ by coloring vertices depending on whether they are an in- or out-neighbor of the individualized vertex.
  \case{$|S_i| < \frac{1}{4}k$}
   Let $D_i = [k] \setminus S_i$.
   Then $|D_1| = |D_2| \geq \frac{3}{4}k$.
   Observe that every $T \subseteq D_i$ is not full with respect to $\FX_i$.
   Let $D_i' = D_i \times \{i\}$ (to make the sets disjoint).
   
   Consider the following category $\mathcal{L}$.
   The objects are the pairs $(T,i)$ where $T \subseteq D_i$ is a $t$-element subset.
   The morphisms $(T,i) \rightarrow (T',i')$ are the bijections computed in Theorem \ref{thm:local-certificates-all-pairs} for the test sets $T$ and $T'$ along with the corresponding sets of $\FP$-strings.
   The morphisms define an equivalence relation on the set $(D_1')^{\angles{t}} \cup (D_2')^{\angles{t}}$
   where $(D_i')^{\angles{t}}$ denotes the set of all ordered $t$-tuples with distinct elements over the set $D_i'$.
   Let $R_1,\dots,R_r$ be the equivalence classes and define $R_j(i) = R_j \cap (D_i')^{\angles{t}}$.
   Then $\mathfrak{A}_i = (D_i',R_1(i),\dots,R_r(i))$ is a canonical $t$-ary relational structure.
   Moreover, the symmetry defect of $\mathfrak{A}_i$ is at least $|D_i| - t +1 \geq |D_i|/4$.
 \end{cs}
\end{proof}

The aggregation process described in Lemma \ref{la:aggregate-local-certificates} either gives a small set of isomorphism-invariant structures with large symmetry defect or a large set of automorphisms.
Both outcomes can be used to significantly reduce the size of the group $\Gamma$ as detailed in the following two lemmas.

\begin{lemma}
 \label{la:find-structure}
 Suppose Option \ref{item:aggregate-local-certificates-1} of Lemma \ref{la:aggregate-local-certificates} is satisfied,
 yielding a number $r \leq k^6$ and a set of relational structures $\FA_{i,j}$ for $i \in [2], j \in [r]$.
 
 Then there are subgroups $\Lambda_j \leq \Gamma$ and elements $\lambda_j \in \Sym(\Omega)$ for $j \in [r]$ such that 
 \begin{equation}
  \Iso_{\Gamma}(\FX_1,\FX_2) = \bigcup_{j \in [r]} \Iso_{\Lambda_j\lambda_j}(\FX_1,\FX_2),
 \end{equation}
 and $|\Gamma^{\varphi} : \Lambda_j^{\varphi}| \geq (4/3)^{k}$ for all $j \in [r]$.
 
 Moreover, given all the relational structures $\FA_{i,j}$ for $i \in [2]$, $j \in [r]$, the groups $\Lambda_j$ and elements $\lambda_j$ can be computed in time $k^{\CO(t^{c} (\log k)^{c})}n^{c}$ for some constant $c$.
\end{lemma}

\begin{proof}
 Let $V_{i,j} \coloneqq V(\FA_{i,j}) \subseteq [k]$ be the domain of $\FA_{i,j}$ for all $i \in [2]$ and $j \in [r]$.
 Let $\FA_1 \coloneqq \FA_{1,1}$ and also $V_1 \coloneqq V_{1,1}$.
 Now define \[\Lambda_j\lambda_j \coloneqq \{\gamma \in \Gamma \mid (V_1)^{\left(\gamma^{\varphi}\right)} = V_{2,j} \wedge (\gamma^{\varphi})|_{V_1} \in \Iso(\FA_1,\FA_{2,j})\}.\]
 Using the quasipolynomial time isomorphism test from \cite{Babai16} the set $\Iso(\mathfrak{A}_1,\mathfrak{A}_{2,j})$ can be computed in time $k^{\mathcal{O}(t^{c} (\log k)^{c})}$ for some constant $c$
 (first translate the relational structures into two graphs of size $k^{\mathcal{O}(t)}$ (see, e.g., \cite{Miller79}) and then apply the isomorphism test from \cite{Babai16} to the resulting graphs).
 Hence, a representation for the sets $\Lambda_j\lambda_j$ can be computed in time $k^{\mathcal{O}(t^{c} (\log k)^{c})} n^{c}$ for some constant $c$ by Theorem \ref{thm:permutation-group-library}.
 
 Also $\FA_1^{\varphi(\gamma)} \in \{\FA_{2,1},\dots,\FA_{2,r}\}$ for every $\gamma \in \Iso_\Gamma(\FX_1,\FX_2)$.
 This implies that
 \[\Iso_\Gamma(\FX_1,\FX_2) = \bigcup_{j \in [r]} \Iso_{\Lambda_j\lambda_j}(\FX_1,\FX_2).\]
 Finally recall that the symmetry defect of $\Lambda_j^{\varphi}$ is at least $\frac{1}{4}$.
 So $|\Gamma^{\varphi} : \Lambda_j^{\varphi}| \geq (4/3)^{k}$ by Corollary \ref{cor:index-subgroup-large-symmetry-defect}.
\end{proof}

\begin{lemma}
 \label{la:find-symmetry}
 Suppose Option \ref{item:aggregate-local-certificates-2} of Lemma \ref{la:aggregate-local-certificates} is satisfied,
 yielding sets $M_i \subseteq [k]$ and groups $\Delta_i \leq \Aut_{\Gamma_{M_i}}(\FX_i)$ for $i \in \{1,2\}$.
 Then there is a number $r \in \{1,2\}$, a subgroup $\Lambda \leq \Gamma$ and elements $\lambda_j \in \Sym(\Omega)$ for $j \in [r]$ such that
 \begin{enumerate}
  \item $\FX_1 \cong_\Gamma \FX_2$ if and only if $\FX_1 \cong_{\Lambda\lambda_j} \FX_2$  for some $j \in [r]$, 
    and given representations for the sets $\Iso_{\Lambda\lambda_j}(\FX_1,\FX_2)$ for all $j \in [r]$
    and a generating set for $\Delta_1$ one can compute in polynomial time a representation for $\Iso_\Gamma(\FX_1,\FX_2)$, and
  \item $|\Gamma^{\varphi} : \Lambda^{\varphi}| \geq (4/3)^{k}$.
 \end{enumerate}
 Moreover, given the sets $M_i$ for both $i \in \{1,2\}$, the group $\Lambda$ and the elements $\lambda_j$ can be computed in polynomial time.
\end{lemma}

\begin{proof}
 Let $\Lambda = \Gamma_{(M_1)}$ (recall that $\Gamma_{(T)} = \varphi^{-1}((\Gamma^{\varphi})_{(T)})$ for $T \subseteq [k]$).
 Pick $\gamma \in \Gamma$ such that $M_1^{\varphi(\gamma)} = M_2$ and $\tau \in \Gamma_{M_1}$ such that $\tau^{\varphi}[M_1]$ is a transposition.
 Now define $\lambda_1 = \gamma$ and $\lambda_2 = \tau \gamma$.
 Then $\FX_1 \cong_\Gamma \FX_2$ if and only if $\FX_1 \cong_{\Lambda\lambda_j} \FX_2$ for some $j \in \{1,2\}$ since $(\Delta_1^{\varphi})[M_1] \geq \Alt(M_1)$.
 Moreover, if $\Gamma_j\gamma_j = \Iso_{\Lambda\lambda_j}(\FX_1,\FX_2)$ then $\Iso_\Gamma(\Fx_1,\Fx_2) = \bigcup_{j=1,2}\langle \Delta_1,\Gamma_j \rangle \gamma_j$.
 Finally, $|\Gamma^{\varphi} : \Lambda^{\varphi}| \geq |\Alt(M_1)| \geq (4/3)^{k}$.
\end{proof}

\subsection{An Algorithm for the Generalized String Isomorphism Problem}

After extending the Local Certificates Routine to the setting of hypergraphs we are now ready to formalize the main algorithm solving the Generalized String Isomorphism Problem for $\mgamma_d$-groups.
Recall that we already showed that it suffices to consider permutation groups that are equipped with an almost $d$-ary sequence of partitions (Theorem \ref{thm:normalize-generalized-string-isomorphism-instance}).
The basic strategy to tackle the problem for such groups is to follow Luks's algorithm along the given sequence of partitions.
Whenever feasible the algorithm simply performs standard Luks reduction (see Subsection \ref{subsec:luks-recursion}).
Otherwise, when the recurrence obtained from the standard Luks reduction does not yield the desired running time, we can find a giant representation allowing us to apply the Local Certificates Routine.
This is formalized by the next lemma.

\begin{lemma}[cf.\ \cite{Babai15}, Theorem 3.2.1]
 \label{la:compute-giant-representation}
 Let $\Gamma \leq S_d$ be a primitive group of order $|\Gamma| \geq d^{1+\log d}$ where $d$ is greater than some absolute constant.
 Then there is a polynomial-time algorithm computing a normal subgroup $N \leq \Gamma$ of index $|\Gamma:N| \leq d$, an $N$-invariant equipartition $\FB$,
 and a giant representation $\varphi\colon N \rightarrow S_k$ where $k \geq \log d$ and $\ker(\varphi) = N_{(\FB)}$.
\end{lemma}

\begin{lemma}
 \label{lem:partition-small-size-recursion}
 Let $\FP$ be a partition of $\Omega$.
 Also, let $\Gamma \leq \Sym(\Omega)$ be a transitive $\mgamma_d$-group and $\FX,\FY$ two sets of $\FP$-strings of $d$-virtual size $s$.
 Moreover, suppose $\Gamma$ has an almost $d$-ary sequence of partitions $\FB_0 \succ \dots \succ \FB_\ell$ such that $|\FB_1| \leq d$ and there is some $i \in [\ell]$ with $\FP = \FB_i$.
 
 Then there are natural numbers $r \in \mathbb{N}$ and $s_1,\dots,s_r \leq s/2$ such that $\sum_{i=1}^{r}s_i \leq 2^{\mathcal{O}((\log d)^{3})}s$ and,
 for each $i \in [r]$ making one recursive call to the Generalized String Isomorphism Problem for instances of $d$-virtual size at most $s_i$, and $(n+m)^{\CO((\log d)^{c})}$ additional computation,
 one can compute a representation for $\Iso_\Gamma(\FX,\FY)$.
\end{lemma}

\begin{proof}
 Let $\FB \succ \FB_1$ be a minimal block system of the group $\Gamma$ and let $\Delta \coloneqq \Gamma[\FB]$ denote the induced action of $\Gamma$ on $\FB$.
 Note that $|\FB| \leq d$.
 If $|\Delta| \leq d^{1+\log d}$ the statement of the lemma follows from applying standard Luks reduction (see Subsection \ref{subsec:luks-recursion}).
 Otherwise, using Lemma \ref{la:compute-giant-representation}, the algorithm computes a normal subgroup $N \leq \Delta$ of index $|\Delta:N| \leq d$, an $N$-invariant equipartition $\FC$,
 and a giant representation $\psi\colon N \rightarrow S_k$ where $k \geq \log d$ and $\ker(\psi) = N_{(\FB)}$.
 First observe $k \leq d$ since the permutation degree of $N$ is bounded by $d$.
 We lift the normal subgroup $N$ and the partition $\FC$ from $\Delta$ to $\Gamma$ obtaining a group $\Gamma' = \{\gamma \in \Gamma \mid \gamma[\FB] \in N\}$ and a partition $\FC' = \{\bigcup_{B \in C} B \mid C \in \FC\}$ (recall that $\FC$ forms a partition of the set $\FB$).
 Clearly $\FC'$ is $\Gamma'$-invariant.
 Since $|\Gamma:\Gamma'| \leq d$ it suffices to prove the statement for the group $\Gamma'$ (introducing an additional factor of $d$ for the number of recursive calls).
 Finally, we can also lift the homomorphism $\psi$ to the group $\Gamma'$ obtaining a giant representation $\varphi\colon \Gamma' \rightarrow S_k\colon \gamma \mapsto (\gamma[\FB])^{\psi}$.
 Note that $\varphi$ is a giant representation and $\ker(\varphi) = (\Gamma')_{(\FC')}$.
 Let $t \coloneqq \max\{9,3 + \log d\}$.
 In case $k \leq 10t$ the statement follows again by standard Luks reduction.
 (In this case $|\Gamma' : (\Gamma')_{(\FC')}| = |\Gamma' : \ker(\varphi)| \leq k! \leq 2^{\CO((\log d)^{2})}$.)
 So suppose $\max\{8,2 + \log d\} < t < k/10$.
 In this case the requirements of Lemma \ref{la:aggregate-local-certificates} are satisfied.
 
 Using Lemma \ref{la:aggregate-local-certificates}, \ref{la:find-structure} and \ref{la:find-symmetry} we can reduce the problem
 (using additional recursive calls to the Generalized String Isomorphism Problem for instances of $d$-virtual size at most $s/2$)
 to at most $k^{6}$ instances of $\Lambda$-isomorphism over the same sets of $\FP$-strings $\FX$ and $\FY$ for groups $\Lambda \leq \Gamma'$ where $|(\Gamma')^{\varphi} : \Lambda^{\varphi}| \geq (4/3)^{k}$.
 Applying the same argument to these instances of $\Lambda$-isomorphism and repeating the process until we can afford to perform standard Luks reduction gives our desired algorithm.
 
 It remains to analyze its running time, that is, we need to analyze the number of times this process has to be repeated until the algorithm reaches a sufficiently small group to perform standard Luks reduction.
 Towards this end, we analyze the parameter $k$ of the giant representation and show that it has to be reduced in each round by a certain amount.
 Recall that the algorithm performs standard Luks reduction as soon as $k \leq 10t$.
 
 Consider the recursion tree of the algorithm (ignoring the additional recursive calls to the Generalized String Isomorphism Problem for instances of $d$-virtual size at most $s/2$ for the moment).
 Recall that $\FC'$ is $\Gamma'$-invariant and thus, it is also $\Lambda$-invariant.
 In case $\Lambda$ is not transitive it is processed orbit by orbit.
 Note that there is at most one orbit $W$ such that $\FX[W]$ has $d$-virtual size greater than $s/2$ (cf.\ Lemma \ref{la:virtual-size-for-partition}) that has to be considered in the current recursion
 (for the other orbits additional recursive calls to the Generalized String Isomorphism Problem for instances of $d$-virtual size at most $s/2$ suffice together with an application of Lemma \ref{la:combine-windows} and these recursive calls are ignored for the moment).
 Let $\varphi'\colon \Lambda' \rightarrow S_{k'}$ be the giant representation computed on the next level of the recursion where $\Lambda'$ is the projection of $\Lambda''$ to an invariant subset of the domain for some $\Lambda'' \leq \Lambda$
 (if no giant representation is computed then the algorithm performs standard Luks reduction and the node on the next level is a leaf).
 Observe that $|\Lambda'[\FC']| \geq \frac{(k')!}{2}$ because $(\Lambda')^{\varphi'} \geq A_{k'}$ and $\Lambda'_{(\FC')} \leq \ker(\varphi')$.
 Also note that $|\Lambda[\FC']| \leq \frac{k!}{(4/3)^{k}}$ since $\ker(\varphi) = \Gamma'_{(\FC')}$ by Lemma \ref{la:compute-giant-representation}.
 So \[\frac{(k')!}{2} \leq \frac{k!}{(4/3)^{k}}.\]
 Hence,
 \[(4/3)^{k} \leq 2 \cdot 2^{(k - k') \log k} \leq (4/3)^{3(k-k')\log k}\]
 since $k$ is sufficiently large.
 So \[k' \leq k - \frac{k}{3 \log k}.\]
 It follows that the height of the recursion tree is $\mathcal{O}((\log d)^{2})$.
 Thus, the number of nodes of the recursion tree is bounded by $d^{\mathcal{O}((\log d)^{2})} = 2^{\mathcal{O}((\log d)^{3})}$.
 By Lemma \ref{la:aggregate-local-certificates}, \ref{la:find-structure} and \ref{la:find-symmetry} each node of the recursion tree makes recursive calls to the Generalized String Isomorphism for instances of $d$-virtual size $s_i \leq s/2$
 where $\sum_i s_i \leq 2^{\mathcal{O}((\log d)^{2})}s$ and uses additional computation $(n+m)^{\mathcal{O}((\log d)^{c})}n^{c}$.
 Putting this together, the desired bound follows.
\end{proof}

\begin{theorem}
 \label{thm:generalized-string-isomorphism-gamma-d-almost-d-ary}
 There is an algorithm that, given a partition $\FP$ of $\Omega$, a $\mgamma_d$-group $\Gamma \leq \Sym(\Omega)$, two sets of $\FP$-strings $\FX,\FY$ and an almost $d$-ary sequence of partitions $\FB_0 \succ \dots \succ \FB_\ell$ for $\Gamma$ such that $\FP = \FB_i$ for some $i \in [\ell]$,
 computes a representation for $\Iso_\Gamma(\FX,\FY)$ in time $(n+m)^{\CO((\log d)^{c})}$, for an absolute constant $c$.
\end{theorem}

\begin{algorithm}
 \caption{Generalized String Isomorphism}
 \label{alg:generalized-string-isomorphism-alg}
 \DontPrintSemicolon
 \SetKwInOut{Input}{Input}
 \SetKwInOut{Output}{Output}
 \Input{$\FP$ partition of $\Omega$, $\Gamma \leq \Sym(\Omega)$ a $\mgamma_d$-group, $\FX,\FY$ two sets of $\FP$-strings
        and an almost $d$-ary sequence of $\Gamma$-invariant partitions $\{\Omega\} \succ \FB_1 \dots \succ \FB_\ell = \{\{\alpha\} \mid \alpha \in \Omega\}$ such that $\FP = \FB_i$ for some $i \in [\ell]$}
 \Output{$\Iso_\Gamma(\FX,\FY)$}
 \BlankLine
 \eIf{$\Gamma$ is not transitive}{
  recursively process group orbit by orbit \tcc*[r]{\small restrict partitions to orbits}
  combine results using Lemma \ref{la:combine-windows}\;
  \Return $\Iso_\Gamma(\FX,\FY)$\;
 }{
  \eIf{$\Gamma[\FB_1]$ is semi-regular}{
   apply standard Luks reduction \tcc*[r]{\small restrict partitions to orbits of $\Gamma_{(\FB_1)}$}
   \Return $\Iso_\Gamma(\Fx,\Fy)$\;
  }(\tcc*[f]{\small assumptions of Lemma \ref{lem:partition-small-size-recursion} are satisfied}){
   apply Lemma \ref{lem:partition-small-size-recursion}\;
   \Return $\Iso_\Gamma(\Fx,\Fy)$\;
  }
 }
\end{algorithm}

\begin{proof}
 The pseudo-code is given in Algorithm \ref{alg:generalized-string-isomorphism-alg}.
 If the input group $\Gamma$ is not transitive the group is processed orbit by orbit (see Subsection \ref{subsec:orbit-by-orbit}).
 If the action of $\Gamma$ on the block system $\FB_1$ is semi-regular, the algorithm applies standard Luks reduction to compute the set $\Iso_\Gamma(\FX,\FY)$ (see Subsection \ref{subsec:luks-recursion}).
 Otherwise $\Gamma$ is transitive and $|\FB_1| \leq d$ (recall that $\{\Omega\} \succ \FB_1 \succ \dots \succ \FB_m = \{\{\alpha\} \mid \alpha \in \Omega\}$ is an almost $d$-ary sequence of $\Gamma$-invariant partitions).
 Then Lemma \ref{lem:partition-small-size-recursion} can be applied to recursively compute $\Iso_\Gamma(\FX,\FY)$.
 
 Clearly, the algorithm computes the desired set of isomorphisms.
 To bound the running time let $s$ be the $d$-virtual size of $\FX$ and $\FY$.
 By Lemma \ref{la:recursion-bound} the recurrence of the algorithm yields a running time of $(m+n+s)^{\CO((\log d)^{c})}$.
 Exploiting the definition of the $d$-virtual size and $\funcnorm(d) = \CO(\log d)$ (see Lemma \ref{la:construct-structure-graph-transitive}) this gives the desired bound.
 Note that the bottleneck is the type of recursion used in Lemma~\ref{lem:partition-small-size-recursion}.
 
 Finally, observe every group $\Delta$, for which the algorithm performs a recursive call, is equipped with an almost almost $d$-ary sequence of $\Gamma$-invariant partitions $\FB_0 \succ \FB_1 \dots \succ \FB_\ell$ such that $\FP = \FB_i$ for some $i \in [\ell]$.
 This follows from Observation \ref{obs:sequence-of-partitions} and the proof of Theorem \ref{thm:local-certificates-all-pairs}.
\end{proof}

Combining Theorem \ref{thm:normalize-generalized-string-isomorphism-instance} and \ref{thm:generalized-string-isomorphism-gamma-d-almost-d-ary} gives the main technical result of this work.

\begin{theorem}
 \label{thm:generalized-string-isomorphism-gamma-d}
 The Generalized String Isomorphism Problem for $\mgamma_d$-groups can be solved in time $(n+m)^{\CO((\log d)^{c})}$ for some constant $c$.
\end{theorem}

As an immediate consequence we obtain one of the main results of this paper.

\begin{corollary}[Theorem \ref{thm:main} restated]
 \label{cor:hypergraph-isomorphism-gamma-d}
 The Hypergraph Isomorphism Problem for $\mgamma_d$-groups can be solved in time $(n+m)^{\CO((\log d)^{c})}$ for some constant $c$.
\end{corollary}

\section{Allowing Color Refinement to Split Small Color Classes}
\label{sec:t-cr-bounded}

In the following two sections we present some simple consequences of the improvement obtained for the Hypergraph Isomorphism Problem for $\mgamma_d$-groups.
In this section, we consider $t$-CR-bounded graphs originally introduced by Ponomarenko\footnote{In \cite{Ponomarenko89} $t$-CR-bounded graphs are referred to as graphs with property $\Pi(0,t)$.} in \cite{Ponomarenko89}, and obtain an isomorphism test for such graphs running in time $n^{\CO((\log t)^{c})}$.
In the next section, we use the improved isomorphism test for $t$-CR-bounded graphs to obtain an algorithm testing isomorphism of graphs of genus at most $g$ in time $n^{\CO((\log g)^{c})}$ (this connection was already observed in \cite{Ponomarenko89}).

\subsection{The Color Refinement Algorithm}

The definition of $t$-CR-bounded graphs builds on the Color Refinement algorithm, a simple combinatorial algorithm that iteratively refines a vertex-coloring in an isomorphism-invariant manner and which forms a fundamental algorithmic tool in the context of the Graph Isomorphism Problem.
We start by formally defining the outcome of the Color Refinement algorithm.

Let $G$ be a graph with vertex coloring $\chi_V \colon V(G) \rightarrow C_V$  and arc coloring $\chi_E \colon \{(v,w) \mid vw \in E(G)\} \rightarrow C_E$.
The \emph{Color Refinement algorithm} is a procedure that, given a vertex- and arc-colored graph $G$,
iteratively computes an isomorphism-invariant refinement $\ColRef{G}$ of the vertex-coloring $\chi_V$.

Let $\chi_1,\chi_2 \colon V \rightarrow C$ be colorings of vertices where $C$ is some finite set of colors.
The coloring $\chi_1$ \emph{refines} $\chi_2$, denoted $\chi_1 \preceq \chi_2$, if $\chi_1(v) = \chi_1(w)$ implies $\chi_2(v) = \chi_2(w)$ for all $v,w \in V$.
Observe that $\chi_1 \preceq \chi_2$ if and only if the partition into color classes of $\chi_1$ refines the corresponding partition into color classes of $\chi_2$.
The colorings $\chi_1$ and $\chi_2$ are \emph{equivalent}, denoted $\chi_1 \equiv \chi_2$, if $\chi_1 \preceq \chi_2$ and $\chi_2 \preceq \chi_1$.

Given a vertex- and arc-colored graph $G$, the Color Refinement algorithm computes a coloring $\ColRef{G}$ as follows.
The initial coloring for the algorithm is defined as $\ColRefIt{G}{0} \coloneqq \chi_V$, the vertex-coloring of the input graph.
The initial coloring is refined by iteratively computing colorings $\ColRefIt{G}{i}$ for $i > 0$.
For $i > 0$ and $v \in V(G)$ we define $\ColRefIt{G}{i}(v) \coloneqq (\ColRefIt{G}{i-1}(v),\mathcal{M}_{i}(v))$ where
\[\mathcal{M}_{i}(v) \coloneqq \left\{\!\left\{\bigl(\ColRefIt{G}{i-1}(w),\chi_E(v,w),\chi_E(w,v)\bigr) \mid w \in N_G(v)\right\}\!\right\}.\]
From the definition of the colorings it is immediately clear that $\ColRefIt{G}{i+1} \preceq \ColRefIt{G}{i}$.
Now let $i \in \NN$ be the minimal number such that $\ColRefIt{G}{i} \equiv \ColRefIt{G}{i+1}$.
For this $i$, the coloring $\ColRefIt{G}{i}$ is called the \emph{stable} coloring of $G$ and is denoted by $\ColRef{G}$.

The \emph{Color Refinement algorithm} takes as input a (vertex- and arc-colored) graph $G$ and computes (a coloring that is equivalent to) $\ColRef{G}$.
I remark that this can be implemented in time almost linear in the number of vertices and edges (see, e.g., \cite{BerkholzBG17}).

\subsection{Splitting Small Color Classes}

Having defined the Color Refinement algorithm, we can now define the notion of $t$-CR-bounded graphs.
The basic idea behind $t$-CR-bounded graphs is the following.
Suppose $(G,\chi)$ is a vertex-colored graph.
Then $(G,\chi)$ is $t$-CR-bounded if it is possible to transform $\chi$ into a discrete coloring (i.e., a coloring where each vertex has its own color) by the following two operations: applying the Color Refinement algorithm and completely splitting color classes of size at most $t$ (i.e., assigning every vertex in such a color class its own color).

To extend the applicability of our results, we actually define $t$-CR-bounded pairs where an additional set $S \subseteq V(G)$ is provided each vertex of which may also be individualized.
The intuition behind this is that we may already have good knowledge about the structure of the automorphism group of $(G,\chi)$ on the set $S$ which can be exploited for isomorphism testing.

\begin{definition}
 \label{def:t-cr-bounded}
 Let $G = (V,E,\chi_V,\chi_E)$ be a vertex- and arc-colored graph and $S \subseteq V(G)$ a set of vertices.
 We define a sequence of colorings $(\chi_i)_{i \geq 0}$ where
 \[\chi_0(v) \coloneqq \begin{cases}
                (v,1)         & \text{if } v \in S\\
                (\chi_V(v),0) & \text{otherwise}
               \end{cases}\]
 and
 \[\chi_{2i+1} \coloneqq \ColRef{(V,E,\chi_{2i},\chi_E)}\]
 and
 \[\chi_{2i+2}(v) \coloneqq \begin{cases}
                     (v,1)              & \text{if } |\chi_{2i+1}^{-1}(\chi_{2i+1}(v))| \leq t\\
                     (\chi_{2i+1}(v),0) & \text{otherwise}
                    \end{cases}\]
 for all $i \geq 0$.
 For the minimal $i \geq 0$ such that $\chi_i \equiv \chi_{i+1}$, we call the coloring $\chi_i$ the \emph{$t$-CR-stable} coloring of $G$ and denote it by $\tColRef{t}{G;S}$.
 If $S = \emptyset$ we also write $\tColRef{t}{G}$ instead of $\tColRef{t}{G;S}$.
 
 A pair $(G,S)$ is \emph{$t$-CR-bounded} if $\tColRef{t}{G;S}$ is discrete.
 A graph $G$ is \emph{$t$-CR-bounded} if the pair $(G,\emptyset)$ is $t$-CR-bounded.
\end{definition}

I remark that the name ``$t$-CR-bounded'' is inspired by color-$t$-bounded graphs \cite{BabaiCSTW13} defined in a similar manner where the letters \emph{CR} refer to the Color Refinement algorithm.

\begin{lemma}
 \label{la:extend-gamma-d-group-t-cr}
 Let $G_1,G_2$ be two vertex- and arc-colored graphs and let $S_1 \subseteq V(G_1)$ and $S_2 \subseteq V(G_2)$.
 Also let $\Gamma \leq \Sym(S_1)$ be a $\mgamma_t$-group and $\theta\colon S_1 \rightarrow S_2$ a bijection.
 Let $\FP_i$ be the partition into color classes of the coloring $\tColRef{t}{G_i;S_i}$.
 
 Then a $\mgamma_t$-group $\Delta \leq \Sym(\FP_1)$ and a bijection $\delta\colon \FP_1 \rightarrow \FP_2$ such that
 \[\big(\Iso_{\Gamma\theta}((G_1,S_1),(G_2,S_2))\big)[\FP_1] \coloneqq \big(\{\sigma \colon G_1 \cong G_2 \mid \sigma|_{S_1} \in \Gamma\theta\}\big)[\FP_1] \subseteq \Delta\delta\]
 can be computed in time $n^{\CO((\log t)^{c})}$ for some absolute constant $c$.
\end{lemma}

\begin{proof}
 For $i \in \{1,2\}$ suppose $G_i = (V_i,E_i,\chi_V^{i},\chi_E^{i})$ and moreover let $(\chi_j^{i})_{j \geq 0}$ be the sequence of colorings from the definition of $t$-CR-bounded pairs for the pair $(G_i,S_i)$.
 For $j \geq 0$ let $\FP_j^{i} \coloneqq \{(\chi_j^{i})^{-1}(c) \mid c \in \im(\chi_j^{i})\}$ be the partition into the color classes of $\chi_j^{i}$.
 Also let $G_{j}^{i} = (\FP_j^{i},E_j^{i},\chi_V^{j,i},\chi_E^{j,i})$ be the vertex- and arc-colored graph defined by
 \[E_j^{i} = \{PQ \mid \exists v \in P, w \in Q \colon vw \in E(G_i)\},\]
 \[\chi_V^{j,i}(P) = \{\!\{ \chi_V^{i}(v) \mid v \in P\}\!\},\]
 and
 \[\chi_E^{j,i}(P,Q) = \{\!\{ \chi_E^{i}(v,w) \mid v \in P, w \in Q, vw \in E(G_i)\}\!\}.\]
 
 The algorithm inductively computes representations for the sets
 \[\Gamma_j\theta_j \coloneqq \Iso_{\Gamma\theta}((G_j^{1},S_1),(G_j^{2},S_2))\]
 for $j \geq 0$ (identifying the singleton $\{v\}$ with the element $v$ for all $v \in S$).
 Here, $\Gamma_j \in \mgamma_t$ for every $j \geq 0$.
 Since all graphs are defined in an isomorphism-invariant manner and there is some $j \leq n$ such that $\chi_j^{i} \equiv \tColRef{t}{G_i;S_i}$, this implies the lemma.
 
 For the base step we define the group
 \[\Gamma_0 \coloneqq \{\gamma \in \Sym(\FP_0^{1}) \mid \exists \delta \in \Gamma \colon \forall v \in S \colon\{v\}^{\gamma} = \{v^{\delta}\} \wedge \forall P \in \FP_0^{1} \setminus \{\{v\} \mid v \in S\}\colon P^{\gamma} = P\}\]
 and the bijection
 \[\theta_0\colon \FP_0^{1} \rightarrow \FP_0^{2}\colon P \mapsto \begin{cases}
                                                                   \{v^{\theta}\} &\text{if } \exists v \in S\colon P = \{v\}\\
                                                                   (\chi_0^{2})^{-1}(\chi_0^{1}(P)) &\text{otherwise}
                                                                  \end{cases}.\]
 
 For the inductive step $j > 0$ we need to distinguish two cases depending on the parity of $j$.
 First consider the more simple case that $j$ is even.
 In this case
 \[\FP_{j}^{i} = \{P \in \FP_{j-1}^{i} \mid |P| > t\} \cup \{\{v\} \mid v \in P, P \in \FP_{j-1}^{i}, |P| \leq t\}.\]
 Also, for every $\sigma \in \Gamma_{j-1}\theta_{j-1}$ and every $P \in \FP_{j-1}^{1}$ it holds that $|P| = |P^{\sigma}|$ since $\chi_V^{j-1,1}(P) = \chi_V^{j-1,2}(P^{\sigma})$.
 
 Now let $\theta_j'\colon \FP_j^{1} \rightarrow \FP_j^{2}$ be a bijection such that, for every $P \in \FP_{j-1}^{1}$, it holds that $\FP_j^{1}[P^{\theta_{j-1}}] = (\FP_j^{1}[P])^{\theta_j'}$.
 Also define
 \[\Gamma_j' \coloneqq \{\gamma' \in \Sym(\FP_{j}^{1}) \mid \exists \gamma \in \Gamma_{j-1} \forall P \in \FP_{j-1}^{1}\colon \FP_j^{1}[P^{\gamma}] = (\FP_j^{1}[P])^{\gamma'}\}.\]
 Note that $\Gamma_j' \in \mgamma_t$.
 Now
 \[\Gamma_j\theta_j = \{\sigma \in \Gamma_j'\theta_j' \mid \sigma\colon G_j^{1} \cong G_j^{2}\}\]
 which can be computed within the desired time by Corollary \ref{cor:graph-isomorphism-gamma-d}.
 
 For the second case suppose that $j$ is odd meaning that $\chi_j^{i} = \ColRef{(V_i,E_i,\chi_{j-1}^{i},\chi_E^{i})}$.
 We consider the rounds of the Color Refinement algorithm.
 For $r \geq 0$ define $\chi_{j,r}^{i} \coloneqq \ColRefIt{(V_i,E_i,\chi_{j-1}^{i},\chi_E^{i})}{r}$ to be the coloring computed in the $r$-th round of the Color Refinement algorithm.
 Also let $\FP_{j,r}^{i} \coloneqq \{(\chi_{j,r}^{i})^{-1}(c) \mid c \in \im(\chi_{j,r}^{i})\}$ be the corresponding partition into color classes.
 
 Moreover, similar to the definitions above, let $G_{j,r}^{i} = (\FP_{j,r}^{i},E_{j,r}^{i},\chi_V^{j,r,i},\chi_E^{j,r,i})$ be the vertex- and arc-colored graph defined by
 \[E_{j,r}^{i} = \{PQ \mid \exists v \in P, w \in Q \colon vw \in E(G_i)\},\]
 \[\chi_V^{j,r,i}(P) = \{\!\{ \chi_V^{i}(v) \mid v \in P\}\!\},\]
 and
 \[\chi_E^{j,r,i}(P,Q) = \{\!\{ \chi_E^{i}(v,w) \mid v \in P, w \in Q, vw \in E(G_i)\}\!\}.\]
 
 To complete the proof it suffices to inductively compute representations for the sets $\Gamma_{j,r}\theta_{j,r} \coloneqq \Iso_{\Gamma\theta}((G_{j,r}^{1},S_1),(G_{j,r}^{2},S_2))$.
 Observe that $\Gamma_j\theta_j = \Gamma_{j,r}\theta_{j,r}$ for the minimal $r \leq n$ such that $\chi_{j,r}^{i} = \ColRef{(V_i,E_i,\chi_{j-1}^{i},\chi_E^{i})}$ for both $i \in \{1,2\}$ (if the number of rounds of the Color Refinement algorithm differs for the two graphs, they are not isomorphic).
 
 For the base step $r = 0$ note that $\Gamma_{j,0}\theta_{j,0} = \Gamma_{j-1}\theta_{j-1}$.
 So let $r > 0$.
 For each $v \in V(G_i)$ we define a string
 \[\Fx_v^{i}\colon \FP_{j,r-1}^{i} \rightarrow C \colon P \mapsto \begin{cases}
                                                                   (1,\{\!\{ (\chi_E^{i}(v,w),\chi_E^{i}(w,v)) \mid w \in P \cap N_{G_i}(v) \}\!\}) &\text{if } v \in P \\
                                                                   (0,\{\!\{ (\chi_E^{i}(v,w),\chi_E^{i}(w,v)) \mid w \in P \cap N_{G_i}(v) \}\!\}) &\text{if } v \notin P
                                                                  \end{cases}.\]
 It is easy to see that $\Fx_v^{i} = \Fx_w^{i}$ if and only if $\chi_{j,r}^{i}(v) = \chi_{j,r}^{i}(w)$ for all $v,w \in V(G_i)$.
 Now let
 \[\Gamma_{j,r}'\theta_{j,r}' \coloneqq \Iso_{\Gamma_{j,r-1}\theta_{j,r-1}}(\{\Fx_v^{1} \mid v \in V(G_1)\},\{\Fx_v^{2} \mid v \in V(G_2)\})\]
 which can be computed in time $n^{\CO((\log t)^{c})}$ by solving an instance of the Hypergraph Isomorphism Problem for $\mgamma_t$-groups (see Corollary \ref{cor:hypergraph-isomorphism-gamma-d}).
 
 Since there is a one-to-one correspondence between $\{\Fx_v^{i} \mid v \in V(G_i)\}$ and $\FP_{j,r}^{i}$, $i \in \{1,2\}$, we can view $\Gamma_{j,r}'$ as a subgroup of $\Sym(\FP_{j,r}^{1})$ and
 $\theta_{j,r}'$ as a bijection from $\FP_{j,r}^{1}$ to $\FP_{j,r}^{2}$.
 Now
 \[\Gamma_{j,r}\theta_{j,r} = \Iso_{\Gamma_{j,r}'\theta_{j,r}'}(G_{j,r}^{1},G_{j,r}^{2})\]
 which can be computed in time $n^{\CO((\log t)^{c})}$ by Corollary \ref{cor:graph-isomorphism-gamma-d}.
\end{proof}

\begin{theorem}
 \label{thm:isomorphism-t-cr-bounded-pairs}
 Let $(G_1,S_1)$ and $(G_2,S_2)$ be two $t$-CR-bounded pairs and also let $\Gamma \leq \Sym(S_1)$ be a $\mgamma_t$-group and $\theta\colon S_1 \rightarrow S_2$ a bijection.
 Then a representation for the set $\Iso_{\Gamma\theta}((G_1,S_1),(G_2,S_2))$ can be computed in time $n^{\CO((\log t)^{c})}$ for some absolute constant $c$.
\end{theorem}

\begin{proof}
 Using Lemma \ref{la:extend-gamma-d-group-t-cr} we compute a $\mgamma_t$-group $\Delta \leq \Sym(V(G_1))$ and a bijection $\delta\colon V(G_1) \rightarrow V(G_2)$ such that
 \[\Iso_{\Gamma\theta}((G_1,S_1),(G_2,S_2)) \subseteq \Delta\delta.\]
 Using Theorem \ref{cor:graph-isomorphism-gamma-d}, we then compute the set of isomorphisms in the desired time.
\end{proof}

\begin{corollary}
 \label{cor:isomorphism-t-cr-bounded}
 The Graph Isomorphism Problem for $t$-CR-bounded graphs can be solved in time $n^{\CO((\log t)^{c})}$ for some absolute constant $c$.
\end{corollary}

I remark that the above algorithm actually describes a polynomial-time Turing reduction from the isomorphism problem for $t$-CR-bounded graphs to the Hypergraph Isomorphism Problem for $\mgamma_t$-groups.
This result may be of independent interest.

\begin{lemma}
 There is a polynomial-time Turing reduction from the Graph Isomorphism Problem for $t$-CR-bounded graphs to the Hypergraph Isomorphism Problem for $\mgamma_t$-groups.
\end{lemma}

\section{Isomorphism for Graphs of Bounded Genus}
\label{sec:genus}

In this section, we present another application of our results and give an algorithm solving the isomorphism problem for graphs of Euler genus at most $g$ in time $n^{\CO((\log g)^{c})}$ for some absolute constant $c$.
Actually, we prove a more general result.

Recall that a graph $H$ is a minor of another graph $G$ if $H$ can be obtained from $G$ by removing vertices as well as removing and contracting edges.
Also define $K_{m,n}$ to be the complete bipartite graphs with $m$ vertices on the left side and $n$ vertices on the right side.
Let $h \geq 3$ and define $\CC_h$ to be the class of graphs that exclude $K_{3,h}$ as a minor.

We present a polynomial-time reduction from the isomorphism problem for the class $\CC_h$ to the isomorphism problem for $t$-CR-bounded graphs where $t \coloneqq h - 1$.
As a result, we obtain an algorithm solving the isomorphism problem for $\CC_h$ in time $n^{\CO((\log h)^{c})}$ for some absolute constant $c$.
In turn, this provides an algorithm for testing isomorphism of bounded genus graphs via the following theorem.

\begin{theorem}[Ringel \cite{Ringel65}]
 \label{thm:genus-complete-bipartite}
 The complete bipartite graph $K_{m,n}$, $m,n \geq 2$, has Euler genus
 \[g = \left\lceil\frac{(m-2)(n-2)}{4}\right\rceil.\]
\end{theorem}

For a more recent reference on this result I refer to \cite{Harary91}.
Hence, the class of graphs of Euler genus at most $g$ excludes $K_{3,4g + 3}$ as a minor.

In order to reduce the isomorphism problem for $\CC_h$ to the problem of testing isomorphism between $t$-CR-bounded graphs, the following lemma is the key.
Intuitively, it investigates the structure of certain colorings that are stable with respect to the Color Refinement algorithm for graphs in the class $\CC_h$.
Recall that a graph $G$ is \emph{$3$-connected} if there are no two vertices $v,w \in V(G)$ such that $G - \{v,w\}$ is disconnected.

\begin{lemma}
 \label{la:extend-t-cr-bounded-scope}
 Let $(G,\chi)$ be a 3-connected, colored graph that excludes $K_{3,h}$ as a minor and suppose $V_1 \uplus V_2 = V(G)$ such that
 \begin{enumerate}
  \item each $v \in V_1$ is a singleton with respect to $\chi$,
  \item $\chi$ is stable with respect to the Color Refinement algorithm,
  \item $|V_1| \geq 3$, and
  \item $N(V_2) = V_1$.
 \end{enumerate}
 Then there is a color class $U \subseteq V_2$ with respect to $\chi$ of size $|U| \leq h-1$.
\end{lemma}

\begin{figure}
 \centering
 \begin{tikzpicture}
  \fill[opacity=0.2,gray,rounded corners] (0,0) -- (1.0,0.56) -- (4.0,0.56) -- (4.0,-0.56) -- (1.0,-0.56) -- cycle;
  \fill[opacity=0.2,gray,rounded corners] (-0.2,1.2) -- (0.8,1.76) -- (4.3,1.76) -- (4.3,0.64) -- (0.8,0.64) -- cycle;
  \fill[opacity=0.2,gray,rounded corners] (-0.2,-1.2) -- (0.8,-1.76) -- (6.4,-1.76) -- (6.4,-1.2) -- (5.7,0.56) -- (5.3,0.56) -- (4.8,-0.64) -- (0.8,-0.64) -- cycle;
  
  \draw (0,0) ellipse (0.8cm and 1.8cm);
  \draw[rounded corners] (1.2,-1.8) rectangle (7.6,1.8);
  
  \node at (-1.2,0) {$V_1$};
  \node at (8.0,0) {$V_2$};
  
  \node[ellipse,minimum height=30pt,minimum width=10pt,align=center,lightcolor1] (v1) at (2.5,0) {};
  \node[ellipse,minimum height=30pt,minimum width=10pt,align=center,lightcolor2] (v2) at (3.5,0) {};
  \node[ellipse,minimum height=30pt,minimum width=10pt,align=center,lightcolor3] (v3) at (4.5,0) {};
  \node[ellipse,minimum height=30pt,minimum width=10pt,align=center,lightcolor4] (v4) at (5.5,0) {};
  \node[ellipse,minimum height=30pt,minimum width=10pt,align=center,lightcolor5] (v5) at (6.5,0) {};
  \node[ellipse,minimum height=30pt,minimum width=10pt,align=center,lightcolor6] (v6) at (2,1.2) {};
  \node[ellipse,minimum height=30pt,minimum width=10pt,align=center,lightcolor7] (v7) at (3,1.2) {};
  \node[ellipse,minimum height=30pt,minimum width=10pt,align=center,lightcolor8] (v8) at (4,1.2) {};
  \node[ellipse,minimum height=30pt,minimum width=10pt,align=center,lightcolor9] (v9) at (5,1.2) {};
  \node[ellipse,minimum height=30pt,minimum width=10pt,align=center,lightcolor10] (v10) at (6,1.2) {};
  \node[ellipse,minimum height=30pt,minimum width=10pt,align=center,lightcolor11] (v11) at (2,-1.2) {};
  \node[ellipse,minimum height=30pt,minimum width=10pt,align=center,lightcolor12] (v12) at (3,-1.2) {};
  \node[ellipse,minimum height=30pt,minimum width=10pt,align=center,lightcolor13] (v13) at (4,-1.2) {};
  \node[ellipse,minimum height=30pt,minimum width=10pt,align=center,lightcolor14] (v14) at (5,-1.2) {};
  \node[ellipse,minimum height=30pt,minimum width=10pt,align=center,lightcolor15] (v15) at (6,-1.2) {};
  
  \scoped[on background layer]{
  \foreach \i/\j in {1/2,2/3,3/4,6/7,7/8,3/8,3/9,4/15,5/15,9/10,11/12,12/13,13/14,14/15}{
   \draw[thickedge] (v\i.center) edge (v\j.center);
  }
  }
    
  \node[emptyvertex,fill=dandelion!80] (w1) at (0.4,0) {};
  \node[emptyvertex,fill=purple!80] (w2) at (0.2,1.2) {};
  \node[emptyvertex,fill=darkspringgreen!80] (w3) at (0.2,-1.2) {};
  
  \foreach \s in {-2,-1,0,1,2}{
   \node[tinyvertex,color1] (a\s) at ($(2.5,0)+(0,0.2*\s)$) {};
   \draw (w1) edge (a\s);
   \node[tinyvertex,color6] (b\s) at ($(2,1.2)+(0,0.2*\s)$) {};
   \draw (w2) edge (b\s);
   \node[tinyvertex,color4a] (c\s) at ($(2,-1.2)+(0,0.2*\s)$) {};
   \draw (w3) edge (c\s);
  }
  
  \foreach \s in {-2,-1,0,1,2}{
   \node[tinyvertex,color3] (a\s) at ($(4.5,0)+(0,0.2*\s)$) {};
  }
 \end{tikzpicture}
 \caption{Visualization of the spanning tree $T$ described in the proof of Lemma \ref{la:extend-t-cr-bounded-scope}. Contracting the gray regions into single vertices gives a subgraph isomorphic to $K_{3,h}$.}
 \label{fig:genus-to-cr-bounded}
\end{figure}

\begin{proof}
 Let $C \coloneqq \im(\chi)$, $C_1 \coloneqq \chi(V_1)$ and $C_2 \coloneqq \chi(V_2)$.
 Also define $H$ to be the graph with vertex set $V(H) \coloneqq C$ and edge set
 \[E(H) = \{c_1c_2 \mid \exists v_1 \in \chi^{-1}(c_1), v_2 \in \chi^{-1}(c_2) \colon v_1v_2 \in E(G)\}.\]
 Let $C' \subseteq C_2$ be the vertex set of a connected component of $H[C_2]$.
 Then $|N_H(C')| \geq 3$ since each $v \in V_1$ is a singleton with respect to $\chi$ and $G$ is $3$-connected.
 
 Now let $c_1,c_2,c_3 \in N_H(C')$ be distinct and also let $v_i \in \chi^{-1}(c_i)$ for $i \in [3]$.
 Also let $T$ be a spanning tree of $H[C' \cup \{c_1,c_2,c_3\}]$ such that $c_1,c_2,c_3 \in L(T)$ (see Figure \ref{fig:genus-to-cr-bounded}).
 Moreover, let $T'$ be the subtree of $T$ obtained from repeatedly removing all leaves $c \in C'$.
 Hence, $L(T') = \{c_1,c_2,c_3\}$.
 Then there is a unique color $c$ such that $\deg_H(c) = 3$.
 Also, for $i \in [3]$, define $C_i'$ to be the set of internal vertices on the unique path from $c_i$ to $c$ in the tree $T'$.
 
 Since $|\chi^{-1}(c_i)| = 1$ and $\chi$ is stable with respect to the Color Refinement algorithm it holds that
 \[G\left[\chi^{-1}(C_i' \cup \{c_i\})\right]\]
 is connected.
 Let $U_i \coloneqq \chi^{-1}(C_i' \cup \{c_i\})$, $i \in [3]$.
 Also let $U = \chi^{-1}(c)$ and suppose that $|U| \geq h$.
 Then $N(U_i) \cap U \neq \emptyset$ by the definition of the tree $T$.
 Moreover, this implies $U \subseteq N(U_i)$ since $\chi$ is stable with respect to the Color Refinement algorithm.
 Hence, $G$ contains a minor isomorphic to $K_{3,h}$.
\end{proof}

\begin{corollary}
 \label{cor:excluded-k3h-t-cr-bounded}
 Let $(G,\chi_V,\chi_E) \in \CC_h$ be a $3$-connected, vertex- and arc-colored graph and let $v_1,v_2,v_3 \in V(G)$.
 Also define $\chi_V^{*}(v_i) \coloneqq (i,1)$ for $i \in [3]$ and $\chi_V^{*}(v) \coloneqq (\chi_V(v),0)$ for all $v \in V(G) \setminus \{v_1,v_2,v_3\}$.
 Then $(G,\chi_V^{*},\chi_E)$ is $(h-1)$-CR-bounded.
\end{corollary}

\begin{proof}
 Let $\chi^{*} \coloneqq \tColRef{(h-1)}{(G,\chi_V^{*},\chi_E)}$.
 Suppose towards a contradiction that $\chi^{*}$ is not discrete (i.e., not every color class is a singleton).
 Let $V_2 \coloneqq \{v \in V(G) \mid |(\chi^*)^{-1}(\chi^{*}(v))| > 1\}$ and let $V_1 \coloneqq N_G(V_2)$.
 Then $|V_1| \geq 3$ since $|V(G) \setminus V_2| \geq 3$ and $G$ is $3$-connected.
 Also note that $\chi^{*}|_{V_1 \cup V_2}$ is a stable coloring for the graph $G[V_1 \cup V_2]$.
 Hence, by Lemma \ref{la:extend-t-cr-bounded-scope}, there is some color $c$ such that $1 < |(\chi^{*})^{-1}(c)| \leq h-1$.
 But this contradicts the definition of the coloring $\chi^{*}$ (cf.\ Definition \ref{def:t-cr-bounded}).
\end{proof}

The last corollary gives some insights into the structure of the automorphism group of graphs $G \in \CC_h$.

\begin{theorem}
 Let $G \in \CC_h$ be a $3$-connected graph and let $v_1,v_2,v_3 \in V(G)$ be distinct vertices.
 Then $(\Aut(G))_{(v_1,v_2,v_3)}$ is a $\mgamma_{h-1}$-group.
\end{theorem}

Also, the corollary can be used to design an isomorphism test for the class $\CC_h$.

\begin{corollary}
 \label{cor:isomorphism-for-class-without-kh3}
 The Graph Isomorphism Problem for the class $\CC_h$ can be solved in time $n^{\CO((\log h)^{c})}$ for some absolute constant $c$.
\end{corollary}

\begin{proof}
 Let $G_1,G_2 \in \CC_h$.
 By standard decomposition techniques it suffices to consider the case where $G_1$ and $G_2$ are (vertex- and arc-colored) $3$-connected graphs (see, e.g., \cite{HopcroftT71,HopcroftT72}).
 
 Suppose $G_1 = (V_1,E_1,\chi_V^{1},\chi_E^{1})$ and $G_2 = (V_2,E_2,\chi_V^{2},\chi_E^{2})$.
 Let $v_1,v_2,v_3 \in V(G_1)$ be three arbitrary vertices.
 For every $w_1,w_2,w_3 \in V(G_2)$ it is checked whether there is some isomorphism $\varphi\colon G_1 \cong G_2$ such that $\varphi(v_i) = w_i$ for all $i \in [3]$.
 Towards this end, define $\widehat{\chi}_V^{1}(v_i) = (i,1)$ for $i \in [3]$ and $\widehat{\chi}_V^{1}(v) = (\chi_V^{1}(v),0)$ for all $v \in V(G_1) \setminus \{v_1,v_2,v_3\}$.
 Similarly, define $\widehat{\chi}_V^{2}(w_i) = (i,1)$ for $i \in [3]$ and $\widehat{\chi}_V^{2}(w) = (\chi_V^{2}(w),0)$ for all $w \in V(G_2) \setminus \{w_1,w_2,w_3\}$.
 Hence, it needs to be checked whether $\widehat{G}_1 \cong \widehat{G}_2$ where $\widehat{G}_j \coloneqq (V_j,E_j,\widehat{\chi}_V^{j},\chi_E^{j})$, $j \in [2]$.
 By Corollary \ref{cor:excluded-k3h-t-cr-bounded} the graphs $\widehat{G}_1$ and $\widehat{G}_2$ are $(h-1)$-CR-bounded.
 Hence, isomorphism of the two graphs can be tested within the desired time by Corollary \ref{cor:isomorphism-t-cr-bounded}.
\end{proof}

\begin{corollary}
 The Graph Isomorphism Problem for graphs of genus at most $g$ can be solved in time $n^{\CO((\log g)^{c})}$ for some absolute constant $c$.
\end{corollary}

\begin{proof}
 This follows from Theorem \ref{thm:genus-complete-bipartite} and Corollary \ref{cor:isomorphism-for-class-without-kh3}.
\end{proof}

\section{Conclusion}

We provided a faster algorithm for the Hypergraph Isomorphism Problem for $\mgamma_d$-groups running in time $(n+m)^{\CO((\log d)^{c})}$ for some absolute constant $c$.
In particular, this gives the fastest available algorithm for testing isomorphism of hypergraphs running in time $(n+m)^{\CO((\log n)^{c})}$.
A natural open question is whether the dependence on the number of hyperedges can be further improved.
Indeed, the algorithm from Theorem \ref{thm:main} is significantly faster than the previous best algorithm \cite{SchweitzerW19} running in time $n^{\CO(d)}m^{\CO(1)}$ only for small numbers of hyperedges.
For large numbers of hyperedges $m = n^{\Omega(d)}$, our algorithm becomes slower than the one by Schweitzer and Wiebking \cite{SchweitzerW19}.
Can the Hypergraph Isomorphism Problem for $\mgamma_d$-groups be solved in time $n^{\CO((\log d)^{c})}m^{\CO(1)}$ for some absolute constant $c$?
Similarly, can isomorphism of hypergraphs be tested in time  $n^{\CO((\log n)^{c})}m^{\CO(1)}$?

On the application side, we obtained an algorithm testing isomorphism of graphs excluding $K_{3,h}$ as a minor in time $n^{\CO((\log h)^{c})}$.
In particular, this gives an isomorphism test for graphs of Euler genus at most $g$ running in time $n^{\CO((\log g)^{c})}$.
With the Hypergraph Isomorphism Problem for $\mgamma_d$-groups being exploited as a subroutine in a number of algorithms testing isomorphism, it seems plausible to hope for further applications beyond the ones presented in this paper.
Indeed, in two follow-up works \cite{GroheNW20,Neuen22}, we provide isomorphism tests running in time $n^{\CO((\log h)^{c})}$ for $n$-vertex graphs excluding an arbitrary $h$-vertex graph as a minor or even only as a topological subgraph.
These algorithms again crucially build on the isomorphism test for hypergraphs as well as the notion of $t$-CR-bounded graphs.
Can we find further applications in this direction?

\bibliographystyle{plainurl}
\small
\bibliography{literature}

\end{document}